\documentclass[11pt]{amsart}

\usepackage{hhline}
\usepackage{graphics}
\usepackage{epstopdf}
\usepackage{a4wide,amsmath,amssymb}
\usepackage{amsfonts}
\usepackage{amsopn}
\usepackage{amsthm}
\usepackage{epsfig}
\usepackage{lscape}
\usepackage{multirow}
\usepackage{url}
\usepackage{tikz}
\usepackage{mathrsfs}
\usepackage[pdftex,bookmarks=true]{hyperref}
\usepackage{pgf,tikz}
\usetikzlibrary{arrows}
\usepackage{verbatim}
\usepackage{footnote}
\allowdisplaybreaks

\usepackage{enumitem,booktabs}
\usepackage[referable]{threeparttablex}
\renewlist{tablenotes}{enumerate}{1}
\makeatletter
\setlist[tablenotes]{label=\tnote{\alph*},ref=\alph*,itemsep=\z@,topsep=\z@skip,partopsep=\z@skip,parsep=\z@,itemindent=\z@,labelindent=\tabcolsep,labelsep=.2em,leftmargin=*,align=left,before={\footnotesize}}
\makeatother

\newcounter{NN}
\newtheorem{observation}[NN]{Observation}
\newtheorem{remarkk}[NN]{Remark}
\newtheorem{proposition}[NN]{Proposition}
\newtheorem{theorem}[NN]{Theorem}
\newtheorem{corollary}[NN]{Corollary}
\newtheorem{conjecture}[NN]{Conjecture}

\def\mod#1{\,({\rm mod\ }#1)}

\DeclareMathOperator{\lcm}{lcm}
\newcommand{\ts}{\hspace{0.5pt}}

\newcommand{\QQ}{\mathbb{Q}\ts}

\newcommand{\NN}{\mathbb{N}\ts}
\newcommand{\ZZ}{\mathbb{Z}\ts}

\usetikzlibrary{calc}

\newcommand{\tikzmark}[1]{\tikz[overlay,remember picture] \node (#1) {};}
\newcommand{\DrawBox}[1][]{%
    \tikz[overlay,remember picture]{
    \draw[red,#1]
      ($(left)+(-0.2em,0.9em)$) rectangle
      ($(right)+(0.2em,-0.3em)$);}
}

\newcommand*\circled[1]{\tikz[baseline=(char.base)]{\node[shape=circle,draw,inner sep=2pt] (char) {#1};}}

\def\N{\mathbb{N}}
\def\Z{\mathbb{Z}}

\def\Q{\mathbb{Q}}
\def\N{\mathbb{N}}

\begin{document}
\title[Vanishing algebraic entropy]
{Algebraic entropy of (integrable) lattice equations and their reductions}
\author{John A.~G.~Roberts, Dinh T.~Tran}
\address{School of Mathematics and Statistics,
University of New South Wales, Sydney NSW 2052, Australia
}

\email{jag.roberts@unsw.edu.au, t.d.tran@unsw.edu.au}

\begin{abstract}
We study the growth of degrees in many autonomous and non-autonomous lattice equations defined by quad rules with corner boundary values,
some of which are known to be integrable by other characterisations.
Subject to an {\em enabling} conjecture, we prove polynomial growth for a  large class of equations which includes the Adler-Bobenko-Suris
equations and Viallet's $Q_V$ and its non-autonomous generalization.
Our technique is to determine the ambient degree growth of the projective version of the lattice equations and to conjecture the growth of their common factors at each lattice vertex,
allowing the true degree growth to be found.
The resulting degrees satisfy a linear partial difference equation which is {\em universal}, i.e. the same for all the integrable lattice equations considered.
When we take periodic reductions of these equations, which includes {\em staircase} initial conditions,
we obtain from this linear partial difference equation an ordinary difference equation for degrees that implies quadratic or linear degree growth.
We also study growth of degree of several non-integrable lattice equations. Exponential growth of degrees of these equations, and their mapping reductions, is also proved
subject to a conjecture.
\end{abstract}

\maketitle

\section{Introduction}
In recent years, there has been growing interest in discrete dynamical systems (e.g. partial or ordinary difference equations)
that are {\em integrable} -- see \cite{HJN} and references therein. Similar to their continuous time and space counterparts,
the term {\em integrable} in discrete dynamical systems refers to possession of one or more signature properties (e.g. singularity confinement,
Lax pairs, consistency around a cube) that
imply  `low complexity' of the dynamics relative to the generic behaviour expected.
As an example, when we iterate a rational map $\phi$ in $n$ dimensions, we typically expect the maximum of the degrees of numerator and denominator to
grow exponentially with the number of iterations. If we work projectively in $n+1$ dimensions, the map transforms to a polynomial map of the
projective space, with components that are homogeneous of degree $d$.
If we let $d_k$ be the degree of $\phi^k$, we can define the {\em algebraic entropy} to be \cite{BV}
\begin{equation} \label{eq:FirstEntropy}
\epsilon:=\lim_{k\rightarrow \infty} \frac{1}{k}\log(d_k).
\end{equation}
Typically, $\epsilon > 0$.
However, in some special cases,  systematic cancellations may occur between numerator and denominator
that may lead to slow (i.e. subexponential) growth whence $\epsilon=0$, the case of {\em vanishing entropy}.
This property in itself may be taken as a definition of integrability
and it has received much recent attention, also in connection to other related integrability properties/detectors
(\cite{BV,Halburd09HowtoDetect,Halburd2005Diophantine-Int, HasselPropp,  Hietarinta1998Singularity-Con, coprimeness, extended_Hietarinta, RGWM16, RobertsTran14, Kamp_growth,  Viallet2015}).

In this paper, we investigate algebraic entropy of partial difference equations in the plane,
or lattice equations.  Consider a square lattice whose sites have coordinates $(l,m) \in \ZZ^2$.
Field variables $u$ are defined on each site of the lattice (see Figure 1 of Section 2) and we assume that on each lattice square
there is an equation which is multi-affine in the variables, i.e. affine in each variable (e.g. see equation \eqref{eq:simplequad} of Section 2).
This equation allows to solve for the field variable on any vertex of the square as a rational function of the field variables on
the other three vertices.
Coefficients in this equation might depend on the lattice site $(l,m)$, in which case
the equation is termed {\em non-autonomous}  and {\em autonomous} otherwise.
Historically, it has often been the case that the former has been obained from the latter for a class of equations via a
deautonomizing process that maintains the integrability signature of singular confinement
~\cite{Halburd09HowtoDetect, Halburd2005Diophantine-Int, Hietarinta1998Singularity-Con}.
This includes non-autonomous versions of the Adler-Bobenko-Suris (ABS) equations,
modified Korteweg-De Vries (mKdV) and sine-Gordon (sG) equations \cite{Halburd09HowtoDetect, non_ABS_GR, SC,SRH}.

The extension of algebraic entropy to lattice equations was begun in \cite{Viallet}.
In the main, lattice equations have been studied by specifying  initial conditions along a staircase (cf. Figure \ref{F:staircase}).
These initial condtions are taken to be affine in an indeterminate variable and the lattice rule allows an evolution of parallel staircases
to the initial one (e.g. upwards in Figure \ref{F:staircase}). Denoting by $d_k$ the degree of the polynomial in the indeterminate on this evolving sequence
of staircases, Viallet proposed to use the same definition \eqref{eq:FirstEntropy} for the algebraic entropy of the lattice equation.
Numerous numerical investigations support the idea that, when not exponential, $d_k$ grows quadratically (or linearly in special cases).
Again, the lattice equations where this subexponential growth occurs typically have other properties associated with integrability
(e.g. singularity confinement \cite{non_ABS_GR}).

Our approach here is to consider directly the degree growth of the partial difference equation without taking a one-dimensional reduction
of the degrees. Instead of boundary values on a staircase, we take corner boundary values along the boundary of the first quadrant
with the boundary values at most affine in an indeterminate. Using the lattice equation, we can calculate the field variables inside
the quadrant starting from the $(1,1)$ vertex. We record the degree in the indeterminate of the numerator (or denominator) of the
field variable at each site, denoted $\overline{d}_{l,m}$. In practice -- see Section 2 below -- we write each field variable in projective coordinates, that is
we introduce $u_{l,m}=\frac{x_{l,m}}{z_{l,m}}$ and follow the evolution of $x_{l,m}$ and $z_{l,m}$ via a projective version of the lattice
rule.  Then $\overline{d}_{l,m}=d_{l,m} - g_{l,m}$, where $d_{l,m}$ is the degree of $x_{l,m}$ and $z_{l,m}$ and
$g_{l,m}$ is the degree of their greatest common divisor (which is cancelled in numerator and denominator of the associated field variable to given
the reduced degree $\overline{d}_{l,m}$).  Generically, one expects no common factor to $x_{l,m}$ and $z_{l,m}$ and the `magic'
is that for certain lattice equations, the common factors grow fast enough that $\overline{d}_{l,m}$ grows subexponentially.

We summarise our main results:

\vspace{0.5cm}
 1. One generically has for any multi-affine rule \eqref{E1:Eq} that
\begin{equation} \label{E:samedeg}
d_{l+1,m+1}=d_{l,m}+d_{l+1,m}+d_{l,m+1},
\end{equation}
which corresponds to exponential degree growth -- see Section 3.
However, some lattice rules distinguish themselves by a systematic factorization
of $x_{l,m}(w)$ and $z_{l,m}(w)$  and the appearance of common factors.
This can lead to sub-exponential growth of $\bar{d}_{l,m}$ if there is sufficient exponential growth of $g_{l,m}$.

 2. For a large class of lattice equations -- see Table 1 above the double-line division and Table 2 -- we find
that for the projective version of these rules iterated on any $2\times 2$ lattice block, the iterates $x_{l+1,m+1}$ and $z_{l+1,m+1}$ at the top
right corner share a common factor.
We denote this factor by
$A_{l+1,m+1}$ and find it depends only on the variables in the $2\times 2$ lattice block at the sites $(l-1, m)$ and $(l,m-1)$.
The factor is homogeneous of degree $4$ in these variables with
\begin{equation} \label{goodA}
\deg(A_{l+1,m+1}(w))= 2 \, (d_{l-1,m} + d_{l,m-1}).
\end{equation}
This local factorization provides the `engine' that leads to continuous factor generation and sub-exponential growth.

Suppose the values $G^{'}_{l,m}(w)$ are generated from the following recurrence involving the aforementioned $A_{l+1,m+1}$:
\begin{equation} \label{E:Gcd_recurenceintro}
G^{'}_{l+1,m+1}=\frac{G^{'}_{l-1,m-1}\ G^{'}_{l+1,m}\ G^{'}_{l,m+1}\ A_{l+1,m+1}}{G^{'}_{l-1,m}\ G^{'}_{l,m-1}}.
\end{equation}
and the boundary values in the first quadrant are chosen to agree with those of our lattice equation's ${\gcd}_{l,m}(w)$ when $0 \le l \le 1$ or $0 \le m \le 1$.
Then we {\em conjecture} for a class of lattice equations specified in Table 2
that this remains true on iteration of both \eqref{E:Gcd_recurenceintro} and the corresponding lattice equation.
Consequently,
$$g_{l,m}=\deg({\gcd}_{l,m}(w))=\deg(G^{'}_{l,m}(w)).$$

3. If the previous conjecture is true for the considered lattice equations,
the reduced degree sequence $\overline{d}_{l,m}$ satisfies the
the following linear partial difference equation with constant coefficients:
\begin{equation} \label{E:degreeRelationintro}
\overline{d}_{l+1,m+1}=\overline{d}_{l+1,m}+\overline{d}_{l,m+1}+\overline{d}_{l-1,m-1}-\overline{d}_{l,m-1}-\overline{d}_{l-1,m}.
\end{equation}
If the corner boundary degrees $(\overline{d}_{1,m_0+1}-\overline{d}_{0,m_0})$ and
$(\overline{d}_{l_0+1,1}-\overline{d}_{l_0,0})$ are bounded polynomially of degree $D$ in their respective coordinates $m_0$ and $l_0$, then
the sequence $\overline{d}_{l,m}$ of \eqref{E:degreeRelation}
is bounded polynomially of degree $D+1$ along each lattice diagonal.
Therefore $\lim_{l \to \infty} \log(\overline{d}_{l,m_0+l})/l = 0$ for fixed $m_0\ge 0$, i.e. zero algebraic entropy.

4. The use of periodic staircases is a standard way to derive integrable maps from integrable lattice equations.
A consequence of \eqref{E:degreeRelationintro} is that it automatically implies quadratic degree growth (at most) for these staircase reductions, as has been found experimentally.
Thus, we regard \eqref{E:degreeRelationintro} as a master equation to describe the vanishing lattice algebraic entropy.

 5. We also provide examples of lattice maps that also develop a factor $A_{l+1,m+1}$ over a $2\times 2$ lattice block, but of lesser degree
than \eqref{goodA}.  We show explicitly that, subject to a postulated recurrence for ${\gcd}_{l,m}(w)$ similar to \eqref{E:Gcd_recurenceintro}, their
degree growth is exponential on the lattice and in reduction.
\vspace{0.5cm}

This paper is organized as follows.
 In section 2, we give a setting to measure the complexity of  a certain
 lattice equation. Using the equation, one can easily write down each vertex in projective coordinates and obtain
  the corresponding rules in projective coordinates. In this section, we also give   a list of integrable lattice
   equations  that we consider in this paper.
   In the next section, we explore growth
 of degrees over $\QQ[w]$.
    Initial values are given as polynomials in $w$ on the horizontal and vertical axes of the $p\times p$ square in the first
 quadrant.
  We calculate an upper bound for degrees of multi-linear lattice equations. In section 4,  we
 present a conjecture that seems to hold for equations in the Adler-Bobenko-Suris (ABS) list and their non autonomous versions whose parameters are given in strips and $Q_V$, non-autonomous $Q_V$ and other equations in the Hietarinta-Viallet list.
 Based on the factorization at the point $(2,2)$, we give a recursive  formula
 to build a common divisor of the the numerator and denominator of each vertex.
 This gives an  upper bound for  the actual degree of each vertex.
 By considering two types of initial values, we are able to show the polynomial growth of these equations.  In particular, using initial
 values as ratios of  degree-one-polynomials, we can explain the quadratic growth of the equations which was given numerically by Viallet \cite{LGS, Viallet}.
 We also study  growth of degrees of the discrete  Korteweg–de Vries  (KdV) equation \cite{Nijhoff1995KdV}, Lotka-Volterra (LV) equation~\cite{LeviYamilov2009}, modified-modified Volterra (mLV) equation~\cite{LeviYamilov2}, Tzitzeica (Tz) equation~\cite{Adler} and the equation given by Mikhailov and Xentinidis (MX equation)~\cite{Pavlos}. We will show that even these equations do not
 have any common factor at the point $(2,2)$, one can bring them to the previous cases by shifting the starting point.
 In section 5, we present a similar conjecture for some non-integrable equations in the paper~\cite{HietVial}. This conjecture then helps us to prove that these equations have exponential growth. In section 6, we  apply results given in section 4 and section 5 to staircase boundary conditions. Especially, when the staircase boundary values are periodic, we obtain results for growth of degrees of mappings.

 \section{The setting and some preliminary results}
In this section, we set up a methodology to measure the complexity of certain lattice equations upon iteration,
following aspects of the foundational papers \cite{HietVial, Tremblay, Viallet}. The methodology is valid for both autonomous and non-autonomous equations.
We are given a square lattice (see Figure \ref{C1F:quadgraph}) whose sites have coordinates $(l,m) \in {\Z}^2$.
Field variables $u$ are defined on each lattice site and are related via
an equation $Q$ which is multi-affine in these variables, i.e. affine in each variable:
\begin{equation} \label{E1:Eq}
Q_{l,m}(u_{l,m}, u_{l+1,m},u_{l,m+1},u_{l+1,m+1})=0.
\end{equation}
\begin{figure}[h]
\begin{center}
\setlength{\unitlength}{6pt}
\begin{picture}(16,16)(2,2)
\put(0,1){$u_{l,m}$}
\put(3,3){\line(1,0){12}}
\put(3,3){\circle*{1}}
\put(3,3){\line(0,1){12}}
\put(0,17){$u_{l,m+1}$}
\put(15,3){\line(0,1){12}}
\put(15,3){\circle*{1}}
\put(12,1){$u_{l+1,m}$}
\put(3,15){\line(1,0){12}}
\put(3,15){\circle*{1}}
\put(15,15){\circle*{1}}
\put(12,17){$u_{l+1,m+1}$}
 \end{picture}
 \end{center}
 \caption{An elementary square of the integer lattice with field variables on the vertices.}\label{C1F:quadgraph}
 \end{figure}
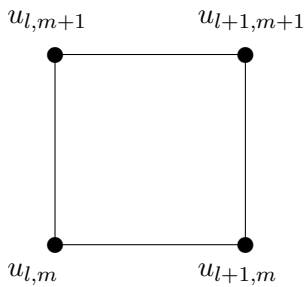
Coefficients in this equation may depend most generally on the lattice site $(l,m)$, called the {\em non-autonomous} case,
accounting for the $(l,m)$ subscript on $Q$. In the {\em autonomous} case, the rule $Q$ is the same on each lattice square.
Generically,  one can solve \eqref{E1:Eq} for $u_{l+1,m+1}$ and obtain
 \begin{equation} \label{E:utopright2}
u_{l+1,m+1}=\frac{h^{(2)}_{l,m}\left(u_{l,m},u_{l+1,m},u_{l,m+1}\right)}{h^{(1)}_{l,m}\left(u_{l,m},u_{l+1,m},u_{l,m+1}\right)},
\end{equation}
where $h^{(1)}$ and $h^{(2)}$ are multi-affine in their arguments with no common factor. We {\bf assume}
that each of the three arguments appears in at least one of  $h^{(1)}$ and $h^{(2)}$
(so, in particular, both  $h^{(1)}$ and $h^{(2)}$ are non-zero and not both constant).
Now we write each field variable in projective coordinates, that is
we introduce $u_{l,m}=\frac{x_{l,m}}{z_{l,m}}$ etc.
We obtain
the projective version of \eqref{E:utopright2}
\begin{eqnarray}
x_{l+1,m+1}&=& f^{(1)}_{l,m}\left(x_{l,m},x_{l+1,m},x_{l,m+1},z_{l,m},z_{l+1,m},z_{l,m+1}\right), \label{projx} \\
z_{l+1,m+1}&=& f^{(2)}_{l,m}\left(x_{l,m},x_{l+1,m},x_{l,m+1},z_{l,m},z_{l+1,m},z_{l,m+1}\right), \label{projz}
\end{eqnarray}
where $f^{(1)}_{l,m}$ and $g^{(1)}_{l,m}$ are homogeneous polynomials of degree $3$.
For example, for the non-autonomous version of the lattice rule $H_1$ given by
\begin{equation} \label{eq:simplequad}
(u_{l,m}-u_{l+1,m+1})(u_{l+1,m}-u_{l,m+1})+\beta_m - \alpha_l = 0,
\end{equation}
we obtain its projective version as
\begin{eqnarray}
x_{l+1,m+1}&=&-x_{l, m}x_{l+1, m}z_{l, m+1}+x_{l, m}x_{l, m+1}z_{l+1, m}+(\alpha_l-\beta_m)\ z_{l, m}z_{l+1, m}z_{l, m+1}, \label{H1projx} \\
z_{l+1,m+1}&=& \left(-x_{l+1, m}z_{l, m+1}+x_{l, m+1}z_{l+1, m}\right)z_{l, m}. \label{H1projz}
\end{eqnarray}
This example encapsulates the general property:
\begin{observation}
\label{O:PCoor}
The projective coordinates $x_{l+1,m+1}$ and $z_{l+1,m+1}$ at the top right vertex of Figure \ref{C1F:quadgraph} are homogeneous polynomials (\ref{projx}) and (\ref{projz}) of
degree $3$, where each term of these polynomials
includes exactly one projective coordinate from
each of the remaining three vertices of the square.
\end{observation}

Table \ref{tab:rules} is the list of non-autonomous and autonomous lattice equations that we will consider in this paper.
For ease of notation in listing them, we denote $u:=u_{l,m},u_1:=u_{l+1,m}, u_2:=u_{l,m+1}$ and $u_{12}:=u_{l+1,m+1}$.
The equations satisfy the assumption  of the form \eqref{E:utopright2} above.

In this table, we use the name  $Q_i$ and $H_i$, $(i=1,2,3)$
for equations given by Adler-Bobenko-Suris (ABS)~\cite{ABS} with parameters given on strips~\cite{SRH, YingShi}, i.e. lattice parameters $\alpha$ depends on $l$ and $\beta$ depend on $m$ only.
Also given are other
well-known discrete  integrable systems such as  KdV, MX, LV, mLV and Tz equations \cite{Adler, ABS, Hirota1995Lotka-Voltera, LeviYamilov2009, LeviYamilov2, Pavlos}.
The non-autonomous mKdV and sG were introduced in \cite{SC}. The non-autonomous KdV equation was studied in~\cite{Kenji}.
In particular, the autonomous KdV equation can be obtained from two copies of autonomous $H_1$ \cite{Nijhoff1995KdV}. The LV equation
can be obtained from the autonomous KdV equation via a Miura transformation \cite{Ramaninotsonew}.
The rest of the list in Table 1 denoted by E$\#\#$ comprises lattice equations introduced by Hietarinta and Viallet~\cite{HietVial}. Their degree growth was studied
using a staircase initial condition.
Among these equations, equations E21 and E22 can be reduced to autonomous versions of $H_1$ and $H_3$ in the ABS list.
Equations above the double line division in Table 1 are known (or suggested) to be integrable via one or more
standard properties (consistency, Lax pair, singularity confinement, numerically-observed low degree growth using a staircase initial condition).\
Equations below the double line were shown in \cite{HietVial} to have exponential degree growth using a staircase initial condition.
We will study (or restudy) all equations in Table 1 with corner boundary values.

In addition to Table 1, we also consider a non-autonomous version of the lattice equation $Q_V$ \cite{Q5} that recently appeared in \cite{LGS}:
\begin{equation} \label{NonQV}
\small
\begin{split}
Q_V^{non}:=&p_{{1,0}}uu_{{1}}u_{{2}}u_{{12}}+p_{7,0}+
 \left( p_{{2,0}}- \left( -1 \right)
^{l}p_{{2,1}}- \left( -1 \right) ^{m}p_{{2,2}}+ \left( -1 \right) ^{l+m}p_{{2,3}} \right) uu_{{2}}u_{{12}}+
\\
&\
 \left( p_{{2,0}}+ \left( -1\right) ^{l}p_{{2,1}}- \left( -1 \right) ^{m}p_{{2,2}}- \left( -1
 \right) ^{l+m}p_{{2,3}} \right) u_{{1}}u_{{2}}u_{{12}}+
\\
&\
\left( p_{{2,0
}}+ \left( -1 \right) ^{l}p_{{2,1}}+ \left( -1 \right) ^{m}p_{{2,2}}+
 \left( -1 \right) ^{l+m}p_{{2,3}} \right) uu_{{1}}u_{{12}}+
\\
&\
\left( p_{{2,0}}- \left( -1 \right) ^{l}p_{{2,1}}+ \left( -1 \right) ^{m}p_{{2,
2}}- \left( -1 \right) ^{l+m}p_{{2,3}} \right) uu_{{2}}u_{{1}}+
\\
&\
 \left( p_{{3,0}}- \left( -1 \right) ^{m}p_{{3,2}} \right) uu_{{1}}+
 \left( p_{{3,0}}+ \left( -1 \right) ^{m}p_{{3,2}} \right) u_{{2}}u_{{12}}+
 \\
&\
\left( p_{{4,0}}- \left( -1 \right) ^{l+m}p_{{4,3}} \right) uu_{
{12}}+
 \left( p_{{4,0}}+ \left( -1 \right) ^{l+m}p_{{4,3}} \right) u_{{1}}u_{{2}}+
 \\
&\
\left( p_{{5,0}}- \left( -1 \right) ^{l}p_{{5,1}}
 \right) u_{{1}}u_{{12}}+ \left( p_{{5,0}}+ \left( -1 \right) ^{l}p_{{5,1}} \right) uu_{{2}}+
\\
&\
 \left( p_{{6,0}}+ \left( -1 \right) ^{l}p_{{6,
1}}- \left( -1 \right) ^{m}p_{{6,2}}- \left( -1 \right) ^{l+m}p_{{6,3}} \right) u+
\\
&\
\left( p_{{6,0}}- \left( -1 \right) ^{l}p_{{6,1}}-
 \left( -1 \right) ^{m}p_{{6,2}}+ \left( -1 \right) ^{l+m}p_{{6,3}}\right) u_{{1}}+
\\
&\
 \left( p_{{6,0}}+ \left( -1 \right) ^{l}p_{{6,1}}+
 \left( -1 \right) ^{m}p_{{6,2}}+ \left( -1 \right) ^{l+m}p_{{6,3}}
 \right) u_{{2}}+
\\
&\
\left( p_{{6,0}}- \left( -1 \right) ^{l}p_{{6,1}}+
 \left( -1 \right) ^{m}p_{{6,2}}- \left( -1 \right) ^{l+m}p_{{6,3}} \right) u_{{12}}=0.
\end{split}
\end{equation}
It has been checked numerically that this equation has polynomial growth and it follows the same degree pattern as seen originally in $Q_V$ \cite{LGS, Q5}. Moreover,
$Q_V$ can be obtained from $Q_V^{non}$ by setting all coefficients that depend on $l$ and $m$ to be $0$.  Recall that the significance of $Q_V$ is that all the equations
in the ABS list can be obtained from it by choosing the parameters appropriately.

\begin{table}[h]
\centerline{
\begin{threeparttable}
\begin{tabular}{|l p{13cm} r|}
\hline     {\em Name} &
        {\em Lattice equation} &
        {\em Reference}
\\ \hline \hline
$Q_1$ & $\alpha_l(u-u_{2})(u_1-u_{12})-\beta_m(u-u_1)(u_2-u_{12})+\delta^2 \alpha_l\beta_m(\alpha_l-\beta_m)=0$ & \cite{ABS}
\\ \hline
$Q_2$ &  $ \alpha_l(u-u_{2})(u_1-u_{12})-\beta_m(u-u_1)(u_2-u_{12})+\alpha_l\beta_m(\alpha_l-\beta_m)(u+u_1+u_2+u_{12})
 -\alpha_l \beta_m(\alpha_l-\beta_m)(\alpha_l^2-\alpha_l \beta_m+\beta_m^2)=0 $ & \cite{ABS}
 \\
 \hline
 $Q_3$ & $(\beta_m^2-\alpha_l^2)(uu_{12}+u_1u_2)+\beta_m(\alpha_l^2-1)(uu_1+u_2u_{12})-\alpha_l(\beta_m^2-1)(uu_2+u_1u_{12})
-\delta^2(\alpha_l^2-\beta_m^2)(\alpha_l^2-1)(\beta_m^2-1)/(4\alpha_l\beta_m)=0$ & \cite{ABS}
 \\
 \hline
 $Q_4$& $\alpha_l(uu_1+u_2u_{12})-\beta_m(uu_2+u_1u_{12})-(\alpha_lB_m-\beta_mA_l)\big(uu_{12}+u_1u_2-\alpha_l\beta_m(1+uu_1u_2u_{12})\big)/(1-\alpha_l^2\beta_m^2)=0$
 &\cite{ABS}\tnotex{tnote:Q4}
  \\
 \hline
 $H_1$ & $(u-u_{12})(u_1-u_2)+\beta_m-\alpha_l=0 $
& \cite{ABS,Nijhoff1995KdV}
 \\
 \hline
$H_2$ & $(u-u_{12})(u_1-u_2)+(\alpha_l-\beta_m)(u+u_1+u_2+u_{12})+\beta_m^2-\alpha_l^2=0 $& \cite{ABS}
 \\
 \hline
$H_3$ & $\alpha_l(uu_1+u_2u_{12})-\beta_m(uu_2+u_1u_{12})+\delta(\alpha_l^2-\beta_m^2)=0$ & \cite{ABS}
 \\
 \hline
 mKdV&$p_{l,m}uu_1-q_{l,m}uu_2+r_{l,m}u_1u_{12}-u_2u_{12}=0$&\cite{SC}\tnotex{tnote:mKdV}
 \\
 \hline
 sG&$ r_{l,m}uu_1u_2u_{12}+p_{l,m}u_1u_2-uu_{12}-q_{l,m}=0$&\cite{SC}\tnotex{tnote:sG}
 \\
 \hline
 E16& $uu_{{1}}p_{{1}}+uu_{{2}}p_{{5}} ( p_{{1}}p_{{3}}+p_{{2}} )
+ ( uu_{{12}}+u_{{1}}u_{{2}} ) p_{{2}}+u_{{1}}u_{{12}}p_{{6
}}+u_{{2}}u_{{12}}p_{{3}} ( p_{{5}}p_{{6}}-p_{{2}} )=0
$
& \cite{HietVial}
\\
\hline
E21& $\left( uu_{{2}}+u_{{1}}u_{{12}} \right) p_{{5}}+ \left( u_{{1}}u+u_{{
2}}u_{{12}} \right) p_{{3}}+s=0
$
&\cite{HietVial}
\\
\hline
E22 & $ \left( u-u_{{12}} \right)  \left( u_{{1}}-u_{{2}} \right) +r_{{1}}
 \left( u+u_{{1}}+u_{{2}}+u_{{12}} \right) +s=0
$
&\cite{HietVial}
\\
\hline
E24& $ ( uu_{{2}}+u_{{1}}u_{{12}} )  ( p_{{2}}+p_{{3}}) + ( u_{{1}}u+u_{{2}}u_{{12}}) p_{{3}}+uu_{{12}}+u
_{{1}}u_{{2}}+p_{{3}}p_{{2}} ( p_{{2}}+p_{{3}} ) {s}^{2}=0
$
& \cite{HietVial}
\\
\hline
E25& $uu_{{12}}+u_{{1}}u_{{2}}+ ( u_{{1}}u+u_{{2}}u_{{12}} ) p_{{
3}}- ( uu_{{2}}+u_{{1}}u_{{12}} )  ( p_{{3}}+1) + ( u_{{12}}-u ) r_{{4}}+ (u_{{1}}-u_{{2}}) r_{{2}}
- \left( s ( p_{{3}}+1) +r_{{4}} \right)
 ( sp_{{3}}+r_{{4}}) +sr_{{2}}=0
$ &\cite{HietVial}
\\
\hline
KdV & $ (\alpha_l-\beta_m)u-(\alpha_{l+1}-\beta_{m+1})u_{12}+(\alpha_{l+1}+\beta_m)/u_1-(\alpha_l+\beta_{m+1})/u_2=0 $ & \cite{Kenji, Nijhoff1995KdV}
 \\
\hline
MX & $ (u+u_{12})u_1u_2+1=0 $ & \cite{Pavlos}
\\
\hline
Tz&$uu_{12}(c^{-1}u_1u_2-u_1-u)+u_{12}+u-c=0$&\cite{Adler}\\
\hline
LV&$(u_1+\beta_m)(u-\beta_m)-(u_{12}-\beta_{m+1})(u_2+\beta_{m+1})=0$&\cite{LeviYamilov2009}
\\
\hline
mLV&$(1+uu_1)(cu_{12}+c^{-1}u_2)-(1+u_2u_{12})(cu+c^{-1}u_1)=0$&\cite{LeviYamilov2}
\\
\hline
\hline
E20&
$u_{{12}}u_{{1}}p_{{6}}+u_{{12}}u_{{2}}p_{{3}}+u_{{1}}up_{{1}}+u_{{2}}up_{{5}}+u_{{12}}p_{{3}}
p_{{6}}r_{{4}}
+u_{{1}}p_{{6}}r_{{2}}+u_{{2}}p_{{3}}r_{{3}}+ u\big( -p_{{1}}p_{{5}}r_{{4}}+p_{{1}}r_{{3}}+p_{{5}}r_{{2}}
 \big) +s=0
$
&\cite{HietVial}
\\
\hline
E26&
$u_{{12}}u_{{1}}{p_{{3}}}^{2}+u_{{2}}u{p_{{6}}}^{2}+u_{{12}}u_{{2}}p_{{1}}p_{{3}}+ {u_{{1}}up_{{3}}{p_{{6}}}^{2}}{p_{{1}}^{-1}}
 + \big( u_{{12}}p_{{3}}+u_{{1}}p_{{3}}+u_{{2}}p_{{6}}+up_{{6}} \big) r_{{1}}+{r_{{1}}}^{2}=0$
&\cite{HietVial}
\\
\hline
E27&
$u_{{12}}u_{{1}}{p_{{6}}}^{2}+u_{{2}}u{p_{{3}}}^{2}+{u_{{2}}u_{{12}}p_{{6}} ( p_{{3}}-1) p
_{{1}}^{-1}}+u_{{1}}up_{{1}}p_{{6}} ( p_{{3}}-1 )
  +
 ( uu_{{12}}+u_{{1}}u_{{l2}} ) p_{{6}}+
 ( u_{{12}}p_{{6}}+u_{{1}}p_{{6}}+u_{{2}}p_{{3}}+up_{{3}}) r_{{4}}+{r_{{4}}}^{2}=0$
&\cite{HietVial}
\\
\hline
E28&
$uu_{{12}}+u_{{1}}u_{{2}}+p_{{3}} ( u_{{2}}u_{{12}}+uu_{{1}}) +u_{{12}}u_{{1}}
p_{{6}}
+ {uu_{{2}} ( p_{{3}}-1) ^{2}}{p_
{{6}}^{-1}}+r_{{3}} ( p_{{6}}-p_{{3}}+1)  ( u_{{1}}+u) +p_{{6}} ( p_{{6}}+1) {r_{{3}}}^{2}=0$
&\cite{HietVial}
\\
\hline
E30&
$ uu_{{1}}+u_{{2}}u_{{12}}+ ( uu_{{12}}+u_{{1}}u_{{2}} ) p_{{3}}+
 ( p_{{3}}-1 )  ( uu_{{2}}+u_{{1}}u_{{12}} )
  +( u-u_{{2}}+u_{{1}}-u_{{12}}) r_{{4}}+{r_{{4}}}^{2}$=0
&\cite{HietVial}
\\
\hline
E17&$p_{{4}}uu_{{12}}+p_{{5}}uu_{{2}}+p_{{2}}u_{{1}}u_{{2}}+p_{{6}}u_{{1}}
u_{{12}}+r_{{1}}u+r_{{2}}u_{{1}}+r_{{3}}u_{{2}}+r_{{4}}u_{{12}}+s=0
$
&\cite{HietVial}
\\
\hline
\end{tabular}
 \begin{tablenotes}
      \item\label{tnote:Q4} $A_l^2=h(\alpha_l), B_m^2=h(\beta_m)$, where $h(x)=x^4+\delta x^2+1$.
      \item\label{tnote:mKdV} The parameters $p,q,r$ for mKdV satisfy \phantom{${}_{{}-1}{}_{{}-1}$}$p_{l-1,m}/p_{l,m}-(q_{l,m-1}r_{l-1,m})/(q_{l,m}r_{l-1,m-1})=0$ \newline\phantom{The parameters $p,q,r$ for mKDV sa}and
     $\ p_{l-1,m-1}/p_{l,m-1}-(q_{l-1,m-1}r_{l,m})/(q_{l-1,m}r_{l,m-1})=0$.
     \item\label{tnote:sG}The parameters $p,q,r$ for sG satisfy $p_{l+1,m-1}p_{l,m}-(r_{l,m}r_{l+1,m-1})/(r_{l+1,m}r_{l,m-1})=0$  \newline\phantom{The parameters $p,q,r$ for SG sat}and  $p_{l+1,m+1}p_{l,m}-(q_{l,m}q_{l+1,m+1})/(q_{l+1,m}q_{l,m+1})=0.$
 \end{tablenotes}
\end{threeparttable}
}
\vskip0.5cm
\caption{A list of the lattice rules considered in this paper. Equations above the double-line division
are integrable and those below it are non-integrable. \label{tab:rules}}
\end{table}

In this paper, we derive the projective versions of the non-autonomous and autonomous rules of Table 1 in the way described above (some can be found in \cite{Roberts_Tran_FF}).
We prescribe {\em corner initial conditions} or {\em corner boundary conditions},
as previously considered in e.g. \cite{HietVial,Tremblay} and also in \cite{Roberts_Tran_FF}.
That is, we
prescribe the values $x_{l,0}$ and $z_{l,0}$ and $x_{0,m}$ and $z_{0,m}$ with $l,m \in \NN$, on  the borders of
the first quadrant of the lattice.
We work out from the origin using these initial conditions and the projective
lattice rule to generate the top right entries in each square (the rule can depend on the lattice site for non-autonomous equations).
It is clear that $x_{i,j}$ and $z_{i,j}$
with $i,j \in \NN$, are both multinomial expressions in the $2(i+j+1)$ variables given by the initial conditions $x_{l,0}$ and $z_{l,0}$,
$0 \le l \le i$, and $x_{0,m}$ and $z_{0,m}$, $0 < m \le j$.  In the non-autonomous case, these expressions will also contain the parameters from the
$i\cdot j$ lattice rules involved.
Furthermore, one may want to put the corner boundary conditions at a point other than the origin to sample all possible arrangements of lattice rules.
For the non-autonomous $Q_V^{non}$, for example, one can start with one of the following points
$\{(0,0), (1,1),(0,1),(1,0)\}$ to sample all possible parities at the corner.
However, we can always bring the corner point to the origin by shifting so for the rest of our paper, we use the origin as a starting point
and initial values are given on the border of the first quadrant.

To probe the phenomenon of cancellation at the heart of algebraic entropy, following \cite{HietVial, Tremblay, Viallet},
we take the initial values $x_{l,0}(w)$ and $z_{l,0}(w)$  and $x_{0,m}(w)$ and $z_{0,m}(w)$, where $l,m\in\N$, as polynomials in an indeterminate $w$.
\footnote{The specialisation to univariate polynomials also enables concepts like {\em greatest common divisor}, which are not always defined in the multinomial case.}
Using the lattice rule, one can calculate $x_{l,m}(w)$ and $z_{l,m}(w)$ with $l,m>0$.
We factor $x_{l,m}(w)$ and $z_{l,m}(w)$ (over $\Q[w]$) and define $\gcd_{l,m}(w)$  to be their greatest common divisor so we can write
\begin{eqnarray} \label{eq:barxz}
x_{l,m}(w) &=& {\gcd}_{l,m}(w)  \; \bar{x}_{l,m}(w), \label{gcx}  \\
z_{l,m}(w) &=& {\gcd}_{l,m}(w) \; \bar{z}_{l,m}(w). \label{gcz}
\end{eqnarray}
As usual, $\gcd_{l,m}(w)$ is taken to be a monic polynomial and
$x_{l,m}(w)$ and $z_{l,m}(w)$ are relatively prime if and only if
$\gcd_{l,m}(w) \equiv 1$.  Suppressing the $w$-dependence for
ease of presentation, define the non-negative integers
\begin{eqnarray} \label{defsdegree}
d_{l,m} &=&\max(\deg(x_{l,m}),\deg(z_{l,m})) \ge 0,  \label{E:defd} \\
g_{l.m}&=&\deg({\gcd}_{l,m}) \ge 0,  \label{E:defg}\\
\bar{d}_{l,m} &=&\max(\deg(\bar{x}_{l,m}),\deg(\bar{z}_{l,m}))= d_{l,m} - g_{l,m}
\ge 0 \label{E:defbard}
\end{eqnarray}
The key issue (for algebraic entropy) relates to the growth of the degree of
$g_{l,m}$.  Because $x_{l,m}(w)$ and $z_{l,m}(w)$ are considered projectively,
when $g_{l,m} \ge 1$,  they can be replaced by their barred versions. This corresponds to the cancellation of  $\gcd_{l,m}(w)$ from the numerator and denominator of
the rational function $u_{l,m}(w)$. Correspondingly, $d_{l,m}$ is replaced by
the lesser $\bar{d}_{l,m}$ as the (true) degree at the vertex.
This was observed previously in \cite{HietVial,Tremblay}.   One manifestation of this
factorization process (leading to a diagnostic for detecting it)
is that if the corner initial conditions are taken to be integers, the resulting integers $x_{i,j}$ and
$z_{i,j}$ after cancellation of common integer factors would be smaller in magnitude than generically expected.
This is equivalent to saying in the non-projective setting of \eqref{E:utopright2}, and with corner initial conditions
taken in $\QQ$, that some lattice rules give a lower {\em height} for the rational number $u_{l+1,m+1}$
than generically expected. This is closely connected to the concept of Diophantine integrability
for lattice maps, explored in the last part of \cite{Halburd2005Diophantine-Int}.

In typical entropy calculations, the $\gcd$ is discarded and only its degree is noted. For what follows in the
next section, we find it useful to also observe some of its internal structure as we proceed
iteratively away from the origin.
We have the following properties:
\begin{proposition} \label{T:gcd}
For the projective lattice rules of Observation~\ref{O:PCoor},  we
have the following:
\begin{enumerate}
\item[1.]
$0 \leq d_{l+1,m+1} \; \leq \; d_{l,m}+d_{l+1,m}+d_{l,m+1}$
\item[2.]
${\gcd}_{l,m}(w) \; {\gcd}_{l+1,m}(w) \; {\gcd}_{l,m+1}(w) \; \big| \;{\gcd}_{l+1,m+1}(w)$
which implies when $x_{l+1,m+1}(w)$ and $z_{l+1,m+1}(w)$ are not both $0$ that
\subitem{2a.}
$g_{l+1,m+1} \; \geq \; g_{l,m}+g_{l+1,m}+g_{l,m+1}$
\subitem{2b.}  ${\gcd}_{l,m}(w)^3 \; \big| \;{\gcd}_{l+1,m+1}(w) \implies
g_{l+1,m+1} \; \ge \;  3\, g_{l,m}$
\end{enumerate}
 \end{proposition}
\begin{proof}
It is straightforward from Observation~\ref{O:PCoor} that 1. and 2. hold. The statements of (2a) and (2b)
follow easily from the first statement in 2.
\end{proof}

\begin{remarkk}   \label{R:samedeg}
If $x_{l,m}(w)$ and $z_{l,m}(w)$ share the same degree,   \eqref{E:defd}-\eqref{E:defbard} become:
\begin{eqnarray} \label{D:samedeg}
d_{l,m} &=&\deg(x_{l,m})=\deg(z_{l,m})) \ge 0 \label{E:defdspec} \\
\bar{d}_{l,m} &=&\deg(\bar{x}_{l,m})=\deg(\bar{z}_{l,m}))= d_{l,m} - g_{l,m} \ge 0. \label{E:defbardspec}
\end{eqnarray}
If this degree equality of the two components holds at each of the 3 vertices used in (\ref{gcx})-(\ref{gcz}), then
{\bf generically} this is retained with
\begin{equation} \label{E:samedeg}
d_{l+1,m+1}=\deg({x}_{l+1,m+1})=\deg({z}_{l+1,m+1}�)=d_{l,m}+d_{l+1,m}+d_{l,m+1}.
\end{equation}
\end{remarkk}

Part 2. of Proposition~\ref{T:gcd} illustrates the nesting and propagation of the greatest common divisors that occurs for growing $l$ and growing $m$
when there starts to be a non-trivial $\gcd$ at some point.  At any lattice point in the first quadrant, the $\gcd$ of $x_{l,m}(w)$ and $z_{l,m}(w)$ includes all
the greatest common divisors of the polynomials at the
$l m + l + m$ lattice points to its left and/or below it.  This leads to great multiplicities of
factors in the $\gcd$.
We remark that \cite{HietVial} used cancellation (i.e. non-trivial $\gcd$) after $2$ steps on the diagonal as
a criterion to detect integrable rules and Part 2b. of Proposition \ref{T:gcd} shows that once this occurs, it propagates at an exponential rate.

With respect to statement 2. of Proposition~\ref{T:gcd}, we remark that
once we factor out the product ${\gcd}_{l,m}(w) \; {\gcd}_{l+1,m}(w) \; {\gcd}_{l,m+1}(w)$ from $x_{l+1,m+1}(w)$ and
$z_{l+1,m+1}(w)$,  we are left with $f_{l,m}^{(1)}$ of (\ref{projx}) and $f_{l,m}^{(2)}$ of
(\ref{projz}), but now with the  barred arguments as defined by (\ref{gcx}) and
(\ref{gcz}). These polynomials in the barred variables may have a further non-trivial $\gcd$.
We define
\begin{eqnarray}  \label{bargcd}
\overline{\gcd}_{l+1,m+1}(w) =
\gcd\{f_{l,m}^{(1)}\left(\bar{x}_{l,m},\bar{x}_{l+1,m},\bar{x}_{l,m+1},\bar{z}_{l,m},\bar{z}_{l+1,m},\bar{z}_{l,m+1}\right),
\nonumber \\
f_{l,m}^{(2)}\left(\bar{x}_{l,m},\bar{x}_{l+1,m},\bar{x}_{l,m+1},\bar{z}_{l,m},\bar{z}_{l+1,m},\bar{z}_{l,m+1}\right)\}
\end{eqnarray}
so that
\begin{eqnarray}
\bar{x}_{l+1,m+1}(w)=
f_{l,m}^{(1)}\left(\bar{x}_{l,m},\bar{x}_{l+1,m},\bar{x}_{l,m+1},\bar{z}_{l,m},\bar{z}_{l+1,m},\bar{z}_{l,m+1}\right)\,/\,
\overline{\gcd}_{l+1,m+1}(w), \label{xinbar}\\
\bar{z}_{l+1,m+1}(w)=
f_{l,m}^{(2)}\left(\bar{x}_{l,m},\bar{x}_{l+1,m},\bar{x}_{l,m+1},\bar{z}_{l,m},\bar{z}_{l+1,m},\bar{z}_{l,m+1}\right)\,/\,
\overline{\gcd}_{l+1,m+1}(w). \label{zinbar}
\end{eqnarray}

\begin{remarkk} \label{R:Computation}
Computationally, it follows that at each vertex, it is most efficient to record
the triple $\{\bar{x}_{l,m}(w),\bar{z}_{l,m}(w),\gcd_{l,m}(w)\}$, from which
one can reconstruct $x_{l,m}(w)$ and $z_{l,m}(w)$ if necessary. One takes barred
variables as the arguments in
the right-hand side of (\ref{projx}) and (\ref{projz}) and checks the greatest common
divisor, $\overline{\gcd}_{l+1,m+1}(w)$, of the resulting $f_{l,m}^{(1)}$ and $f_{l,m}^{(2)}$. This gives $\bar{x}_{l+1,m+1}(w)$
of \eqref{xinbar} and $\bar{z}_{l+1,m+1}(w)$ of \eqref{zinbar}
whereas $\gcd_{l+1,m+1}(w)$ is updated via
\begin{equation} \label{bargcd2}
{\gcd}_{l+1,m+1}(w) = \overline{\gcd}_{l+1,m+1}(w)\;{\gcd}_{l,m}(w) \; {\gcd}_{l+1,m}(w) \; {\gcd}_{l,m+1}(w).
\end{equation}
\end{remarkk}


\section{Growth of ambient degrees before cancellation }
Our approach in this section is to identify a linear partial difference equation satisfied by the upper bound on the degree
at each vertex, which also represents the generic degree at each vertex before common factors and cancellations are considered. This enables the exponential growth of
these degrees to be verified. Once again, the results hold for both autonomous and non-autonomous equations.
\subsection{Solution of recurrence for the upper bound on degrees}
Part 1. of Proposition~\ref{T:gcd} shows that an upper bound for the degree
$d_{l+1,m+1}$ is provided by the non-negative integer sequence $(a_{l+1,m+1})$ satisfying
\begin{equation} \label{E:bound}
a_{l+1,m+1}=a_{l,m}+a_{l+1,m}+a_{l,m+1}.
\end{equation}
The value $d_{l+1,m+1}$ satisfies this same recurrence at any vertex  if  either $x_{l+1,m+1}$
or $z_{l+1,m+1}$ contains the term comprising a maximal degree polynomial in $w$ from each of the other 3 vertices.  The argument of Remark \ref{R:samedeg} above suggests that, generically, $d_{l+1,m+1}=a_{l+1,m+1}$
when the boundary values are chosen with both components at each vertex having the same degree in $w$.

For initial numerical experiments, we take two cases of such boundary values, namely:
\begin{eqnarray}
\text{Case I:}&
&\begin{cases}
             x_{0,0}= aw +b\text{ and } z_{0,0}=cw+d, \quad  a\ne 0,c\ne 0,\ a,b,c,d\in\ZZ, \\
             x_{0,m}, z_{0,m} \in \ZZ, \quad m=1,2,\ldots \\
             x_{l,0}, z_{l,0} \in \ZZ, \quad l=1,2,\ldots.
               \end{cases}
\label{ICI}\\
\text{Case II:}&
&\begin{cases}
             x_{0,m}, z_{0,m}, x_{l,0}, z_{l,0} \text{ all affine} \in\ZZ[w], \quad l,m =0,1,2 \ldots .
                \end{cases}
\label{ICII}
\end{eqnarray}
For Case I boundary values, we have $d_{0,0}=1$ and $d_{0,m}=d_{l,0}=0$ for $l,m > 0$.
For Case II boundary values, we have  $d_{l,0}=d_{0,m}=1$ for $\l,m \ge 0$.

Our numerical experiments on the equations of Table 1 use Maple
to calculate the polynomials $x_{l,m}(w)$ and
$z_{l,m}(w)$ for $0 \le l, m \le 11$ (i.e. a square of 144 lattice sites based at the origin)
and Case I or Case II boundary values with random integer coefficients
chosen in the interval $[1 ,400]$. For ABS equations,  parameters $\{\alpha_i, \beta_j \mid i=0,...,10,j=0,....10\}$ are also taken randomly with integer values such that the corresponding equations are not degenerate.
They confirm the agreement  $d_{l+1,m+1}=a_{l+1,m+1}$ where $d_{l+1,m+1}$ is calculated from \eqref{defsdegree} with \eqref{projx} and \eqref{projz} and $a_{l+1,m+1}$ from \eqref{E:bound}, with the corresponding same boundary values in each case. Consequently, we call $a_{l+1,m+1}$ the {\em ambient} degree at the vertex (i.e. before any analysis of a possible common factor and cancellation of this factor). If, furthermore, there is actually no
non-trivial $\gcd$ for the two polynomials $x_{l+1,m+1}$ and $z_{l+1,m+1}$ (so {\em no} cancellations possible), then $d_{l+1,m+1}=a_{l+1,m+1}=\bar{d}_{l+1,m+1}$ gives the true degree throughout the lattice.

The solution of \eqref{E:bound} with Case I boundary values can be represented as (part of) a semi-infinite array $A_{0,0}=A_{0,0}\,[l,m]$,
with the indices $m \ge 0$ and $l \ge 0$  measured vertically, respectively to the right,
with respect to the origin at the bottom left corner so $A_{0,0}\,[0,0]=1$. The solution for Case II boundary values corresponds to removing the first column and the last row of this array.

\begin{equation}
\label{M:Total_degree_constant}
A_{0,0}[l,m]= \tiny \left[ \begin {array}{cccccccccc} 0&1&17&145&833&3649&13073&40081&
108545&265729\\ \noalign{\medskip}0&1&15&113&575&2241&7183&19825&48639
&108545\\ \noalign{\medskip}0&1&13&85&377&1289&3653&8989&19825&40081
\\ \noalign{\medskip}0&1&11&61&231&681&1683&3653&7183&13073
\\ \noalign{\medskip}0&1&9&41&129&321&681&1289&2241&3649
\\ \noalign{\medskip}0&1&7&25&63&129&231&377&575&833
\\ \noalign{\medskip}0&1&5&13&25&41&61&85&113&145\\ \noalign{\medskip}0
&1&3&5&7&9&11&13&15&17\\ \noalign{\medskip}0&1&1&1&1&1&1&1&1&1
\\ \noalign{\medskip}1&0&0&0&0&0&0&0&0&0\end {array} \right].
 \end{equation}
A formula for the double-indexed sequence $A_{0,0}[l,m]$
is given by the following theorem.
\begin{theorem}
\label{T:GerFunc}
Let $A_{0,0}[l,m]$ denote the integer sequence solution $a_{l,m}$ satisfying \eqref{E:bound}
\label{E:RecEq}
for all $l,m\geq 0$ with boundary values $a_{0,0}=1$ and $a_{0,m}=a_{l,0}=0$ for $l,m > 0$.
For $l,m >0$, let $c_{l,m}$ be the coefficient of  $x^{m-1}$ in the Taylor expansion of the following function $g_{l}(x)$ around $0$, i.e,
\begin{equation} \label{E:Taylor}
g_l(x)=\frac{(1+x)^{l-1}}{(1-x)^l}=\sum_{m=1}^{\infty} c_{l,m}\,x^{m-1}.
\end{equation}
We  have for all $l,m>0$,
\begin{equation} \label{E:UpperBound}
A_{0,0}[l,m]=c_{l,m}=\sum_{i+j=m-1}\binom{l-1}{i}\binom{j+l-1}{j}.
\end{equation}
\end{theorem}
\begin{proof}
It is easy to see from \eqref{E:bound} that $a_{1,m}=a_{l,1}=1$ is what results from applying the recurrence with prescribed
initial conditions given on the axes as in the statement of the theorem.  On the other hand,
\eqref{E:Taylor} also shows $c_{1,m}=1$ as the Taylor coefficients of $g_1(x)=1/(1-x)$ for all $m \ge 1$ and $c_{l,1}=g_l(0)=1$ for all $l \ge 1$.
It is adequate to  prove that  $c_{l,m}$ satisfies the same recurrence~\eqref{E:bound} as
$a_{l,m}$.
We have
\begin{equation*}
g_{l+1}(x)=\sum_{m=1}^{\infty} c_{l+1,m}\,x^{m-1}.
\end{equation*}
On the other hand, we also have
\begin{align*}
g_{l+1}(x)&=g_{l}(x)\;\frac{1+x}{1-x}\\
&=\left(\sum_{i=1}^{\infty} c_{l,i}\,x^{i-1}\right)\left(1+\sum_{j=1}^{\infty} 2\, x^j\right).
\end{align*}
Equating coefficients of $x^{m-1}$, we obtain
\[
c_{l+1,m}=c_{l,m}+\sum_{i=1}^{m-1}2\, c_{l,i}.
\]
This implies
\begin{align*}
c_{l+1,m+1}&=\sum_{i=1}^{m}2 c_{l,i}+c_{l,m+1}\\
&=\left(\sum_{i=1}^{m-1}2 c_{l,i}+c_{l,m}\right)+c_{l,m}+c_{l,m+1}\\
&=c_{l+1,m}+c_{l,m}+c_{l,m+1},
\end{align*}
which is the same recurrence given by~~\eqref{E:bound}.

We now calculate $c_{l,m}$  using the following formula for $l\geq 1$,
\[
\frac{1}{(1-x)^l}=\sum_{j=0}^{\infty}\binom{j+l-1}{j} x^j.
\]
Expanding $(1+x)^{l-1}$, we get
\begin{equation*}
g_{l}(x)=\left(\sum_{i=0}^{\infty}\binom{l-1}{i} x^i\right)\left( \sum_{j=0}^{\infty}\binom{j+l-1}{j} x^j\right).
\end{equation*}
This gives us
the formula given by~\eqref{E:UpperBound}.
\end{proof}

We mention that from the particular solution $A_{0,0}[l,m]$ of \eqref{E:UpperBound},
we can generate the general solution  for  the recurrence \eqref{E:bound} with corner boundary values \cite[Proposition 6]{RobertsTran14}.


\subsection{Exponential growth of the ambient degree}
We remark that the sequence $a_{l,m}$ is well known as the
Delannoy numbers.
Apart from the {\em ad hoc} approach to our solutions above, an alternative way to study the recurrence \eqref{E:bound} is as a linear
partial difference equation with constant coefficients, for
which the method of bivariate generating functions applies \cite[Chapter 6.4]{dobrushkin}.
Generating functions for the full double sequence and diagonal (central) sequence with Case I boundary conditions are known:
\begin{eqnarray*}
A(x,y)&=&\sum_{l=0}^{\infty} \sum_{m=0}^{\infty} a_{l,m}\, x^{l} y^m = \frac{1}{1-x-y-xy}, \\
D(x)&=& \sum_{m=0}^{\infty} a_{m,m}\, x^{m} = \frac{1}{\sqrt{1-6x+x^2}}.
\end{eqnarray*}
These lead to the asymptotics
$$a_{m,m} \sim \frac{\cosh(\frac{\log 2}{4})}{\sqrt{\pi}}\,
 (3 + 2 \sqrt{2})^m\, \frac{1}{\sqrt{m}} $$
whence the diagonal entropy
\begin{equation} \label{entropy}
\epsilon = \lim_{m\rightarrow \infty} \frac{1}{m} \log(a_{m,m}) =\log (3+2\sqrt(2))\approx 1.76.
\end{equation}

\section{ Integrable lattice rules-vanishing entropy}

In this section, we identify linear recurrences for the degree of the greatest common divisor at each vertex for the lattice equations given in section 2.  Together with the recurrence \eqref{E:bound}, this helps to furnish
us with a linear partial difference equation for $\bar{d}_{l,m}$,
whose growth is ultimately our interest.  Sometimes, we can solve
this recurrence and conclude polynomial growth of degree in the vertex coordinates $l$ and $m$,
or exponential growth.

\subsection{Exponentially growing gcd: the inherited gcd and the spontaneous gcd}

As mentioned in the previous section, we use Maple
to calculate the polynomials $x_{l,m}(w)$ and
$z_{l,m}(w)$ for $0 \le l, m \le 9$
and Case I or Case II boundary values \eqref{ICI}-\eqref{ICII} with random integer coefficients
chosen in the interval $[1,400]$. Now, however, we are interested in the common factor of these polynomials so
we follow the procedure outlined in Remark \ref{R:Computation} of Section 2. For all
autonomous versions of equations given above KdV in Table 1, the degree array $g_{l,m}$ of the $\gcd$ (see \eqref{E:defg} for Case I boundary values \eqref{ICI} is given by
\begin{equation} \label{gcd1}
g^{I}_{l,m}= \tiny \left[
 \begin {array}{cccccccccc} 0&0&14&140&826&3640&13062&40068&
108530&265712\\ \noalign{\medskip}0&0&12&108&568&2232&7172&19812&48624
&108530\\ \noalign{\medskip}0&0&10&80&370&1280&3642&8976&19812&40068
\\ \noalign{\medskip}0&0&8&56&224&672&1672&3642&7172&13062
\\ \noalign{\medskip}0&0&6&36&122&312&672&1280&2232&3640
\\ \noalign{\medskip}0&0&4&20&56&122&224&370&568&826
\\ \noalign{\medskip}0&0&2&8&20&36&56&80&108&140\\ \noalign{\medskip}0
&0&0&2&4&6&8&10&12&14\\ \noalign{\medskip}0&0&0&0&0&0&0&0&0&0
\\ \noalign{\medskip}0&0&0&0&0&0&0&0&0&0\end {array}
 \right].
\end{equation}
and for Case II boundary values \eqref{ICII} is
\begin{equation}  \label{gcd2}
g^{II}_{l,m}= \tiny \left[
\begin {array}{cccccccccc} 0&0&144&1104&5568&22272&75408&
224016&598272&1462400\\ \noalign{\medskip}0&0&112&784&3584&12992&39984
&108432&265600&598272\\ \noalign{\medskip}0&0&84&532&2184&7112&19740&
48540&108432&224016\\ \noalign{\medskip}0&0&60&340&1240&3592&8916&
19740&39984&75408\\ \noalign{\medskip}0&0&40&200&640&1632&3592&7112&
12992&22272\\ \noalign{\medskip}0&0&24&104&288&640&1240&2184&3584&5568
\\ \noalign{\medskip}0&0&12&44&104&200&340&532&784&1104
\\ \noalign{\medskip}0&0&4&12&24&40&60&84&112&144\\ \noalign{\medskip}0
&0&0&0&0&0&0&0&0&0\\ \noalign{\medskip}0&0&0&0&0&0&0&0&0&0\end {array}
 \right].
\end{equation}
The zero entries in the  first row and column of  $g^{I}_{l,m}$ and $g^{II}_{l,m}$ reflect
the generic case of no common factors in the boundary values, i.e.
\begin{equation} \label{generalgcd}
g_{l,0}=g_{0,m}=0, \quad l,m \ge 0.
\end{equation}
We note that these arrays give a
lower bound for the non-autonomous version of these equations, i.e. the degree $g_{l,m}$ of the gcd is  sometimes higher than for their  autonomous versions.
One explanation for this is the freedom of coefficients of the original equations, whence
there might be extra common factors appearing at some vertices  when we iterate those rules.
Therefore, it leads to even bigger cancellations  which does not affect the polynomial growth of these equations.

The entries of $g^{I}_{l,m}$ and $g^{II}_{l,m}$ illustrate the $\gcd$ properties Part 2.(a,b) of
Proposition~\ref{T:gcd}. The numbers on the
diagonal in both cases appear to be growing exponentially with an exponent $\simeq
1.7$. This is close to the rate of exponential growth of the maximal degree on the diagonal
\eqref{entropy} and exceeds the minimum exponential growth of $\log(3) \simeq 1.099$ expected by Part 2.(b). This is
because of spontaneous new additions to $\gcd_{l+1,m+1}(w)$ that augment the product
${\gcd}_{l,m}(w) \; {\gcd}_{l+1,m}(w) \; {\gcd}_{l,m+1}(w)$ inherited from the 3 other vertices of the lattice square.  We call the latter
product the {\em inherited} $\gcd$ and we use the term {\em spontaneous} $\gcd$ for
$\overline{\gcd}_{l+1,m+1}(w)$ of \eqref{bargcd}
and \eqref{bargcd2} for $l,m \ge 0$. Defining
\begin{equation} \label{barg}
\overline{g}_{l+1,m+1}=\deg({\overline{\gcd}}_{l+1,m+1}) \geq 0,
\end{equation}
we have from \eqref{bargcd2}:
\begin{equation} \label{barg2}
\overline{g}_{l+1,m+1}=g_{l+1,m+1}-g_{l,m}-g_{l+1,m}-g_{l,m+1}, \quad l,m \ge 0.
\end{equation}


Using \eqref{barg2}, we can calculate the array of spontaneous $\gcd$ degrees
corresponding to \eqref{gcd1} and \eqref{gcd2} -- see the left hand side of, respectively, Figures \ref{F:sponred1} and \ref{F:sponred2}. Using
\eqref{M:Total_degree_constant} and the $\gcd$ values \eqref{gcd1} and \eqref{gcd2}, we can also calculate the reduced degree values \eqref{E:defbard}
in each case - see the right hand side of Figures \ref{F:sponred1} and \ref{F:sponred2}. These arrays are found to be the same for all rules  given above the KdV equation of Table 1.

\begin{figure}
\begin{equation} \label{M:S_constant_gcd}
\overline{g}^{I}_{l,m}=
\tiny \left[ \begin {array}{cccccccccc} 0&0&2&6&10&14&18&22&26&28\\
 \noalign{\medskip}0&0&2&6&10&14&18&22&24&26\\
  \noalign{\medskip}0&0&2&6&10&14&18&20&22&22\\
\noalign{\medskip}0&0&2&6&10&14&16&18&18&18\\
\noalign{\medskip}0&0&2&6&10&12&14&14&14&14\\
\noalign{\medskip}0&0&2&6
&\fbox{{\bf8}}&10&10&10&10&10\\ \noalign{\medskip}0&0&2&\fbox{{\bf 4}}&6&6&6&6&6&6\\
\noalign{\medskip}0&0&0&2&2&2&2&2&2&2\\
 \noalign{\medskip}0&0&0&0&0&0&0&0&0&0\\
 \noalign{\medskip}0&0&0&0&0&0&0&0&0&0
\end {array} \right] \qquad
\overline{d}^{I}_{l,m}=
\tiny \left[ \begin {array}{cccccccccc} 0&1&3&5&7&9&11&13&15&17
\\ \noalign{\medskip}0&1&3&5&7&9&11&13&15&15\\
\noalign{\medskip}0&1&3&5&7&9&11&13&13&13\\
\noalign{\medskip}0&1&3&5&7&9&11&\tikzmark{left}11&\color{blue}{\underline{{\it 11}}}&\color{blue}{{\bf11}}
\\ \noalign{\medskip}0&1&3&5&7&9&9&\color{blue}{{\bf 9}}&9&\color{blue}{\underline{\it 9}}\\
\noalign{\medskip}0&1&3&5&7&7&7&\color{blue}{\underline{\it 7}}&\color{blue}{{\bf 7}}&7\tikzmark{right}\\
\noalign{\medskip}0&1&{\bf\circled{$3$}}&5&5&5&5&5&5&5
\\ \noalign{\medskip}0&1&3&{\bf \circled{$3$}} &3&3&3&3&3&3\\ \noalign{\medskip}0&1&1&1&1
&1&1&1&1&1\\ \noalign{\medskip}1&0&0&0&0&0&0&0&0&0
\end {array}
 \right].
\end{equation}
\DrawBox[thin]
\caption{\small Array of spontaneous $\gcd$ degrees defined by \eqref{barg2} (left) and reduced degrees defined by
\eqref{E:defbard} (right) for the autonomous $H_1$ equation  and boundary values Case I.}
   \label{F:sponred1}
\end{figure}
\begin{figure}
\begin{equation}
\label{M:S_linear_gcd}
\overline{g}^{II}_{l,m}=\tiny \left[ \begin {array}{cccccccccc} 0&0&32&64&96&128&160&192&224&256
\\ \noalign{\medskip}0&0&28&56&84&112&140&168&196&224
\\ \noalign{\medskip}0&0&24&48&72&96&120&144&168&192
\\ \noalign{\medskip}0&0&20&40&60&80&100&120&140&160
\\ \noalign{\medskip}0&0&16&32&48&64&80&96&112&128
\\ \noalign{\medskip}0&0&12&24&36&\fbox{{ \bf  48}}&60&72&84&96
\\ \noalign{\medskip}0&0&8&16&\fbox{{\bf 24}}&32&40&48&56&64
\\ \noalign{\medskip}0&0&4&8&12&16&20&24&28&32
\\ \noalign{\medskip}0&0&0&0&0&0&0&0&0&0
\\ \noalign{\medskip}0&0&0&0&0&0&0&0&0&0
\end {array} \right]
\qquad
\overline{d}^{II}_{l,m}=
\tiny  \left[ \begin {array}{cccccccccc} 1&19&37&55&73&91&109&127&145&163
\\ \noalign{\medskip}1&17&33&49&65&81&97&113&129&145
\\ \noalign{\medskip}1&15&29&43&57&71&85&99&113&127
\\ \noalign{\medskip}1&13&25&37&49&61&73&85&97&109
\\ \noalign{\medskip}1&11&21&31&41&51&61&71&81&91
\\ \noalign{\medskip}1&9&17&25&33&41&49&\tikzmark{left}57&\color{blue}{\underline{{\it 65}}}&\color{blue}{{\bf 73}}
\\ \noalign{\medskip}1&7&13&{\bf \circled{ $ 19$}}&25&31&37&\color{blue}{{\bf 43}}&49&\color{blue}{\underline{{\it 55}}}
\\ \noalign{\medskip}1&5&9&13&{\bf \circled{ $ 17$}} &21&25&\color{blue}{\underline{{\it 29}}}&\color{blue}{{\bf 33}}&37\tikzmark{right}
\\ \noalign{\medskip}1&3&5&7&9&11&13&15&17&19
\\ \noalign{\medskip}1&1&1&1&1&1&1&1&1&1
\end {array}
 \right].
\end{equation}
\DrawBox[thin]
\caption{\small Array of spontaneous $\gcd$ degrees defined by \eqref{barg2} (left) and reduced degrees defined by \eqref{E:defbard} (right) for the autonomous $H_1$ equation and boundary values Case II. }
  \label{F:sponred2}
\end{figure}

The spontaneous $\gcd$ degrees reveal more apparent structure than their unbarred counterparts.
In \eqref{M:S_constant_gcd}, $\overline{g}^{I}_{l+1,m+1}=\overline{g}^{I}_{l,m}+4$, once a non-zero
entry appears on a diagonal $m-l \in \NN$. In \eqref{M:S_linear_gcd}, $\overline{g}^{II}_{l+1,m}=\overline{g}^{II}_{l,m}+4(m-1)$ for $l,m \ge 2$.
We also observe from our numerical experiments:

\begin{observation}\label{observe}
For the autonomous lattice equations  given above the KdV equation of Table 2 and the boundary conditions Case I and Case II, we have found that
\begin{equation} \label{bargcddeg}
\overline{g}_{l,m}+\overline{g}_{l+1,m+1}=
2(\overline{d}_{l,m-1}+\overline{d}_{l-1,m})
\end{equation} for $l,m\geq 2$.
This leads to
\begin{equation}
\label{E:deg_recursive}
\overline{d}_{l+1,m+1}=\overline{d}_{l+1,m}+\overline{d}_{l,m+1}+\overline{d}_{l-1,m-1}-\overline{d}_{l,m-1}-\overline{d}_{l-1,m}.
\end{equation}
For autonomous KdV and MX, it holds for  $l,m\geq 3$.
\end{observation}

In the two figures, the two numbers enclosed by boxes in the left array and by circles in the right array illustrate the relationship \eqref{bargcddeg}. In the right array,
the three numbers in bold and the three numbers in underlining/italic  in $\overline{d}_{l,m}$ illustrate the recurrence \eqref{E:deg_recursive}.
It will turn out that \eqref{E:deg_recursive} is a consequence of  Theorem \ref{Main} that covers both autonomous and non-autonomous equations of the
next section.

\subsection{Polynomial growth of integrable lattice rules\label{SS:poly_growth}}
We note that  sustained occurrences of a spontaneous $\gcd$ as we iterate the rule in the first quadrant is a necessary condition
for subexponential growth of $\overline{d}_{l,m}$. As seen in \eqref{gcd2}, the first gcd appears at the point $(2,2)$ for Case II initial values. It is actually  related to an important observation of  \cite{HietVial,Q5}
which provides a mechanism for this and one that can be verified quickly for a rule.  From \cite{HietVial,Q5}, it is known that for autonomous equations in the ABS list and $Q_V$ for any $2\times 2$ lattice square $[l-1,l+1] \times [m-1,m+1$],  we obtain a  common factor $A_{l+1,m+1}$ of  $x_{l+1,m+1}$ and $z_{l+1,m+1}$ of \eqref{projx}-\eqref{projz} for {\em arbitrary} initial values
at the 5 corner sites $\{ (l-1,m-1), (l-1, m), (l-1, m+1), (l,m-1), (l+1,m-1)\}$, see Figure \ref{CIF:Alm}.
This factorization property over any $2\times 2$ lattice square was used to search for new integrable equations in \cite{HietVial}.
For many equations, as indicated in the figure,
the common factor  $A_{l+1,m+1}$ can actually be written in terms of  coordinates $x$ and $z$ at just the 2 sites
$(l-1,m)$ and $(l,m-1)$ and might depend also on lattice parameters or the parity of $(l-1)$ and $(m-1)$. We omit all factors that do not depend on variables $x$ and $z$.

We find that this property also holds for many non-autonomous versions lattice equations.
Table \ref{tab:commonfactor} gives the details of $A_{l,m}$ for integrable equations in Table \ref{tab:rules} above the KdV equation.
In the appendix, we give the factor $A_{l,m}$ for the non-autonomous $Q_V^{non}$ of \eqref{NonQV}.
It should be noted that for all equations in Table 2, and in the appendix, the common factor  $A$  is `quartic' in terms of the off-diagonal variables
of the first square.

\begin{figure}[h]
\begin{center}
\begin{tikzpicture}[line cap=round,line join=round,>=triangle 45,x=3.9cm,y=3.9cm]
\draw (-0.5,-0.5)-- (0.5,-0.5);
\draw (-0.5,-0.5)-- (-0.5,0.5);
\draw (-0.5,0.5)-- (0.5,0.5);
\draw (0.5,0.5)-- (0.5,-0.5);
\draw (0.0,0.5)-- (0.0,-0.5);
\draw (-0.5,0.0)-- (0.5,0.0);
\draw  (-0.5,-0.5) circle (1.5pt);
\draw  (0.5,-0.5) circle (1.5pt);
\draw  (-0.5,0.5) circle (1.5pt);
\draw  (0.0,0.5) circle (1.5pt);
\draw  (0.0,0.0) circle (1.5pt);
\draw  (0.5,0.5) circle (1.5pt);
\draw (0.5,0.0) circle (1.5pt);
\draw [dash pattern=on 1pt off 1pt] (-0.5,0.0)-- (0.0,-0.5);
 \fill (0.0,-0.5) circle (3pt);
\fill (-0.5,0.0) circle (3 pt);
\fill (0.5,0.5) circle (3pt);
\begin{scriptsize}
\draw(-0.50062430240991325,-0.57066662394213983) node {$(l-1,m-1)$};
\draw (0.542851971291231,-0.57066662394213983) node {$(l+1,m-1)$};
\draw (-0.50062430240991325,0.5718096497590043) node {$(l-1,m+1)$};
\draw(0.02268261726609411,0.5718096497590043) node {$(l,m+1)$};
\draw (0.6642851971291231,0.05350273008299695) node {$(l+1,m)$};
\draw (0.09568261726609411,0.05350273008299695) node {$(l,m)$};
\draw (0.02268261726609411,-0.57066662394213983) node {$(l,m-1)$};
\draw (-0.71962430240991325,0.05350273008299695) node {$(l-1,m)$};
\draw (0.532851971291231,0.5718096497590043) node {$(l+1,m+1)$};
\end{scriptsize}
\draw (0.7351971291231,0.458096497590043) node {$\bf {A_{l+1,m+1}}$};
\end{tikzpicture}
\end{center}
 \caption{Some lattice equations produce a common factor $A_{l+1,m+1}$ of  $x_{l+1,m+1}$ and $z_{l+1,m+1}$ over any $2 \times 2$ square that
often only depends on $x$ and $z$ at the 2 sites
$(l-1,m)$ and $(l,m-1)$.}\label{CIF:Alm}.
\end{figure}
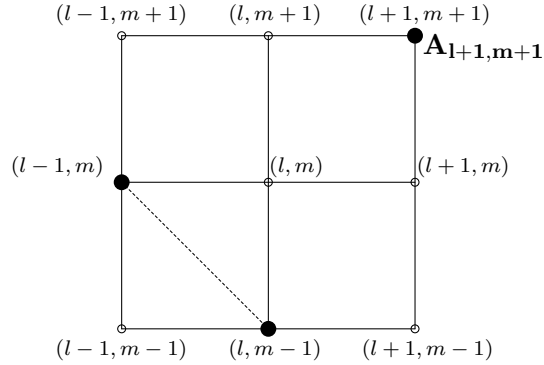
In summary, we have the following proposition.
\begin{proposition} \label{prop:A}
Consider the lattice rules  of Table \ref{tab:rules} which are mentioned in Table \ref{tab:commonfactor} and the non-autonomous $Q_{V}^{non}$ of \eqref{NonQV}.
For the projective version of these rules iterated on any $2\times 2$ lattice block (Figure 4), the iterates $x_{l+1,m+1}$ and $z_{l+1,m+1}$ at the top right corner share a common factor
$A_{l+1,m+1}$ that is given in Table 2 or the appendix.  This factor depends only on the variables in the $2\times 2$ lattice block at the sites $(l-1, m)$ and $(l,m-1)$ and is
homogeneous of degree $4$ in these variables. Extending the lattice block to the right and upwards to a $3\times 3$ lattice block with top right vertex at $(l+2,m+2)$, we find
the divisibility property
$$A_{l+1,m+1} \mid A_{l+2,m+2}.$$
In the case that $x_{l,m}$ and $z_{l,m}$ are taken as polynomials of equal degree $d_{l,m}$ in an indeterminate $w$, we have:
\begin{equation} \label{eq:uniformdeg}
\deg(A_{l+1,m+1}(w))= 2 \, (d_{l-1,m} + d_{l,m-1}).
\end{equation}
\end{proposition}
\begin{proof}
The existence and form of $A_{l+1,m+1}$ for the indicated equations follows from direct computations using maple.
\end{proof}

\begin{landscape}
\begin{table}[h]
\begin{tabular}{|l |p{19cm}|}
\hline
$Q_1$ &  $\delta\,{\alpha_{{l-1}}}^{2}{z_{{l,m-1}}}^{2}{z_{{l-1,m}}}^{2}-2\,
\delta\,\alpha_{{l-1}}\beta_{{m-1}}{z_{{l,m-1}}}^{2}{z_{{l-1,m}}}^{2}+
\delta\,{\beta_{{m-1}}}^{2}{z_{{l,m-1}}}^{2}{z_{{l-1,m}}}^{2}-{x_{{l,m
-1}}}^{2}{z_{{l-1,m}}}^{2}+2\,x_{{l,m-1}}x_{{l-1,m}}z_{{l,m-1}}z_{{l-1
,m}}-{x_{{l-1,m}}}^{2}{z_{{l,m-1}}}^{2}
$
\\
\hline
$Q_2$& ${\alpha_{{l-1}}}^{4}{z_{{l,m-1}}}^{2}{z_{{l-1,m}}}^{2}-4\,{\alpha_{{l-
1}}}^{3}\beta_{{m-1}}{z_{{l,m-1}}}^{2}{z_{{l-1,m}}}^{2}+6\,{\alpha_{{l
-1}}}^{2}{\beta_{{m-1}}}^{2}{z_{{l,m-1}}}^{2}{z_{{l-1,m}}}^{2}-4\,
\alpha_{{l-1}}{\beta_{{m-1}}}^{3}{z_{{l,m-1}}}^{2}{z_{{l-1,m}}}^{2}+{
\beta_{{m-1}}}^{4}{z_{{l,m-1}}}^{2}{z_{{l-1,m}}}^{2}-2\,{\alpha_{{l-1}
}}^{2}x_{{l,m-1}}z_{{l,m-1}}{z_{{l-1,m}}}^{2}-2\,{\alpha_{{l-1}}}^{2}x
_{{l-1,m}}{z_{{l,m-1}}}^{2}z_{{l-1,m}}+4\,\alpha_{{l-1}}\beta_{{m-1}}x
_{{l,m-1}}z_{{l,m-1}}{z_{{l-1,m}}}^{2}+4\,\alpha_{{l-1}}\beta_{{m-1}}x
_{{l-1,m}}{z_{{l,m-1}}}^{2}z_{{l-1,m}}-2\,{\beta_{{m-1}}}^{2}x_{{l,m-1
}}z_{{l,m-1}}{z_{{l-1,m}}}^{2}-2\,{\beta_{{m-1}}}^{2}x_{{l-1,m}}{z_{{l
,m-1}}}^{2}z_{{l-1,m}}+{x_{{l,m-1}}}^{2}{z_{{l-1,m}}}^{2}-2\,x_{{l,m-1
}}x_{{l-1,m}}z_{{l,m-1}}z_{{l-1,m}}+{x_{{l-1,m}}}^{2}{z_{{l,m-1}}}^{2}
$
\\
\hline
$Q_3$ &
$
{\delta}^{2}{\alpha_{{l-1}}}^{4}{z_{{l,m-1}}}^{2}{z_{{l-1,m}}}^{2}-2\,
{\delta}^{2}{\alpha_{{l-1}}}^{2}{\beta_{{m-1}}}^{2}{z_{{l,m-1}}}^{2}{z
_{{l-1,m}}}^{2}+{\delta}^{2}{\beta_{{m-1}}}^{4}{z_{{l,m-1}}}^{2}{z_{{l
-1,m}}}^{2}-4\,{\alpha_{{l-1}}}^{3}\beta_{{m-1}}x_{{l,m-1}}x_{{l-1,m}}
z_{{l,m-1}}z_{{l-1,m}}+4\,{\alpha_{{l-1}}}^{2}{\beta_{{m-1}}}^{2}{x_{{
l,m-1}}}^{2}{z_{{l-1,m}}}^{2}+4\,{\alpha_{{l-1}}}^{2}{\beta_{{m-1}}}^{
2}{x_{{l-1,m}}}^{2}{z_{{l,m-1}}}^{2}-4\,\alpha_{{l-1}}{\beta_{{m-1}}}^
{3}x_{{l,m-1}}x_{{l-1,m}}z_{{l,m-1}}z_{{l-1,m}}
$
\\
\hline
$Q_4$
&
$ \left( {\alpha_{{l-1}}}^{4}-2\,{\alpha_{{l-1}}}^{2}{\beta_{{m-1}}}^{2
}+{\beta_{{m-1}}}^{4} \right) {x_{{l,m-1}}}^{2}{x_{{l-1,m}}}^{2}+
 ( -{\alpha_{{l-1}}}^{4}{\beta_{{m-1}}}^{2}-{\alpha_{{l-1}}}^{2}{
\beta_{{m-1}}}^{4}-2\,\delta\,{\alpha_{{l-1}}}^{2}{\beta_{{m-1}}}^{2}-
2\,A_{{l-1}}B_{{m-1}}\alpha_{{l-1}}\beta_{{m-1}}-{\alpha_{{l-1}}}^{2}-
{\beta_{{m-1}}}^{2}) {x_{{l,m-1}}}^{2}{z_{{l-1,m}}}^{2}+
 ( 4\,{\alpha_{{l-1}}}^{3}{\beta_{{m-1}}}^{3}+2\,\delta\,{\alpha_
{{l-1}}}^{3}\beta_{{m-1}}+2\,\delta\,\alpha_{{l-1}}{\beta_{{m-1}}}^{3}
+2\,A_{{l-1}}B_{{m-1}}{\alpha_{{l-1}}}^{2}+2\,A_{{l-1}}B_{{m-1}}{\beta
_{{m-1}}}^{2}+4\,\alpha_{{l-1}}\beta_{{m-1}} ) z_{{l-1,m}}z_{{l,
m-1}}x_{{l-1,m}}x_{{l,m-1}}+ ( -{\alpha_{{l-1}}}^{4}{\beta_{{m-1}
}}^{2}-{\alpha_{{l-1}}}^{2}{\beta_{{m-1}}}^{4}-2\,\delta\,{\alpha_{{l-
1}}}^{2}{\beta_{{m-1}}}^{2}-2\,A_{{l-1}}B_{{m-1}}\alpha_{{l-1}}\beta_{
{m-1}}-{\alpha_{{l-1}}}^{2}-{\beta_{{m-1}}}^{2} ) {z_{{l,m-1}}}^
{2}{x_{{l-1,m}}}^{2}+ ( {\alpha_{{l-1}}}^{4}-2\,{\alpha_{{l-1}}}^
{2}{\beta_{{m-1}}}^{2}+{\beta_{{m-1}}}^{4} ) {z_{{l-1,m}}}^{2}{z
_{{l,m-1}}}^{2}
$
\\
\hline
$H_1$ & $\left(x_{l-1,m}z_{l,m-1}-z_{l-1,m}x_{l,m-1}\right)^2$
\\
\hline
$H_2$ &
$(\alpha_{{l-1}}z_{{l,m-1}}z_{{l-1,m}}-\beta_{{m-1}}z_{{l,m-1}}z_{{l-1,m}}-x_{{l,m-1}}z_{{l-1,m}}+x_{{l-1,m}}z_{{l,m-1}} )
(\alpha_{{l-1}}z_{{l,m-1}}z_{{l-1,m}}-\beta_{{m-1}}z_{{l,m-1}}z_{{l-1,m}}+x_{{l,m-1}}z_{{l-1,m}}-x_{{l-1,m}}z_{{l,m-1}})
$
\\
\hline
$H_3$ &
$
\left( \alpha_{{l-1}}x_{{l,m-1}}z_{{l-1,m}}-\beta_{{m-1}}x_{{l-1,m}}z
_{{l,m-1}} \right)  \left( -\alpha_{{l-1}}x_{{l-1,m}}z_{{l,m-1}}+\beta
_{{m-1}}x_{{l,m-1}}z_{{l-1,m}} \right)
$
\\
\hline
mKdV & $\left( r_{{l-1,m-1}}x_{{l,m-1}}z_{{l-1,m}}-x_{{l-1,m}}z_{{l,m-1}}
 \right)  \left( p_{{l-1,m-1}}x_{{l,m-1}}z_{{l-1,m}}-q_{{l-1,m-1}}x_{{
l-1,m}}z_{{l,m-1}} \right)
$
\\
\hline
sG &
$
\left( p_{{l-1,m-1}}x_{{l,m-1}}x_{{l-1,m}}-q_{{l-1,m-1}}z_{{l,m-1}}z_
{{l-1,m}} \right)  \left( r_{{l-1,m-1}}x_{{l,m-1}}x_{{l-1,m}}-z_{{l,m-
1}}z_{{l-1,m}} \right)
$
\\
\hline
E16 &
$( z_{{l-1,m}}x_{{l,m-1}}+p_{{5}}p_{{3}}\ z_{{l,m-1}}x_{{l-1,m}})
 (p_{{1}} p_{{3}}p_{{5}}p_{{6}} \ x_{{l-1,m}}z_{{l,m-1}}-p_{{1}}p_{{2}}p_
{{3}}\  x_{{l-1,m}}z_{{l,m-1}}+p_{{2}} p_{{5}}p_{{6}}x_{{l-1,m}}z_{
{l,m-1}}-{p_{{2}}}^{2}z_{{l,m-1}}x_{{l-1,m}}+p_{{1}}p_{{6}}z_{{
l-1,m}}x_{{l,m-1}})
$
\\
\hline
E21&
$( p_{{3}}\ x_{{l,m-1}}z_{{l-1,m}}+p_{{5}}\ z_{{l,m-1}}x_{{l-1,m}})
( p_{{5}}\ x_{{l,m-1}}z_{{l-1,m}}+p
_{{3}}\ z_{{l,m-1}}x_{{l-1,m}} )
$
\\
\hline
E22 &
$
( -z_{{l,m-1}}x_{{l-1,m}}+r_{{1}}z_{{l,m-1}}z_{{l-1,m}}+x_{{l,m-
1}}z_{{l-1,m}} )  ( -x_{{l,m-1}}z_{{l-1,m}}+z_{{l,m-1}}x_{{
l-1,m}}+r_{{1}}z_{{l,m-1}}z_{{l-1,m}} )
$
\\
\hline
E24&
$
-2\,p_{{2}}p_{{3}}\ z_{{l-1,m}}z_{{l,m-1}}x_{{l-1,m}}x_{{l,m-1}}+{p_{{2}
}}^{2}p_{{3}}{s}^{2}\ {z_{{l,m-1}}}^{2}{z_{{l-1,m}}}^{2}+p_{{2}}{p_{{3}}
}^{2}\ {z_{{l,m-1}}}^{2}{z_{{l-1,m}}}^{2}{s}^{2}+x_{{l-1,m}}z_{{l,m-1}}x
_{{l,m-1}}z_{{l-1,m}}-p_{{2}}p_{{3}}\ {z_{{l,m-1}}}^{2}{x_{{l-1,m}}}^{2}
-p_{{2}}p_{{3}}\ {z_{{l-1,m}}}^{2}{x_{{l,m-1}}}^{2}-{p_{{3}}}^{2}\ {x_{{l,
m-1}}}^{2}{z_{{l-1,m}}}^{2}-2\,{p_{{3}}}^{2}\ z_{{l-1,m}}z_{{l,m-1}}x_{{
l-1,m}}x_{{l,m-1}}-{p_{{2}}}^{2}\ z_{{l-1,m}}z_{{l,m-1}}x_{{l-1,m}}x_{{l
,m-1}}-{p_{{3}}}^{2}\ {z_{{l,m-1}}}^{2}{x_{{l-1,m}}}^{2}
$
\\
\hline
E25 &
$
( s\,z_{{l-1,m}}z_{{l,m-1}}-z_{{l,m-1}}x_{{l-1,m}}+z_{{l-1,m}}x_{{
l,m-1}})  ( {p_{{3}}}^{2}s\, z_{{l,m-1}}z_{{l-1,m}}+p_{{3}}s\ z_{{l-1,
m}}z_{{l,m-1}}-z_{{l-1,m}}z_{{l,m-1}}r_{{2}}+2\, r_{{4}}p_{{3}}\, z_{{l-1,m}}z_{{
l,m-1}}+r_{{4}}\, z_{{l-1,m}}z_{{l,m-1}}-p_{{3}}\, z_{{l-1,m}}
x_{{l,m-1}}-{p_{{3}}}^{2}\ z_{{l-1,m}}x_{{l,m-1}}+p_{{3}}\, z_{{l,m-1}}x_{{
l-1,m}}+{p_{{3}}}^{2}\,z_{{l,m-1}}x_{{l-1,m}} )
$
\\
\hline
\end{tabular}
\vskip0.5cm
\caption{List of common factors  $A_{l+1,m+1}$ of  $x_{l+1,m+1}$ and $z_{l+1,m+1}$ for the projective versions of equations above the KdV equation in Table \ref{tab:rules}  for {\em arbitrary} initial values
at the 5 corner sites $\{ (l-1,m-1), (l-1, m), (l-1, m+1), (l,m-1), (l+1,m-1)\}$ of the $2 \times 2$ lattice square.
\label{tab:commonfactor}}
\end{table}
\end{landscape}


This Proposition is important as it helps us to see how the new gcd appears at each vertex.
Assuming corner boundary values as polynomials in an indeterminate $w$, such as \eqref{ICI} and \eqref{ICII}, we can assume that at the point $(l+1,m+1)$:
\begin{itemize}
\item ${\gcd}_{l,m}(w) \; {\gcd}_{l+1,m}(w) \; {\gcd}_{l,m+1}(w) \; \big| \;{\gcd}_{l+1,m+1}(w)$, \\
\item $A_{l+1,m+1}(w)\ | \ \gcd_{l+1,m+1}(w)$
\end{itemize}
where all quantities belong to the ring of integer polynomials $\ZZ[w]$. The first statement is part 2 of Proposition \ref{T:gcd}.
From standard divisibility results, suppressing the $w$ dependence for brevity, we have :
\begin{equation} \label{Divisor}
\lcm (\rm{gcd}_{l,m}\, \rm{gcd}_{l+1,m}\, \rm{gcd}_{l,m+1},\ A_{l+1,m+1})=  \frac{\gcd_{l,m} \gcd_{l+1,m} \gcd_{l,m+1}\ A_{l+1,m+1}}{ \gcd (\gcd_{l,m} \gcd_{l+1,m} \gcd_{l,m+1},\ A_{l+1,m+1})} \; | \;  \rm{gcd}_{l+1,m+1}.
\end{equation}

Motivated by the form of the divisor on the rhs of \eqref{Divisor}, we define the
sequence $G_{l,m}$ by the recurrence:
\begin{equation}
\label{E:predictG}
G_{l+1,m+1}=\frac{G_{l,m}\ G_{l+1,m}\ G_{l,m+1}\ A_{l+1,m+1}}{ \gcd (G_{l,m}\ G_{l+1,m}\ G_{l,m+1},\ A_{l+1,m+1})}, \quad  l,m\geq 2.
\end{equation}
Suppose we take Case II boundary values of \eqref{ICII} for $l, m \le 9$
such that $\gcd_{l,m}=1$ if $l\leq 1$ or $m\leq 1$, and we define $G_{l,m}=\gcd_{l,m}=1$ if $l\leq 1$ or $m\leq 1$, so that $G_{l,m}$
shares the same boundary conditions. We then compare the
results of iterating the lattice equations of Table \ref{tab:commonfactor}  plus the non-autonomous
$Q_V$, to find ${\gcd}_{l+1,m+1}$ (which involves factoring $x_{l,m}$
and $z_{l,m}$ at every site) to $G_{l,m}$ calculated directly by iterating \eqref{E:predictG} (which avoids needing to factor at each site).
We find that, in fact,
\begin{equation}\label{eq:eureka}
G_{l,m}(w) =  {\gcd}_{l,m}(w) \cdot \, C_{l,m}, \quad  2 \le l,m \le 9,
\end{equation}
where $C_{l,m}$ is a lattice-dependent constant that depends on the equation for all the autonomous equations and non-autonomous $Q_V$.
For other non-autonomous equations, sometimes instead of having $G_{l,m}=\gcd_{l,m}(w)$ (up to a constant factor),
$G_{l,m}$ is actually a  non-trivial divisor of $\gcd_{l,m}$.

We now take a closer look at the rhs of \eqref{Divisor}. Recalling the definitions of barred variables from \eqref{eq:barxz}
and using the form of $A_{l+1,m+1}$ in Proposition \ref{prop:A}, we can write
\begin{equation} \label{eq:Asplit}
A_{l+1,m+1} = {\gcd}_{l-1,m}^2\, {\gcd}_{l,m-1}^2 \, \overline{A}_{l+1,m+1},
\end{equation}
where $\overline{A}_{l+1,m+1}$ is similar in form to ${A}_{l+1,m+1}$ but using the bar variables.
As a result of the $\gcd$ inheritance property from Proposition 2, we can write
\begin{equation}\label{eq:Bsplit}
\gcd ({\gcd}_{l,m} {\gcd}_{l+1,m} {\gcd}_{l,m+1},\, A_{l+1,m+1})
= {\gcd}_{l-1,m}^2\, {\gcd}_{l,m-1}^2 \, \overline{B}_{l+1,m+1}.
\end{equation}
Here $\overline{B}_{l+1,m+1}$ is a divisor of $\overline{A}_{l+1,m+1}.$
Therefore, using both \eqref{eq:Asplit} and \eqref{eq:Bsplit} in \eqref{Divisor}:
\begin{equation}\label{eq:motivateG1}
\frac{\gcd_{l,m} \gcd_{l+1,m} \gcd_{l,m+1}\ A_{l+1,m+1}}{ \gcd (\gcd_{l,m} \gcd_{l+1,m} \gcd_{l,m+1},\ A_{l+1,m+1})}
= \frac{\gcd_{l,m} \gcd_{l+1,m} \gcd_{l,m+1}\, \overline{A}_{l+1,m+1}}{\overline{B}_{l+1,m+1}} \; | \;  {\gcd}_{l+1,m+1},
\end{equation}
which via \eqref{bargcd2} implies:
\begin{equation} \label{eq:oneway}
\frac{\overline{A}_{l+1,m+1}}{\overline{B}_{l+1,m+1}} \mid \overline{\gcd}_{l+1,m+1}.
\end{equation}
Alternatively, we can also rewrite \eqref{Divisor} using just \eqref{eq:Bsplit} in the denominator and
replacing $\gcd_{l,m}$ in the numerator using (the downshift) of \eqref{bargcd2}:
\begin{equation}\label{eq:motivateG}
\frac{\gcd_{l,m} \gcd_{l+1,m} \gcd_{l,m+1}\ A_{l+1,m+1}}{ \gcd (\gcd_{l,m} \gcd_{l+1,m} \gcd_{l,m+1},\ A_{l+1,m+1})}
=
\frac{\gcd_{l-1,m-1}\, \gcd_{l+1,m}\, \gcd_{l,m+1}\, A_{l+1,m+1}}{{\gcd}_{l-1,m}\, {\gcd}_{l,m-1}} \; \;
\frac{\overline{\gcd}_{l,m}}{\overline{B}_{l+1,m+1}},
\end{equation}
and this too is a divisor of ${\gcd}_{l+1,m+1}$.

Motivated by the first term on the rhs  of \eqref{eq:motivateG}, we consider the following recurrence
on each $2\times 2$ lattice block:
\begin{equation}
\label{E:Gcd_recurence}
G^{'}_{l+1,m+1}=\frac{G^{'}_{l-1,m-1}\ G^{'}_{l+1,m}\ G^{'}_{l,m+1}\ A_{l+1,m+1}}{G^{'}_{l-1,m}\ G^{'}_{l,m-1}}.
\end{equation}
Extensive numerical investigations for
$1 \le l,m \le 9$ and a range of boundary conditions leads us to the following
\begin{conjecture}[{\bf Enabling Conjecture}]
\label{C:Integrable_Equations}
Suppose the values $G^{'}_{l,m}(w)$ for the recurrence \eqref{E:Gcd_recurence}
in the first quadrant are chosen to agree with ${\gcd}_{l,m}(w)$ when $0 \le l \le 1$ and $0 \le m \le 1$.
Then we conjecture for the lattice equations of Table \ref{tab:commonfactor} plus
(non-autonomous) $Q_V$ that this remains true on iteration of both \eqref{E:Gcd_recurence} and the lattice equation:
\begin{equation}\label{eq:eurekastrong}
G^{'}_{l,m}(w) =  {\gcd}_{l,m}(w) \cdot \, C_{l,m} \implies g^{'}_{l,m}:=\deg(G^{'}_{l,m}(w))=\deg({\gcd}_{l,m}(w))=g_{l,m}.
\end{equation}
where $C_{l,m}$ is a lattice-dependent constant.
From \eqref{eq:motivateG}, our conjecture is equivalent to saying that $\overline{B}_{l+1,m+1} = \overline{\gcd}_{l,m}$,
(up to a lattice constant).
\end{conjecture}

The conjecture means that calculating ${\gcd}_{l,m}(w)$ directly for the lattice rules can be replaced by using \eqref{E:Gcd_recurence}
and the only dependence on the particular lattice rule is through the factor $A_{l+1,m+1}(w)$ from Table \ref{tab:commonfactor}.
Evaluating $A_{l+1,m+1}(w)$ still necessitates iterating the lattice rule.  If, however, we are only
interested in the degree growth through the lattice, the uniformity of $\deg(A_{l+1,m+1}(w))$ across the different rules,
expressed by \eqref{eq:uniformdeg},
allows a statement for {\em all} lattice equations with the same form of $A_{l+1,m+1}(w)$.

Recall the reduced degrees $\overline{d}_{l,m}$ of \eqref{E:defbardspec} with \eqref{barg2}.
The conjecture leads to:
\begin{theorem} \label{Main}
For all the lattice equations given in Table \ref{tab:commonfactor} and (non-autonomous) $Q_V^{non}$,
the reduced degree sequence $\overline{d}_{l,m}$ of \eqref{E:defbardspec}  satisfies the
the following linear partial difference equation with constant coefficients:
\begin{equation} \label{E:degreeRelation}
\overline{d}_{l+1,m+1}=\overline{d}_{l+1,m}+\overline{d}_{l,m+1}+\overline{d}_{l-1,m-1}-\overline{d}_{l,m-1}-\overline{d}_{l-1,m},
\end{equation}
if conjecture \eqref{C:Integrable_Equations} is true.
If the corner boundary degree differences $(\overline{d}_{1,m_0+1}-\overline{d}_{0,m_0})$ and
$(\overline{d}_{l_0+1,1}-\overline{d}_{l_0,0})$ are bounded polynomially of degree $D$ in their respective coordinates $m_0$ and $l_0$, then
the sequence $\overline{d}_{l,m}$ of \eqref{E:degreeRelation}
is bounded polynomially of degree $D+1$ along each lattice diagonal.
Therefore $\lim_{l \to \infty} \log(\overline{d}_{l,m_0+l})/l = 0$ for fixed $m_0\ge 0$.
In particular, the corner boundary values of Case I (Case II) of \eqref{ICI}-\eqref{ICII}
produce linear (quadratic) growth for $\overline{d}_{l,m}$.
\end{theorem}
\begin{proof}
Taking degrees in  \eqref{E:Gcd_recurence},  $g^{'}_{l,m}:=\deg(G^{'}_{l,m}(w))$ satisfies the following recurrence for $l,m\geq 1$
\begin{equation}\label{eq:crucial}
g^{'}_{l+1,m+1}=2(d_{l,m-1}+d_{l-1,m})+g^{'}_{l-1,m-1}+g^{'}_{l+1,m}+g^{'}_{l,m+1}-g^{'}_{l-1,m}-g^{'}_{l,m-1},
\end{equation}
where the first term comes from \eqref{eq:uniformdeg}.
Adding this to  \eqref{E:samedeg} and its downshifted version gives
\eqref{E:degreeRelation}. Alternatively, conjecture \eqref{C:Integrable_Equations} gives \eqref{bargcddeg}
and writing $\overline{d}_{l+1,m+1}={d}_{l+1,m+1}-g_{l+1,m+1}$ and using \eqref{E:samedeg}
and \eqref{barg2} to replace ${d}_{l+1,m+1}$ and  $g_{l+1,m+1}$, and then downshifting the result to replace $\overline{d}_{l,m}$ will give
again \eqref{E:degreeRelation}.

Let $v_{l,m}:=\overline{d}_{l+1,m+1}-\overline{d}_{l,m}$ so that if, say, $m-l=m_0 \ge 0$,
\begin{equation} \label{eq:topdiag}
\overline{d}_{l,l+m_0}= \overline{d}_{0,m_0} + \sum_{j=0}^{l-1} v_{j,j+m_0},
\end{equation}
with an analogous expression for the growth of reduced degree along diagonals $l-m=l_0 > 0$
below the main diagonal.

Equation~\eqref{E:degreeRelation} gives
\begin{equation}
\label{E:JA_Eq}
v_{l,m}+v_{l-1,m-1}=v_{l-1,m}+v_{l,m-1}.
\end{equation}
which is equivalent to
$v_{l,m}-v_{l-1,m}=v_{l,m-1}-v_{l-1,m-1}$.
Iterating this downward in $m$ gives
\begin{equation} \label{eq:addit}
v_{l,m}-v_{l-1,m}=v_{l,0}-v_{l-1,0},
\end{equation}
which can be downshifted in $l$:
\begin{align*}
v_{l-1,m}-v_{l-2,m}&=v_{l-1,0}-v_{l-2,0},\\
v_{l-2,m}-v_{l-3,m}&=v_{l-2,0}-v_{l-3,0},\\
\vdots\qquad &  \qquad\quad \vdots\\
v_{1,m}-v_{0,m}&=v_{1,0}-v_{0,0}.
\end{align*}
Adding both sides from \eqref{eq:addit} onwards gives
\begin{equation} \label{vic}
v_{l,m}=v_{l,0}+v_{0,m}-v_{0,0},
\end{equation}
which expresses $v_{l,m}$ in terms of boundary values.
Substitution in \eqref{eq:topdiag} yields:
\begin{equation} \label{eq:topdiag2}
\overline{d}_{l,l+m_0}= \overline{d}_{0,m_0} - l\,v_{0,0} + \sum_{j=0}^{l-1} v_{j,0}+ v_{0,j+m_0},
\end{equation}
Suppose $|v_{j,0}| < P_1(j)$ and $|v_{0,j+m_0}| < P_2(j)$ where $P_i$ are polynomials with constant coefficients of respective
degrees $D_1, D_2 \ge 0$. Then by standard results, the sum $\sum_{j=0}^{l-1} j^D$ is a polynomial in $l-1$ of
degree $1+D$. Consequently,  $\sum_{j=0}^{l-1} |v_{j,0}+ v_{0,j+m_0}|$
is bounded by a polynomial in $l-1$ of degree $1+\max(D_1,D_2)$, and so is  $|\overline{d}_{l,l+m_0}|$.

For Case I boundary values, boundary values are constant on both axes.
Explicitly, we have $\overline{d}_{l,1}=\overline{d}_{1,m}=1$  and $\overline{d}_{l,0}=\overline{d}_{0,m}=0$
for all $l,m>0$ and $d_{0,0}=1$. It implies that $v_{0,0}=0$ and $v_{l,0}=v_{0,m}=1$ for all $l,m>0$. Hence
$P_1(j)$ and $P_2(j)$ are constant so $D_1=D_2=0$ and we have
linear growth along the diagonal. In fact, directly from \eqref{vic}, $v_{l,m}=2$ for all $l,m>0$ so the coefficient of linear growth is $2$.
This is what is observed in the right hand side of figure \ref{F:sponred1}.

For Case II boundary values, boundary values are affine on both axes in $w$ and one can deduce that
$v_{0,j+m_0})=\overline{d}_{1,j+m_0+1}-\overline{d}_{0,j+m_0}=2(j+m_0+1)$.  Hence $D_1=D_2=1$ and
we have quadratic growth in $l$ along the diagonal for $\overline{d}_{l,l+m_0}$. Consequently, the second difference of the latter is
necessarily constant and can be calculated. We have for $l,m >0$
    \begin{align*}
    \label{E:first_difference}
    \Delta_{l,m}&=v_{l+1,m+1}-v_{l,m}\\
    &=v_{l+1,0}-v_{l,0}+v_{0,m+1}-v_{0,m}=d_{l+2,0}+d_{l,0}+d_{0,m+2}+d_{0,m}\\
    &=4
    \end{align*}
This value for the second difference is observed in the right hand side of Figure \ref{F:sponred2}.
\end{proof}

The coefficient $2$ in \eqref{eq:crucial} and  \eqref{eq:uniformdeg}
proves crucial in the proof of Theorem \ref{Main}. It is a consequence of the fact that $A_{l+1,m+1}$ is quartic in its arguments.
It will be shown in section 5 that when instead $A_{l+1,m+1}$ is quadratic in its arguments, the growth of the reduced degree can be exponential.

We remark that some alternative degree equations equivalent to \eqref{E:degreeRelation} are given in \cite{TranRob}.  These equivalent forms prove useful for identifying and studying
linearisable lattice equations.
\subsection{Polynomial growth of other integrable lattice equations }

Consider the integrable equations given  in  Table \ref{tab:rules} that do not appear in Table \ref{tab:commonfactor}.
Although these equations have  different factorization patterns compared to the other integrable equations mentioned in the previous subsection,  we will show that Conjecture \ref{C:Integrable_Equations} and Theorem \ref{Main} still apply.

We start with the  KdV equation.  By using direct calculation, with arbitrary corner boundary values, we find that the first non-trivial gcd of $x_{l,m}$ and $z_{l,m}$  appears at the vertices $(2,3)$ and $(3,2)$ where
$\gcd_{2,3}=x_{1,1}$ and $\gcd_{3,2}=x_{1,1}$.  At the  vertex $(3,3)$,  the gcd is given as follows
\begin{equation} \label{eq:KdVquad}
{\gcd}_{3,3}=x_{1,2}^2\, x_{2,1}^2\, {\gcd}_{2,3}\,{\gcd}_{3,2}.
\end{equation}
It is noted that this formula is quite large if we write it in terms of
the boundary values.
Similarly, for mLV equation, we first see the non trivial  gcd of $x_{l,m}$ and $z_{l,m}$ at  $(2,3)$ and $(3,2)$. Using  Gr\"{o}bner bases,  one can find that
\[
{\gcd}_{3,3}=(x_{2,1}^2-c^2z_{2,1}^2)(x^2_{1,2}-c^2z^2_{2,1}) {\gcd}_{2,3}\,{\gcd}_{3,2}.
\]

For the MX equation, at the point $(2,2)$ the common factor is $\gcd_{2,2}=x_{1,0} \; x_{0,1}$, a quadratic rather than a quartic factor.
However at the vertex $(3,3)$ we  also obtain \eqref{eq:KdVquad}.
Since, the MX equation can be seen as a degeneration of the Tz equation \cite{Pavlos}, it suggests that Tz might behave similarly. By using Maple, we know that at the point $(2,2)$, the common factor is `quadratic' which is not enough for our Conjecture~\ref{C:Integrable_Equations}. Extending to the $3\times 3$ square, we find that
\[
{\gcd}_{3,3}=z_{1,2}\, z_{2,1} (cx_{2,1}z_{1,2}+cx_{1,2}z_{2,1}-x_{2,1}x_{1,2}-z_{2,1}z_{1,2})\, {\gcd}_{2,3}\,{\gcd}_{3,2}.
\]

Therefore, for these equations one can try to consider $A_{3,3}$ as a factor which plays the similar role as the $A$ given in the previous section. In general,
for $l,m\geq 2$, we have the following table of  $A_{l+1,m+1}$
 which is a common factor of $x_{l+1,m+1}$ and $z_{l+1,m+1}$.
 \begin{table}[h]
\begin{tabular}{|l |p{12.5cm}|}
\hline
KdV, MX & $ x_{l-1,m}^2 \ x_{l,m-1}^2$
\\
\hline
Tz&$z_{l,m-1}\, z_{l-1,m}\ (cx_{l,m-1}z_{l-1,m}+cx_{l-1,m}z_{l,m-1}-x_{l,m-1}x_{l-1,m}-z_{l,m-1}z_{l-1,m})$
\\
\hline
mLV&$(x_{l,m-1}^2-c^2z_{l,m-1}^2)\ (x^2_{l-1,m}-c^2z^2_{l,m-1})$
\\
\hline
\end{tabular}
\vskip0.5cm
\caption{List of common factors  $A_{l+1,m+1}$ of  $x_{l+1,m+1}$ and $z_{l+1,m+1}$ of the indicated rules of Table \ref{tab:rules} -- see also Figure \ref{CIF:Alm}.
\label{tab:other_inte_commonfactor}}
\end{table}

We recover the setup of the previous subsection by taking
 $l,m \ge 2$.
That is, we replicate the situation of Figure \ref{CIF:Alm}, except for the fact that the $2 \times 2$ squares over which the factorization
emerges starts one column in and one row up from the previous subsection.  Numerical experiments support that Conjecture \ref{C:Integrable_Equations} holds for $l,m\geq 2$ and the fact that
$A_{l+1,m+1}$ is again quartic here implies Theorem \ref{Main},
yielding polynomial growth for these  equations.

We will end this subsection by consider the LV equation. The first non trivial gcd of $x_{l,m}$ and $z_{l,m}$  appears at the vertex $(2,1)$ which is earlier than other integrable equations.
However, at the vertex $(2,2)$ the gcd is given as follows
\[
{\gcd}_{2,2}={z_{{1,0}}}^{2}z_{{0,0}} \left( \beta_{{1}}z_{{0,1}}+x_{{0,1}}
 \right),
\]
which does not seem to be enough for Conjecture \ref{C:Integrable_Equations}. Extending to the $3\times 3$  square, one can see that
\[  z_{{2,1}}z_{{1,2}} \left( \beta_{{1}}z_{{2,1}}+x_{{2,1}} \right) \left(
\beta_{{2}}z_{{1,2}}+x_{{1,2}} \right)|{\gcd}_{3,3}.
\]
This plays the similar role of $A_{l+1,m+1}$ in the previous subsection. In fact, numerical results show that Conjecture \ref{C:Integrable_Equations} holds for $l\geq 3, m\geq 2$ where
\[
A_{l+1,m+1}=z_{{l,m-1}} z_{{l-1,m}}\left( \beta_{{m-1}}z_{{l,m-1}}+x_{{l,m-1}} \right)  \left(
\beta_{{m}}z_{{l-1,m}}+x_{{l-1,m}} \right).
\]
This conjecture again helps us to prove polynomial growth of LV equation.

\section{Non integrable equations with factorization}
In this section, we consider the equations of Table \ref{tab:rules}
below the double-line division. These autonomous equations taken from \cite{HietVial} also  have some factorizations at  the top right corner
of any $2 \times 2$ lattice square, as in Figure \ref{CIF:Alm}, but the degree of the factor is insufficient to prevent exponential growth of the reduced degrees as found heuristically in~\cite{HietVial}.

\subsection{Equations with `quadratic' factorization }
We consider the equations E20, E26, E27, E28 and E30 of Table \ref{tab:rules} and present their factors $A_{l+1,m+1}$
in Table \ref{tab:non_commonfactor}.   Replicating the approach of section 4.2, we have again that \eqref{Divisor} holds and, as before, we find numerically that \eqref{E:predictG} reproduces ${\gcd}_{l,m}(w)$ up to a lattice-dependent constant, analogous to \eqref{eq:eureka}.
Because the $A_{l+1,m+1}$ of Table \ref{tab:non_commonfactor}   are quadratic in their arguments rather than quartic,  we have \eqref{eq:Asplit}  and \eqref{eq:Bsplit} but without the exponent $2$ on each $gcd$. Consequently, we were led to consider the following recurrence
\begin{equation}
\label{E:Non_inte_gcd}
G^{'}_{l+1,m+1}=\frac{A_{l+1,m+1} G^{'}_{l+1,m}G^{'}_{l,m+1}G^{'}_{l,m}}{G^{'}_{l,m-1}G^{'}_{l-1,m}},
\end{equation}
for $l,m\geq 2$ and $G^{'}_{l,m}=\gcd_{l,m}$ if $l<2$ or $m<2$.
Note that the rhs of \eqref{E:Non_inte_gcd} is based upon the divisor of \eqref{eq:motivateG}
after inserting \eqref{eq:Bsplit} without the exponent $2$ on each $gcd$ in the denominator
and excepting the factor $\overline{B}_{l+1,m+1}$.

\begin{table}[h]
\begin{tabular}{|l |p{12.5cm}|}
\hline
$E20$ & $p_6x_{l,m-1}z_{l-1,m}+p_3z_{l,m-1}x_{l-1,m}+p_3p_6r_4 z_{l,m-1}z_{l-1,m} $
\\
\hline
$E 26$&$p_1r_1z_{l,m-1}z_{l-1,m}+p_1p_6z_{l,m-1}x_{l-1,m}+p_3p_6x_{l,m-1}z_{l-1,m}$
\\
\hline
$E27$&${p_{{1}}}^{2}{p_{{6}}}^{2} \left( p_{{3}}-1 \right) ^{2} \left( p_{{3}}z_{{l,m-1}}x_{{l-
1,m}}+p_{{1}}p_{{6}}z_{{l-1,m}}x_{{l,m-1}}+r_{{4}}z_
{{l,m-1}}z_{{l-1,m}} \right)$
\\
\hline
$E28$&${p_
{{6}}}^{2}\left( -p_{{6}}x_{{l,m-1}}z_{{l-1,m}}+p_{{6}}r_{{3}}z_{{l-1,m}}z_{{l,
m-1}}+x_{{l-1,m}}z_{{l,m-1}}-p_{{3}}x_{{l-1,m}}z_{{l,m-1}} \right) $
\\
\hline
$E30$&$x_{{l-1,m}}z_{{l,m-1}}+x_{{l,m-1}}z_{{l-1,m}}+r_{{4}}z_{{l,m-1}}z_{{l
-1,m}}$
\\
\hline
\end{tabular}
\vskip0.5cm
\caption{List of common factors  $A_{l+1,m+1}$ of  $x_{l+1,m+1}$ and $z_{l+1,m+1}$ of the indicated rules of Table \ref{tab:rules} for {\em arbitrary} initial values
at the 5 corner sites $\{ (l-1,m-1), (l-1, m), (l-1, m+1), (l,m-1), (l+1,m-1)\}$ of the $2 \times 2$ lattice square -- see also Figure \ref{CIF:Alm}.
\label{tab:non_commonfactor}}
\end{table}



Numerical  experiments
for different trials of boundary values on $7\times 7$ lattice squares
lead us to
\begin{conjecture}[Enabling Conjecture]
\label{C:Non_inte}
Suppose the values $G^{'}_{l,m}(w)$ for the recurrence \eqref{E:Non_inte_gcd}
in the first quadrant are chosen to agree with ${\gcd}_{l,m}(w)$ when $0 \le l \le 1$ and $0 \le m \le 1$.
Then for $l, m \ge 2$, we find $\gcd_{l,m}(w)$ is a divisor of $G^{'}_{l,m}(w)$ given by the
recurrence~\eqref{E:Non_inte_gcd}.  For lattice rules E20, E26 and E27, we actually find $\gcd_{l,m}(w)$ agrees with  $G^{'}_{l,m}(w)$ up to a lattice-dependent constant.
\end{conjecture}

The conjecture leads to:
\begin{theorem} \label{MainNon}
If the enabling conjecture \ref{C:Non_inte} holds, the reduced degrees $\overline{d}_{l,m}=\deg(\overline{x}_{l,m}(w))=\deg(\overline{z}_{l,m}(w))$ of the lattice equations of Table \ref{tab:non_commonfactor} are bounded below by the sequence
$\overline{d}^{'}_{l,m}$ satisfying the following linear partial difference equation with constant coefficients:
 \begin{equation}
 \label{E:Non_actual_degree}
 \overline{d}^{'}_{l+1,m+1}=  \overline{d}^{'}_{l+1,m}+ \overline{d}^{'}_{l,m+1} + \overline{d}^{'}_{l,m} - \overline{d}^{'}_{l,m-1}- \overline{d}^{'}_{l-1,m}.
 \end{equation}
Consequently, the equations of Table \ref{tab:non_commonfactor} have non-vanishing entropy.
\end{theorem}
\begin{proof}
It follows directly from the formula~\eqref{E:Non_inte_gcd} that
the degree of $G^{'}_{l,m}$  satisfies the recurrence
 \begin{equation}
 \label{E:Non_inte_degree_gcd}
 g^{'}_{l+1,m+1}=(d_{l-1,m}+d_{l,m-1})+g^{'}_{l,m}+g^{'}_{l+1,m}+g^{'}_{l,m+1}-(g^{'}_{l,m-1}+g^{'}_{l-1,m}).
 \end{equation}
 where $ \overline{d}^{'}_{l,m}=d_{l,m}-g^{'}_{l,m}$.
Denote $v_{l,m}=\overline{d}^{'}_{l+1,m+1}-\overline{d}^{'}_{l,m}$. Using~\eqref{E:Non_actual_degree} we obtain for $l,m>0$
 \begin{equation}
 \label{E:Non_first_difference}
 v_{l,m}=v_{l,m-1}+v_{l-1,m},
 \end{equation}
 which  forms a Pascal's triangle.
 It implies that $v_{l,m}\geq v_{l,m-1} $ and $v_{l,m}\geq v_{l,m-1}$.
 Thus, we also  have $v_{l,m-1}\geq v_{l-1,m-1}$ and $v_{l-1,m}\geq v_{l-1,m-1}$.
 Therefore, we get $v_{l,m}\geq  2v_{l-1,m-1}$.
 In particular, along the diagonal we obtain $v_{l,l}\geq 2^{l}v_{0,0}$.
Since $v_{0,0}=\overline{d}^{'}_{1,1}-\overline{d}^{'}_{0,0}=d_{1,0}+d_{0,1}>0$ for non-constant initial values at $(0,1)$ and $(1,0)$.
Note that in the case of constant initial values, we have $v_{l,l}\geq 2^{l-1}v_{1,1}$ where $v_{1,1}>0$.
It shows that $v_{l,l}=\overline{d}^{'}_{l+1,l+1}-\overline{d}^{'}_{l,l}$ grows exponentially. Hence, $d^{'}_{l,l}$ grows exponentially.
\end{proof}

We note that in the case where initial values are linear, we  have found the following.
\begin{itemize}
\item We can write~\eqref{E:Non_actual_degree} as follows
\[
\overline{d}^{'}_{l+1,m+1}-\overline{d}^{'}_{l,m+1}-\overline{d}^{'}_{l+1,m}=\overline{d}^{'}_{l,m}-\overline{d}^{'}_{l-1,m}-\overline{d}^{'}_{l,m-1}.
\]
\item
It leads to  $\overline{d}^{'}_{l,m}=\overline{d}^{'}_{l-1,m}+\overline{d}^{'}_{l,m-1}+1$.
\item
  For $l,m>0$, we  have
\begin{align*}
v_{l,m}&=2\binom{l+m+2}{m+1}-2\binom{l+m}{m},\\
\overline{d}^{'}_{l,m}&=2\binom{l+m}{m}-1.
\end{align*}
This can be proved by induction.
\end{itemize}

\subsection{Equation $E17$}
Finally, we consider $E17$ given in Table \ref{tab:rules}. This equation has `bigger factorization' than the non-integrable equations
mentioned in the previous subsection. However, we will show that the factorization  is not big enough to allow this equation
to have  vanishing entropy.
It can be checked, with reference to Figure \ref{CIF:Alm}, that the common factor  of $x_{l+1,m+1}$ and $z_{l+1,m+1}$ is
\begin{equation}
\label{E:A17}
A17_{l+1,m+1}=\left( p_{{4}}x_{{l-1,m-1}}z_{{l,m-1}}+x_{{l,m-1}}p_{{6}}z_{{l-1,m-1}
}+r_{{4}}z_{{l-1,m-1}}z_{{l,m-1}} \right) {z_{{l-1,m}}}^{2}=z_{l-1,m}z_{l,m}.
\end{equation}
It is important to  note that factorization appears first at the point $(1,2)$, where $gcd_{1,2}=z_{0,1}$.
We know that $gcd_{l+1,m},gcd_{l,m+1},z_{l,m}$  are devisors of $gcd_{l+1,m+1}$.
We first used the test for $G_{l,m}$ as described in subsection~\ref{SS:poly_growth} and we found that $G_{l,m}$ is a divisor of the actual $gcd_{l,m}$. However, we found that
\begin{equation}
\label{E:E17_recur_bar}
\overline{d}_{l+1,m+1}=\overline{d}_{l+1,m}+\overline{d}_{l,m+1}.
\end{equation}
It implies that
$$
g_{l+1,m+1}=d_{l,m}+g_{l+1,m}+g_{l,m+1}.
$$
Therefore, we predict that
\begin{equation}
\label{E:E17gcd}
gcd_{l+1,m+1}=z_{l,m}gcd_{l+1,m}gcd_{l,m+1},
\end{equation}
up to a constant.
We have tested this recurrence and obtain the following conjecture.
\begin{conjecture}
\label{C:E17}
Given Case II boundary values,  let $G_{l,m}=gcd_{l,m}$  for $l<2$ or $m<2$ and let  $G_{l+1,m+1}$ be
defined by the following recursive formula
\begin{equation}
\label{E:E17_predict_gcd}
G_{l+1,m+1}=z_{l,m}G_{l+1,m}G_{l,m+1}.
\end{equation}
Then $G_{l,m}=gcd_{l,m}$ up to a constant factor.
\end{conjecture}
Using this conjecture, we have the following corollary
\begin{corollary}
\label{C:E17_expon_growth}
Given Case II boundary values, then equation $E17$ has exponential growth along diagonals, i.e. going from $(l,m)$ to $(l+1,m+1)$.
\end{corollary}
\begin{proof}
Using the recursive formula~\eqref{E:E17_recur_bar}, we have $\overline{d}_{l+1,m+1}\geq 2 \overline{d}_{l,m}$. This shows that $\overline{d}_{l,m}$ grows exponentially along the diagonal.
\end{proof}

We note that in the case where initial values are linear, for $l,m>0$ one can prove
\[
\overline{d}_{l,m}=\frac{2l+m}{l+m}\binom{l+m}{l}.
\]
The proof was done by using induction.
We note that the Lotka-Volterra equation which is given as follows
\begin{equation}
\label{E:Lk}
u_{{2}} \left( 1+u \right) -u_{{1}} \left( 1+u_{{12}} \right)=0
\end{equation}
is a special case of $E17$, however the Lotka-Volterra equation has the vanishing entropy.


\section{Application to staircase boundary values for lattice equations and to maps from periodic staircase reductions}

Thus far, we have considered our lattice equations with corner boundary values which, upon iteration, gives all
values of the field in the positive quadrant of the lattice.
An alternative and commonly-considered case is when the boundary values are considered on a staircase
that zig-zags through the vertices of $\ZZ^2$.  From this initial staircase, if well-posed, use of the lattice
rule allows values of the field variable to be calculated on a sequence of parallel disjoint staircases which cover
the lattice. Entropy of the lattice equation in this case has been addressed in \cite{Viallet} .
If we now consider a periodic staircase or a travelling-wave solution to the lattice equation, we
can reduce the integrable lattice equation, or partial difference equation, to an ordinary difference equation or map.
Entropy of birational maps has been considered in \cite{extended_Hietarinta, RGWM16, Kamp_growth, Viallet2015}, including those obtained from reduction of lattice maps.

The purpose of this section is to show how Theorem \ref{Main} and the degree relation \eqref{E:degreeRelation} lend
themselves to these other contexts as particular cases of our corner boundary value problem and reproduce degree growth results obtained previously in these contexts. This lends
further weight
to our Conjecture \ref{C:Integrable_Equations} on which Theorem \ref{Main} relies.

\subsection{Staircase boundary values}
The degree growth for staircase initial conditions was considered in \cite{Halburd09HowtoDetect,HydonViallet,Viallet}.
We consider firstly the simplest case of a staircase with unit horizontal
and vertical step sizes and we choose its orientation to be decreasing to the left and downwards -- see Figure \ref{F:staircase}.

\definecolor{uuuuuu}{rgb}{0.26666666666666666,0.26666666666666666,0.26666666666666666}
\definecolor{ffqqqq}{rgb}{1.,0.,0.}
\definecolor{wqwqwq}{rgb}{0.3764705882352941,0.3764705882352941,0.3764705882352941}
\definecolor{qqqqff}{rgb}{0.,0.,1.}
\begin{figure}
\centering
\begin{tikzpicture}[line cap=round,line join=round,>=triangle 45,x=1.5cm,y=1.5cm]
\clip(-1.7031885162350906,2.330198324914398) rectangle (7.962926807156245,8.774275207175307);
\draw [line width=1.2pt] (3.,4.) -- (3.,8.774275207175307);
\draw [line width=1.2pt,domain=3.0:7.962926807156245] plot(\x,{(--16.-0.*\x)/4.});
\draw [dash pattern=on 2pt off 2pt] (0.,8.)-- (3.,8.);
\draw [dash pattern=on 2pt off 2pt] (1.,7.)-- (3.,7.);
\draw [dash pattern=on 2pt off 2pt] (2.,6.)-- (3.,6.);
\draw [line width=1.2pt,color=qqqqff] (0.,8.)-- (0.,7.);
\draw [line width=1.2pt,color=qqqqff] (0.,7.)-- (1.,7.);
\draw [line width=1.2pt,color=qqqqff] (1.,7.)-- (1.,6.);
\draw [line width=1.2pt,color=qqqqff] (1.,6.)-- (2.,6.);
\draw [line width=1.2pt,color=qqqqff] (2.,6.)-- (2.,5.);
\draw [line width=1.2pt,color=qqqqff] (2.,5.)-- (3.,5.);
\draw [line width=1.2pt,color=qqqqff] (3.,5.)-- (3.,4.);
\draw [line width=1.2pt,color=qqqqff] (3.,4.)-- (4.,4.);
\draw [line width=1.2pt,color=qqqqff] (4.,4.)-- (4.,3.);
\draw [line width=1.2pt,color=qqqqff] (4.,3.)-- (5.,3.);
\draw (4.,4.) -- (4.,8.774275207175307);
\draw (5.,4.) -- (5.,8.774275207175307);
\draw (6.,4.) -- (6.,8.774275207175307);
\draw (7.,4.) -- (7.,8.774275207175307);
\draw [domain=3.0:7.962926807156245] plot(\x,{(--20.-0.*\x)/4.});
\draw [domain=3.0:7.962926807156245] plot(\x,{(--24.-0.*\x)/4.});
\draw [domain=3.0:7.962926807156245] plot(\x,{(--28.-0.*\x)/4.});
\draw [domain=3.0:7.962926807156245] plot(\x,{(--32.-0.*\x)/4.});
\draw [dash pattern=on 2pt off 2pt] (1.,7.) -- (1.,8.774275207175307);
\draw [dash pattern=on 2pt off 2pt] (2.,6.) -- (2.,8.774275207175307);
\draw [dash pattern=on 2pt off 2pt,domain=5.0:7.962926807156245] plot(\x,{(--6.-0.*\x)/2.});
\draw [dash pattern=on 2pt off 2pt] (0.,8.) -- (0.,8.774275207175307);
\draw [color=qqqqff] (0.,8.)-- (-0.7,8.);
\draw (1.925672960391623,8.253541721740083) node[anchor=north west] {\tiny $(7, 0)$};
\draw (2.9345940884223682,7.260893515129185) node[anchor=north west] {\tiny $(7, 0)$};
\draw (3.9272422950332633,6.251972387098436) node[anchor=north west] {\tiny $(7, 0)$};
\draw (4.936163423064008,5.259324180487539) node[anchor=north west] {\tiny $(7, 0)$};
\draw (5.928811629674903,4.25040305245679) node[anchor=north west] {\tiny $(7, 0)$};
\draw (6.937732757705649,3.2577548458458927) node[anchor=north west] {\tiny $(7, 0)$};
\draw (2.9345940884223682,8.253541721740083) node[anchor=north west] {\tiny $(17, 4)$};
\draw (3.9272422950332633,7.260893515129185) node[anchor=north west] {\tiny $(17, 4)$};
\draw (4.936163423064008,6.251972387098436) node[anchor=north west] {\tiny $(17, 4)$};
\draw (5.928811629674903,5.259324180487539) node[anchor=north west] {\tiny $(17, 4)$};
\draw (6.937732757705649,4.25040305245679) node[anchor=north west] {\tiny $(17, 4)$};
\draw (3.9272422950332633,8.253541721740083) node[anchor=north west] {\tiny $(41, 20)$};
\draw (4.936163423064008,7.260893515129185) node[anchor=north west] {\tiny $(41, 20)$};
\draw (5.928811629674903,6.251972387098436) node[anchor=north west] {\tiny $(41, 20)$};
\draw (6.937732757705649,5.259324180487539) node[anchor=north west] {\tiny $(41, 20)$};
\draw (4.936163423064008,8.253541721740083) node[anchor=north west] {\tiny $(99, 68)$};
\draw (5.928811629674903,7.260893515129185) node[anchor=north west] {\tiny $(99, 68)$};
\draw (6.937732757705649,6.251972387098436) node[anchor=north west] {\tiny $(99, 68)$};
\draw [color=ffqqqq,domain=-1.7031885162350906:7.962926807156245] plot(\x,{(--40.-4.*\x)/4.});
\draw [color=ffqqqq,domain=-1.7031885162350906:7.962926807156245] plot(\x,{(--16.-2.*\x)/2.});
\draw [color=ffqqqq,domain=-1.7031885162350906:7.962926807156245] plot(\x,{(--21.-3.*\x)/3.});
\draw [color=ffqqqq,domain=-1.7031885162350906:7.962926807156245] plot(\x,{(--36.-4.*\x)/4.});
\draw [dash pattern=on 2pt off 2pt] (6.,4.) -- (6.,2.330198324914398);
\draw [dash pattern=on 2pt off 2pt] (7.,4.) -- (7.,2.330198324914398);
\draw (-0.07589637425001727,8.253541721740083) node[anchor=north west] {\tiny $(1, 0)$};
\draw (0.9330247537807282,7.260893515129185) node[anchor=north west] {\tiny $(1, 0)$};
\draw (1.925672960391623,6.251972387098436) node[anchor=north west] {\tiny $(1, 0)$};
\draw (2.9345940884223682,5.259324180487539) node[anchor=north west] {\tiny $(1, 0)$};
\draw (3.9272422950332633,4.25040305245679) node[anchor=north west] {\tiny $(1, 0)$};
\draw (4.936163423064008,3.2577548458458927) node[anchor=north west] {\tiny $(1, 0)$};
\draw (0.9330247537807282,8.253541721740083) node[anchor=north west] {\tiny $(3, 0)$};
\draw (1.925672960391623,7.260893515129185) node[anchor=north west] {\tiny $(3, 0)$};
\draw (2.9345940884223682,6.251972387098436) node[anchor=north west] {\tiny $(3, 0)$};
\draw (3.9272422950332633,5.259324180487539) node[anchor=north west] {\tiny $(3, 0)$};
\draw (4.936163423064008,4.25040305245679) node[anchor=north west] {\tiny $(3, 0)$};
\draw (5.928811629674903,3.2577548458458927) node[anchor=north west] {\tiny $(3, 0)$};
\draw (-0.46644648832643487,7.000526772411573) node[anchor=north west] {\tiny $(1, 0)$};
\draw (0.5262017182844598,6.007878565800675) node[anchor=north west] {\tiny $(1, 0)$};
\draw (1.5351228463152053,4.998957437769926) node[anchor=north west] {\tiny $(1, 0)$};
\draw (2.5277710529261,4.006309231159029) node[anchor=north west] {\tiny $(1, 0)$};
\draw (3.5366921809568455,2.99738810312828) node[anchor=north west] {\tiny $(1, 0)$};
\draw [dash pattern=on 2pt off 2pt] (5.,4.)-- (5.,3.);
\draw [line width=1.2pt,color=qqqqff] (5.,3.) -- (5.,2.330198324914398);
\begin{scriptsize}
\draw [fill=black] (0.,8.) circle (1.0pt);
\draw [fill=qqqqff] (0.,7.) circle (1.0pt);
\draw [fill=qqqqff] (1.,7.) circle (1.0pt);
\draw [fill=qqqqff] (1.,6.) circle (1.0pt);
\draw [fill=qqqqff] (2.,6.) circle (1.0pt);
\draw [fill=qqqqff] (2.,5.) circle (1.0pt);
\draw [fill=qqqqff] (3.,5.) circle (1.0pt);
\draw [fill=black] (3.,4.) circle (2.5pt);
\draw [fill=qqqqff] (4.,4.) circle (1.0pt);
\draw [fill=qqqqff] (4.,3.) circle (1.0pt);
\draw [fill=qqqqff] (5.,3.) circle (1.0pt);
\draw [fill=black] (3.,8.) circle (1.0pt);
\draw [fill=wqwqwq] (7.,4.) circle (1.0pt);
\draw [fill=wqwqwq] (3.,7.) circle (1.0pt);
\draw [fill=wqwqwq] (3.,6.) circle (1.0pt);
\draw [fill=black] (7.,3.) circle (1.0pt);
\draw [fill=black] (7.,2.) circle (1.0pt);
\draw [fill=black] (1.,8.) circle (1.0pt);
\draw [fill=black] (2.,8.) circle (1.0pt);
\draw [fill=wqwqwq] (5.,4.) circle (1.0pt);
\draw [fill=wqwqwq] (6.,4.) circle (1.0pt);
\draw [fill=wqwqwq] (4.,8.) circle (1.0pt);
\draw [fill=wqwqwq] (5.,8.) circle (1.0pt);
\draw [fill=wqwqwq] (6.,8.) circle (1.0pt);
\draw [fill=wqwqwq] (7.,8.) circle (1.0pt);
\draw [fill=wqwqwq] (7.,5.) circle (1.0pt);
\draw [fill=wqwqwq] (7.,6.) circle (1.0pt);
\draw [fill=wqwqwq] (7.,7.) circle (1.0pt);
\draw [fill=wqwqwq] (4.,7.) circle (1.0pt);
\draw [fill=wqwqwq] (5.,7.) circle (1.0pt);
\draw [fill=wqwqwq] (6.,7.) circle (1.0pt);
\draw [fill=wqwqwq] (4.,6.) circle (1.0pt);
\draw [fill=wqwqwq] (5.,6.) circle (1.0pt);
\draw [fill=wqwqwq] (6.,6.) circle (1.0pt);
\draw [fill=wqwqwq] (4.,5.) circle (1.0pt);
\draw [fill=wqwqwq] (5.,5.) circle (1.0pt);
\draw [fill=wqwqwq] (6.,5.) circle (1.0pt);
\draw [fill=black] (2.,7.) circle (1.0pt);
\end{scriptsize}
\draw[color=black] (-1.4,8.65) node {$\overline{d}_0$};
\draw[color=black] (-0.4,8.65) node {$\overline{d}_1$};
\draw[color=black] (0.6,8.65) node {$\overline{d}_2$};
\draw[color=black] (1.6,8.65) node {$\overline{d}_3$};
\draw [fill=uuuuuu] (6.,3.) circle (1.pt);
\end{tikzpicture}
\caption{A staircase of boundary values for a lattice equation is shown in blue. For this staircase,
reduced degrees are the same on each diagonal of slope $-1$ and the first four are shown in red.
The boundary values on the staircase induce the corner boundary value problem highlighted by the black dot and black axes.
At each vertex, the first entry of the $2$-tuple gives the degree of $x$ and $z$ and the second entry the degree of their $\gcd$.
Hence the reduced degree at each vertex is the difference of the first and second entries.}\label{F:staircase}
\end{figure}

The methodology of the aforementioned
references is to ascribe affine polynomials of the same degree in $w$ for the $x$ and $z$ values at each point of the staircase, as we did with Case II corner boundary values,
with the generic assumption that the polynomials at each site are coprime. By symmetry, it is seen that on each diagonal of slope $-1$ parallel to the initial staircase,
we can assign a degree $d_i:=d_{l,-l+i}$, $i \ge 0$ using \eqref{E:samedeg}, which gives the degree of $x$ and $z$ along each diagonal.
We take $d_0=1$ for the initial staircase and see that $d_1=1$, $d_2=3$, $d_3=7$ etc.
Assuming our Conjecture \ref{C:Integrable_Equations} for the equations of Table 2, a nontrivial $\gcd$ only arises via the factor that appears over each
$2 \times 2$ lattice block of the iterated lattice rule and the degree of the $\gcd$ follows \eqref{eq:crucial}.
Symmetry again shows that we can define the reduced degrees $\overline{d}_i=\overline{d}_{l,-l+i}$, $i \ge 0$ for each diagonal with
$\overline{d}_i=d_i$ for $i=0,1,2,3$.
The first genuinely-reduced degree arises when $i=4$ since a $\gcd$ of degree $4$ arises as anticipated by \eqref{eq:crucial}, so
$$ \overline{d}_4=d_3-4=17-4=13.$$
To proceed, we use \eqref{E:degreeRelation} adapted to a corner created at one edge of the original staircase, as indicated in Figure \ref{F:staircase}.
From \eqref{E:degreeRelation} with our special case $\overline{d}_i=\overline{d}_{l,-l+i}$ reflecting the equality of degree along each diagonal
of slope $-1$,
we have the fourth order linear ordinary difference equation:
\begin{equation} \label{11staircase}
\overline{d}_{i+4}-2\, \overline{d}_{i+3} + 2\, \overline{d}_{i+1} - \overline{d}_{i} =0, \quad i \ge 0.
\end{equation}
The associated characteristic equation is
\begin{equation} \label{eq:linrecur1}
\lambda^4-2\,\lambda^3+2\,\lambda-1= (\lambda+1)\,(\lambda -1 )^3=0
\end{equation}
with general solution:
$$ \overline{d}_i= c_1 (-1)^i + c_2 + c_3 i + c_4 i^2,\quad i \ge 0.$$
The four initial conditions given above mean that we have
\begin{equation} \label{eq:linrecur2}
 \overline{d}_i = 1-i+i^2 =(1,1,3,7,13,21, 31, 43,57,73,91\ldots).
\end{equation}
This sequence with quadratic growth has been identified numerically for staircases like that of Figure \ref{F:staircase} for various equations including (non-autonomous) $Q_V$, $SG$
and others (\cite[Section 6]{Q5},
\cite[Table 1]{HydonViallet}, \cite[Section 8, 9]{Viallet} \cite{LGS}).
The method involves fitting in each case a rational univariate generating function $F(s)=\sum_{i=0}^{\infty} \overline{d}_i\, s^i$
to the first few computer-generated values of $\overline{d}_i$. Finding a rational generating function is equivalent to the sequence $(\overline{d}_i)$
satisfying a linear recurrence with integer coefficients.

Our working above shows that the recurrence \eqref{eq:linrecur1} and solution \eqref{eq:linrecur2}
are an automatic consequence for the lattice equations covered by Theorem \ref{Main} with boundary values given on the staircase of Figure \ref{F:staircase}.

The staircase of Figure \ref{F:staircase}  is built from the head-to-foot repetition of the finite generating staircase of one step to the right and one step down.
More generally, consider any staircase that is constructed from the head-to-foot repetition of
a finite staircase, which is assumed to be well-posed.
Without loss of generality, we assume the generator staircase starts with a horizontal step and ends with a vertical step and spans a horizontal width $q>1$ from left to right
on the lattice and a vertical height $|p|>1$ from bottom to top between its starting and ending points (i.e. $p$ positive (negative) means a rising (descending) staircase from left to right).
We call this a $(q,p)$-{\em repeating staircase} -- Figure \ref{F:staircase} has $(1,-1)$.
Viallet  \cite{Viallet} gives other simple examples, which repeat just one step of horizontal length $q$ and one of vertical height $p$.

Note the orientation of the corner is a clockwise rotation of $\pi/2$ to that we considered earlier and our south-easterly propagation
of parallel staircases entails solving for the bottom right vertex of our quad rather than the top right.  However, one checks that all the
equations of Table 2 are form-invariant under the transformations which correspond to the clockwise $\pi/2$-rotation, namely:
$$ u \mapsto u_1, u_1 \mapsto u_{12}, u_2 \mapsto u, u_{12} \mapsto u_2. $$

\subsection{Applications to mappings obtained as reductions of periodic staircases}
In this section we show that our conjectures still hold for mappings obtained reductions of lattice equations via the staircase method.
In particular, we consider the $(q,p)$ reductions where $\gcd(q,p)=1$ and $q,p>0$. We first present the $(q,p)$ reduction. Then, we give
consequences of conjectures~\eqref{C:Integrable_Equations}, \eqref{C:Non_inte} and \eqref{C:E17} for the corresponding maps.

Recall that the $(q,p)$ reduction of lattice equations gives us a $(p+q)$-dimensional map. The $(q,p)$ reduction is described as follows.
We introduce the travelling-wave ansatz:
\begin{equation} \label{periodstair}
u_{l,m}=:V_n \mbox{ where  }  n=l\,p-m\,q+1.
\end{equation}
This ansatz imposes the periodicity condition
\begin{equation} \label{periodfield}
u_{l,m}=u_{l+q,m+p}
\end{equation}
and, conversely, this periodicity condition implies that all
fields $u_{l,m}$ on the line $n=l\,p-m\,q+1$ are the same, provided $q$ and $p$ are coprime.
Hence field variables $V_n$ are now defined on each parallel
line of slope $p/q$.
We make a choice of origin $(l,m)=(0,0)$ on the lattice, at which we place $V_1$.  Since $\gcd(q,p)=1$, one can always find $0<l' \le q$ and $0\leq m' <p$ such that $l'p-m'q+1=n$
for each $n \in \{2,3,\ldots,p+q\}$. This set of points joined from left to right (increasing $n$ by $p$) and down to up (decreasing $n$ by $q$) yields the staircase of initial conditions.
The lattice equation \eqref{E1:Eq} specialises with \eqref{periodstair} to the ordinary difference equation for $V_n$:
\[
Q_n(V_n,V_{n+p},V_{n-q},V_{n+p-q},\alpha_l,\beta_m)=0,
\]
where for consistency we must impose periodicity on the parameters:
$$\alpha_l+q=\alpha_{l}, \beta_m+p=\beta_{m}.$$

By solving  $V_{n+p}$ from this equation -- which is certainly possible if $Q_{l,m}$ is multilinear in its arguments -- we obtain the map
$$F_n:(V_{n}, V_{n+1},\ldots V_{n+p+q-1})\mapsto (V_{n+1}, V_{n+2},\ldots V_{n+p+q}),$$
where the parameters are $\alpha_{l}, \beta_{m}$. For example, the $(3,2)$ reduction of a lattice equation is illustrated in Figure \ref{F:reduction_map}.
We observe there the staircase of initial conditions $V_1, V_2, V_3, V_4$ and the periodicity of it and later iterates of $V$.  The  ordinary difference equation arising from the lattice rule
is:
\[
Q_n(V_n,V_{n+2},V_{n-3},V_{n-1})=0.
\]
\definecolor{uuuuuu}{rgb}{0.26666666666666666,0.26666666666666666,0.26666666666666666}
\begin{figure}
\centering
\begin{tikzpicture}[line cap=round,line join=round,>=triangle 45,x=1.5cm,y=1.5cm]
\clip(-2.36,-0.3) rectangle (39.72,2.32);
\draw (0.,0.)-- (1.,0.);
\draw (1.,0.)-- (2.,0.);
\draw (2.,0.)-- (2.,1.);
\draw (3.,1.)-- (2.,1.);
\draw (3.,1.)-- (3.,2.);
\draw [color=blue][dash pattern=on 2pt off 2pt] (5.,0.)-- (4.,2.);
\draw [color=blue][dash pattern=on 2pt off 2pt] (3.,1.)-- (5.,0.);
\draw [color=blue][dash pattern=on 2pt off 2pt] (3.,1.)-- (4.,2.);
\draw [dash pattern=on 2pt off 2pt] (2.,0.)-- (3.,0.);
\draw [dash pattern=on 2pt off 2pt] (3.,0.)-- (4.,0.);
\draw [dash pattern=on 2pt off 2pt] (4.,0.)-- (5.,0.);
\draw [dash pattern=on 2pt off 2pt] (5.,0.)-- (6.,0.);
\draw (3.,2.)-- (5.,2.);
\draw [dash pattern=on 2pt off 2pt] (5.,2.)-- (6.,2.);
\draw [dash pattern=on 2pt off 2pt] (3.,1.)-- (3.,0.);
\draw [dash pattern=on 2pt off 2pt] (4.,2.)-- (4.,0.);
\draw [dash pattern=on 2pt off 2pt] (5.,2.)-- (5.,0.);
\draw [dash pattern=on 2pt off 2pt] (6.,2.)-- (6.,0.);
\draw [dash pattern=on 2pt off 2pt] (3.,1.)-- (6.,1.);
\begin{scriptsize}
\draw [fill=uuuuuu] (0.,0.) circle (1.0pt);
\draw[color=uuuuuu] (0.1,-0.16) node {$V_{1}$};
\draw [fill=black] (1.,0.) circle (1.0pt);
\draw[color=black] (1.08,-0.16) node {$V_{3}$};
\draw [fill=black] (2.,0.) circle (1.0pt);
\draw[color=black] (2.06,-0.16) node {$V_{5}$};
\draw [fill=black] (2.,1.) circle (1.0pt);
\draw[color=black] (2.06,1.14) node {$V_2$};
\draw [fill=black] (3.,1.) circle (1.0pt);
\draw[color=black] (3.18,1.14) node {$V_4$};
\draw [fill=black] (3.,2.) circle (1.0pt);
\draw[color=black] (3.08,2.14) node {$V_1$};
\draw [fill=black] (3.,0.) circle (1.0pt);
\draw[color=black] (3.08,-0.16) node {$V_{7}$};
\draw [fill=black] (4.,0.) circle (1.0pt);
\draw[color=black] (4.06,-0.16) node {$V_9$};
\draw [fill=black] (5.,0.) circle (1.0pt);
\draw[color=black] (5.06,-0.16) node {$V_{11}$};
\draw [fill=black] (6.,0.) circle (1.0pt);
\draw[color=black] (6.08,-0.16) node {$V_{13}$};
\draw [fill=black] (4.,2.) circle (1.0pt);
\draw[color=black] (4.04,2.14) node {$V_3$};
\draw [fill=black] (4.,1.) circle (1.0pt);
\draw[color=black] (4.18,1.14) node {$V_6$};
\draw [fill=black] (5.,2.) circle (1.0pt);
\draw[color=black] (5.06,2.14) node {$V_5$};
\draw [fill=black] (5.,1.) circle (1.0pt);
\draw[color=black] (5.16,1.14) node {$V_8$};
\draw [fill=black] (6.,2.) circle (1.0pt);
\draw[color=black] (6.1,2.14) node {$V_7$};
\draw [fill=black] (6.,1.) circle (1.0pt);
\draw[color=black] (6.22,1.14) node {$V_{10}$};
\end{scriptsize}
\end{tikzpicture}
\caption{The $(3,2)$ reduction}\label{F:reduction_map}
\end{figure}
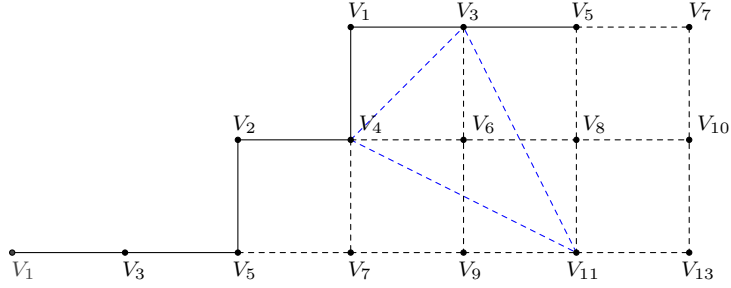

\subsection{Reductions of integrable lattice equations}
The partial difference equation \eqref{E:degreeRelation} for the reduced degrees that holds over a $2 \times 2$ lattice square
specialises via \eqref{periodstair} and \eqref{F:reduction_map} to the difference equation setting. Now
$$ \overline{d}_n:=\deg(numerator(V_n))= \deg(denominator(V_n)),
$$
and \eqref{E:degreeRelation} reduces to
the ordinary difference equation of order $2(p+q)$:
\begin{equation} \label{keydiffeq}
\bar{d}_{n+2(p+q)} = \bar{d}_n + \bar{d}_{n+p+2q} + \bar{d}_{n+2p+q} - \bar{d}_{n+q} - \bar{d}_{n+p}.
\end{equation}
The associated characteristic equation is
$$P(\lambda)=\lambda^{2(p+q)} - \lambda^{p+2q}- \lambda^{2p+q} + \lambda^q + \lambda^p - 1
= (\lambda^p - 1)\,(\lambda^q - 1)\,(\lambda^{p+q} - 1)=0. $$

\definecolor{ffqqqq}{rgb}{1.,0.,0.}
\definecolor{qqqqff}{rgb}{0.,0.,1.}
\begin{figure}[h]
\centering
\begin{tikzpicture}[line cap=round,line join=round,>=triangle 45,x=1.5cm,y=1.5cm]
\clip(-1.24,-0.42) rectangle (5.72,4.54);
\draw [dash pattern=on 4pt off 4pt] (0.,0.)-- (4.,0.);
\draw [dash pattern=on 4pt off 4pt] (4.,0.)-- (4.,4.);
\draw [dash pattern=on 4pt off 4pt] (0.,4.)-- (4.,4.);
\draw [dash pattern=on 4pt off 4pt] (0.,4.)-- (0.,0.);
\draw [color=qqqqff] (4.,0.)-- (0.,2.);
\draw [color=qqqqff] (0.,2.)-- (2.,4.);
\draw [color=qqqqff] (2.,4.)-- (4.,0.);
\draw [color=ffqqqq] (0.,4.)-- (2.,0.);
\draw [color=ffqqqq] (2.,0.)-- (4.,2.);
\draw [color=ffqqqq] (0.,4.)-- (4.,2.);
\draw [dash pattern=on 4pt off 4pt] (2.,4.)-- (2.,0.);
\draw [dash pattern=on 4pt off 4pt] (0.,2.)-- (4.,2.);
\begin{scriptsize}
\draw [fill=black] (0.,0.) circle (1.0pt);
\draw[color=black] (-0.24,-0.22) node {$\overline{d}_{n+2q}$};
\draw [fill=black] (0.,2.) circle (1.0pt);
\draw[color=black] (-0.44,2.06) node {$\overline{d}_{n+q}$};
\draw [fill=black] (0.,4.) circle (1.0pt);
\draw[color=black] (-0.24,4.04) node {$\overline{d}_n$};
\draw [fill=black] (2.,4.) circle (1.0pt);
\draw[color=black] (2.0,4.2) node {$\overline{d}_{n+p}$};
\draw [fill=black] (4.,4.) circle (1.0pt);
\draw[color=black] (4.5,4.) node {$\overline{d}_{n+2p}$};
\draw [fill=black] (4.,2.) circle (1.0pt);
\draw[color=black] (4.7,2.06) node {$\overline{d}_{n+2p+q}$};
\draw [fill=black] (4.,0.) circle (1.0pt);
\draw[color=black] (4.,-0.22) node {$\overline{d}_{n+2p+2q}$};
\draw [fill=black] (2.,0.) circle (1.0pt);
\draw[color=black] (2.,-0.22) node {$\overline{d}_{n+p+2q}$};
\draw [fill=black] (2.,2.) circle (1.0pt);
\draw[color=black,anchor=south east] (2,2) node {$\overline{d}_{n+p+q}$};
\end{scriptsize}
\end{tikzpicture}
\caption{The specialisation of the reduced degree formula \eqref{E:degreeRelation} for the $(q,p)$ reduction of a lattice equation.
The sum of the degrees on the blue triangle equals the sum on the red triangle.}\label{F:reduction_map}
\end{figure}
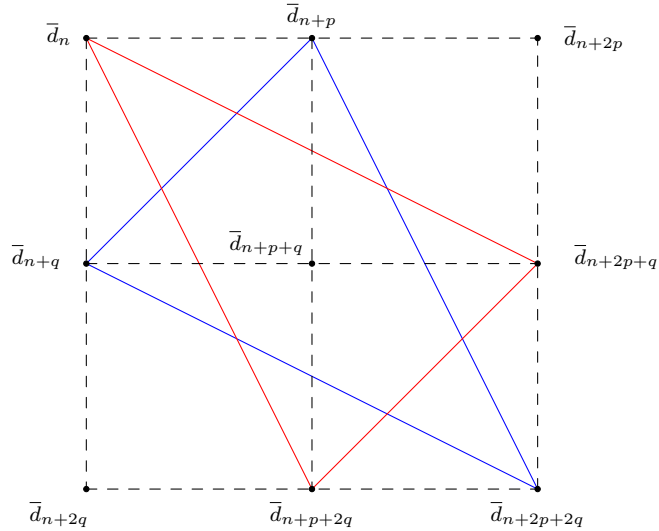

Now for $n \ge 1$
$$ x^n-1 = (x-1) \prod_{1\le k < n} (x - e^{i \frac{2 \pi k}{n}}).$$
By assumption, $p$ and $q$ are coprime so that precisely one of the set $\{q,p,q+p\}$ is even and
there is no divisor greater than one common within the set. Hence
\begin{equation} \label{cpoly}
P(\lambda)= (\lambda - 1)^3\, (\lambda+1)\, \prod_{d \ge 3 \mid \{q,p,q+p\}} \Phi_{d}(\lambda).
\end{equation}
Here $\Phi_i$ is the $i$th cyclotomic polynomial.
Because of the triple eigenvalue $1$ and single eigenvalue $-1$, the solution of \eqref{keydiffeq}  has a part that is generally quadratic in $n$ supplemented with a period-2 part:
$$  \overline{d}_i= c_1 (-1)^i + c_2 + c_3 i + c_4 i^2, i > 0.$$
If the last part of \eqref{cpoly} is present, its roots  come in complex conjugate pairs $e^{\pm\, i \frac{2 \pi k}{d}}$ with $k$ and $d$ coprime.
Each pair of roots leads to a term in the solution of \eqref{keydiffeq} of the form
$$ \overline{d}_i=C_i \cos(\frac{2 \pi i k}{d} + D_i), $$
where $C_i$ and $D_i$ are two arbitary constants determined by initial conditions.
This part of the solution is periodic: $ \overline{d}_{i+d}=\overline{d}_i$ with period $d\ge 3$.
In summary, the solution to \eqref{keydiffeq} will be linear or quadratic in $n$, supplemented by a periodic
part with period $\ge 2$ that is the $\lcm$ of $2$ together with all possible divisors $\ge 3$ of $\{q,p,q+p\}$.

The fact that the case $q=p=1$ in \eqref{keydiffeq} recovers \eqref{11staircase} for the $(1,-1)$ repeating staircase we discussed above
highlights something. The equation \eqref{keydiffeq} will be obtained whenever the  degrees
on a staircase are periodic:
\begin{equation} \label{perioddeg}
\overline{d}_{l,m}=\overline{d}_{l+q,m+p}  \implies {d}_{l,m}=d_n  \quad n:=l p - mq +1
\end{equation}
with $q$ and $p$ coprime.
This does not necessarily require the fields to be periodic on the staircase
as in the $(q,p)$ periodic reduction but certainly \eqref{periodfield} implies \eqref{perioddeg}.
A case in point is when a staircase is $(q,p)$ repeating as defined above and the degrees of the fields assigned on the staircase respect this repetition. In fact, saying that \eqref{perioddeg} holds on a staircase forces its geometry to be $(q,p)$ repeating as a byproduct.

So we have proved:

\begin{theorem} \label{reduction}
Suppose Conjecture \eqref{C:Integrable_Equations} is true
for all the lattice equations given in Table \ref{tab:commonfactor} and (non-autonomous) $Q_V$.
Consider a periodic staircase of boundary values on which the assigned degrees have the periodicity property \eqref{perioddeg}, which includes the case of (non-autonomous) maps obtained from the $(q,p)$ periodic reduction of the lattice equation.
The resulting degree $\overline{d}_n$ as a function of n is the sum of a polynomial part, at most quadratic in $n$,
and an additional (necessarily bounded) periodic part in $n$ .
\end{theorem}

Note the presence or absence of the periodic part can be inferred from taking the first or second difference of $\overline{d}_n$:
if the first (second) difference is constant, then $\overline{d}_n$ is linear (quadratic) in $n$ with no periodic part.
Otherwise, the first or second difference will be itself a periodic function of $n$.

We note that numerous numerical studies \cite{HietVial,HydonViallet, Kamp_growth, Viallet} of individual $(q,p)$ reductions of lattice equations relevant to Theorem~\ref{reduction}  (e.g. $ (q,p)=(1,1),\  (1,2),\  (1,3)$)
have exhibited the linear or quadratic degree growth, sometimes with an additional periodic part as evidenced by a non-constant second difference.
Our point here is that the truth of conjecture \eqref{C:Integrable_Equations} will guarantee this, uniformly in $q$ and $p$.


\subsection{Reductions of non-integrable lattice equations}
We also apply our results in section 5 to  $(q,p)$ reductions of non-integrable lattice equations.

 For equations with `quadratic' factorization (given in Table 3), the recurrence~\eqref{E:Non_actual_degree} becomes as follows
\begin{equation}
\label{E:non_map_degree}
\overline{d}_{n+2(p+q)} = \overline{d}_{n+p+q} + \overline{d}_{n+p+2q} + \overline{d}_{n+2p+q} - \overline{d}_{n+q} - \overline{d}_{n+p}.
\end{equation}
This give us the corresponding characteristic equation
\begin{equation}
\label{E:non_charac}
(\lambda^{p+q}-1)(\lambda^q-\lambda^{q-p}-1)=0.
\end{equation}

 For equation E17, the recursive formula~\eqref{E:E17_recur_bar} reduces to the  following equation
\begin{equation}
\label{E:E17_recur_bar_map}
\overline{d}_{n+2p+2q}=\overline{d}_{n+2p+q}+\overline{d}_{n+p+2q}.
\end{equation}
The characteristic equation of this linear recurrence is
\begin{equation}
\label{E:non_charac_E17}
\lambda^q-\lambda^{q-p}-1=0.
\end{equation}

 We note that for $q>p>0$, equation~\eqref{E:non_charac_E17} and hence
 equation~\eqref{E:non_charac} have at least one real root in $(1,2)$. This implies that the sequence $\overline{d}_i$ grows exponentially for both cases.

\section{Discussion}
We have presented two conjectures that have been used to obtain some results of
 growth of degrees of some integrable and non-integrable  lattice equations including a non-autonomous version of
$Q_V$. In particular, subject to a ocnjecture we have shown that a non-autonomous version of $Q_V$ which has been introduced recently in \cite{LGS} has  vanishing entropy.  Given an autonomous equation or a non-autonomous equation on the square, by looking at the factorization locally at the top right corner of
any $(2,2)$ square, one might be able to predict the integrability in the sense of having vanishing entropy. However, the question here is
`how big is that factor  in order to have vanishing entropy'. We have given some examples in which the algebraic entropy is vanishing. It seems that for these equations, the common factor needs to be `quartic' in terms of the off diagonal variables of the first square. In addition, this factor needs to have the divisible property, i.e. $A_{2,2}|A_{3,3}$.

Furthermore, we also gave some examples where the common factor is `quadratic' in terms of the off diagonal variables of the first square. These equations
turn out to be non-integrable as they have non-vanishing entropy. Therefore, it is worth studying this problem in the future.

In addition, there are some integrable equations where the factorization behaves in a more complicated way such as Tzitzeica and Lotka-Volterra equations. With these equations we need to go beyond the $2\times 2$ square. We have not been able to provide any recursive formula of the gcd for these equations. Thus, it might be interesting to investigate this problem further. Furthermore, there are  non-linear affine equations which are integrable in the sense of possessing a Lax pair \cite{Saha}. It would be
worth to study growth of degrees of these equations  as well.

\section*{Acknowledgements}
This work was supported by the Australian Research Council. We would like to thank  Joshua Capel for useful discussions.

\section*{Appendix}
In this Appendix, we give a factor $A_{l+1,m+1}$ for the autonomous and  non-autonomous $Q_V$ equations \cite{LGS, Q5} -- see Figure 4.  We distinguish $4$ cases  for the non-autonomous equation
$Q_V^{non}$ of \eqref{NonQV} based on the parity of $l$ and $m$. We denote $x_{l,m-1}=n_1, z_{l,m-1}=d_1, x_{l-1,m}=n_2$
and $d_{n-1,m}=d_2$.
The factor at the point $(2,2)$ of $Q_V$ is given by cf.  \cite{RobertsTran14,Q5, Viallet2015}
\begin{align*}
A_:&= \left( p_{{1,0}}p_{{6}}-p_{{2,0}}p_{{3,0}}-p_{{2,0}}p_{{4,0}}+p_{{2,0}}
p_{{5,0}} \right) {n_1}^{2}n_2d_2
+ \left( p_{{2,0}
}p_{{6,0}}-p_{{3,0}}p_{{4,0}} \right) n_1^2d_2^2
\\&
\quad
+\left( p_{{1,0}}p_{{5,0}}-{p_{{2,0}}}^{2} \right) n_1^2n_2^2
 +\left( p_{{1,0}}p_{{6,0}}-p_{{2,0}}p_{{3,0}}-p_{{2,0}}p_{{4,0}}+p_{{2,0}}p_{{5,0}}
 \right) n_1n_2^2d_1
 \\&
\quad
 + \left( p_{{1,0}}p_{{7,0}
}-{p_{{3,0}}}^{2}-{p_{{4,0}}}^{2}+{p_{{5,0}}}^{2} \right)n_1n_2d_1d_2
+ \left( p_{{2,0}}p_{{6,0}}-p_{{3}}p_{{4,0}} \right)n_2^2d_1^2
\\
&\quad
+
\left( p_{{2,0}}p_{{7,0}}-p_{{3,0}}p_{{6,0}}-p_{{4,0
}}p_{{6,0}}+p_{{5,0}}p_{{6,0}} \right) n_1d_1d_2^2+ \left( p_
{{5,0}}p_{{7,0}}-{p_{{6,0}}}^{2} \right)d_1^2d_2^2
\\
&
\quad
+ \left( p_{{2,0}}p_{{7,0}}-p_{{3,0}}p_{{6,0}}-p_{{4,0}}p_{{6,0}}+p_{
{5,0}}p_{{6,0}} \right) n_2d_1^2d_2.
\end{align*}
The factor at the top right corner of any $2\times 2$ square of the non-autonomous $Q_V^{non}$ is given as follows.
\begin{itemize}
\item Case 1 where $(l-1,m-1)=(0,0)\mod 2$, we have
\begin{align*}
A&=\left( p_{{4,0}}p_{{7,0}}-p_{{4,3}}p_{{7,0}}-{p_{{6,0}}}^{2}+2\,p_{{6
,0}}p_{{6,3}}+{p_{{6,1}}}^{2}-2\,p_{{6,1}}p_{{6,2}}+{p_{{6,2}}}^{2}-{p
_{{6,3}}}^{2} \right) {d_{{2}}}^{2}{d_{{1}}}^{2}
\\
&\quad
+ ( p_{{2,0}}p_{{7,0}}-p_{{2,1}}p_{{7,0}}-p_{{2,2}}p_{{7,0}}+p_{{2,3}}p_{{7,0}}-p_{{3,0
}}p_{{6,0}}-p_{{3,0}}p_{{6,1}}+p_{{3,0}}p_{{6,2}}+p_{{3,0}}p_{{6,3}}-p
_{{3,2}}p_{{6,0}}
\\
&\quad
-p_{{3,2}}p_{{6,1}}+p_{{3,2}}p_{{6,2}}+p_{{3,2}}p_{{6,3}}
+p_{{4,0}}p_{{6,0}}+p_{{4,0}}p_{{6,1}}+p_{{4,0}}p_{{6,2}}+p_{{4,0}
}p_{{6,3}}-p_{{4,3}}p_{{6,0}}
-p_{{4,3}}p_{{6,1}}
\\
&\quad
-p_{{4,3}}p_{{6,2}}
-p_
{{4,3}}p_{{6,3}}-p_{{5,0}}p_{{6,0}}+p_{{5,0}}p_{{6,1}}-p_{{5,0}}p_{{6,
2}}+p_{{5,0}}p_{{6,3}}
-p_{{5,1}}p_{{6,0}}+p_{{5,1}}p_{{6,1}}-p_{{5,1}}
p_{{6,2}}
\\
&\quad
+p_{{5,1}}p_{{6,3}} ) d_{{2}}{d_{{1}}}^{2}n_{{2}}
+
 ( p_{{2,0}}p_{{6,0}}+p_{{2,0}}p_{{6,1}}+p_{{2,0}}p_{{6,2}}+p_{{2
,0}}p_{{6,3}}-p_{{2,1}}p_{{6,0}}-p_{{2,1}}p_{{6,1}}-p_{{2,1}}p_{{6,2}}
\\
&
\quad
-p_{{2,1}}p_{{6,3}}-p_{{2,2}}p_{{6,0}}-p_{{2,2}}p_{{6,1}}
-p_{{2,2}}p_{
{6,2}}-p_{{2,2}}p_{{6,3}}+p_{{2,3}}p_{{6,0}}
+p_{{2,3}}p_{{6,1}}+p_{{2,
3}}p_{{6,2}}+p_{{2,3}}p_{{6,3}}
\\
&\quad
-p_{{3,0}}p_{{5,0}}-p_{{3,0}}p_{{5,1}}-
p_{{3,2}}p_{{5,0}}-p_{{3,2}}p_{{5,1}}) {d_{{1}}}^{2}{n_{{2}}}^{2}
+ ( p_{{2,0}}p_{{7,0}}+p_{{2,1}}p_{{7,0}}+p_{{2,2}}p_{{7,0}}+p_
{{2,3}}p_{{7,0}}
\\
&
\quad
-p_{{3,0}}p_{{6,0}}+p_{{3,0}}p_{{6,1}}-p_{{3,0}}p_{{6,
2}}+p_{{3,0}}p_{{6,3}}+p_{{3,2}}p_{{6,0}}-p_{{3,2}}p_{{6,1}}+p_{{3,2}}
p_{{6,2}}-p_{{3,2}}p_{{6,3}}+p_{{4,0}}p_{{6,0}}
\\
&
\quad
-p_{{4,0}}p_{{6,1}}
-p_{{4,0}}p_{{6,2}}+p_{{4,0}}p_{{6,3}}-p_{{4,3}}p_{{6,0}}+p_{{4,3}}p_{{6,1
}}+p_{{4,3}}p_{{6,2}}-p_{{4,3}}p_{{6,3}}-p_{{5,0}}p_{{6,0}}-p_{{5,0}}p
_{{6,1}}
\\
&
\quad
+p_{{5,0}}p_{{6,2}}+p_{{5,0}}p_{{6,3}}+p_{{5,1}}p_{{6,0}}
+p_{{5,1}}p_{{6,1}}-p_{{5,1}}p_{{6,2}}-p_{{5,1}}p_{{6,3}}) d_{{1}}{d
_{{2}}}^{2}n_{{1}}
\\
&
\quad
+ ( p_{{1,0}}p_{{7,0}}+4\,p_{{2,0}}p_{{6,3}}+4
\,p_{{2,1}}p_{{6,2}}+4\,p_{{2,2}}p_{{6,1}}+4\,p_{{2,3}}p_{{6,0}}-{p_{{
3,0}}}^{2}+{p_{{3,2}}}^{2}+{p_{{4,0}}}^{2}-{p_{{4,3}}}^{2}
\\
&
\quad
-{p_{{5,0}}}^{2}
+{p_{{5,1}}}^{2}) d_{{1}}d_{{2}}n_{{1}}n_{{2}}+ ( p_{{
1,0}}p_{{6,0}}+p_{{1,0}}p_{{6,1}}+p_{{1,0}}p_{{6,2}}+p_{{1,0}}p_{{6,3}
}-p_{{2,0}}p_{{3,0}}
-p_{{2,0}}p_{{3,2}}
\\
&
\quad
+p_{{2,0}}p_{{4,0}}
+p_{{2,0}}p_{{4,3}}-p_{{2,0}}p_{{5,0}}-p_{{2,0}}p_{{5,1}}+p_{{2,1}}p_{{3,0}}+p_{{2
,1}}p_{{3,2}}-p_{{2,1}}p_{{4,0}}-p_{{2,1}}p_{{4,3}}-p_{{2,1}}p_{{5,0}}
\\
&
\quad
-p_{{2,1}}p_{{5,1}}-p_{{2,2}}p_{{3,0}}
-p_{{2,2}}p_{{3,2}}-p_{{2,2}}p_{
{4,0}}-p_{{2,2}}p_{{4,3}}+p_{{2,2}}p_{{5,0}}+p_{{2,2}}p_{{5,1}}+p_{{2,
3}}p_{{3,0}}+p_{{2,3}}p_{{3,2}}
\\
&
\quad
+p_{{2,3}}p_{{4,0}}+p_{{2,3}}p_{{4,3}}+
p_{{2,3}}p_{{5,0}}+p_{{2,3}}p_{{5,1}} ) d_{{1}}n_{{1}}{n_{{2}}}^
{2}
\\
&
\quad
+ ( p_{{2,0}}p_{{6,0}}-p_{{2,0}}p_{{6,1}}-p_{{2,0}}p_{{6,2}}+p
_{{2,0}}p_{{6,3}}+p_{{2,1}}p_{{6,0}}-p_{{2,1}}p_{{6,1}}-p_{{2,1}}p_{{6
,2}}+p_{{2,1}}p_{{6,3}}+p_{{2,2}}p_{{6,0}}
\\
&
\quad
-p_{{2,2}}p_{{6,1}}-p_{{2,2}
}p_{{6,2}}+p_{{2,2}}p_{{6,3}}+p_{{2,3}}p_{{6,0}}-p_{{2,3}}p_{{6,1}}-p_
{{2,3}}p_{{6,2}}+p_{{2,3}}p_{{6,3}}-p_{{3,0}}p_{{5,0}}+p_{{3,0}}p_{{5,
1}}
\\
&
\quad
+p_{{3,2}}p_{{5,0}}-p_{{3,2}}p_{{5,1}}) {d_{{2}}}^{2}{n_{{1}
}}^{2}+( p_{{1,0}}p_{{6,0}}-p_{{1,0}}p_{{6,1}}-p_{{1,0}}p_{{6,2}
}+p_{{1,0}}p_{{6,3}}-p_{{2,0}}p_{{3,0}}
\\
&
\quad
+p_{{2,0}}p_{{3,2}}+p_{{2,0}}p_
{{4,0}}+p_{{2,0}}p_{{4,3}}-p_{{2,0}}p_{{5,0}}+p_{{2,0}}p_{{5,1}}-p_{{2
,1}}p_{{3,0}}+p_{{2,1}}p_{{3,2}}+p_{{2,1}}p_{{4,0}}+p_{{2,1}}p_{{4,3}}
\\
&
\quad
+p_{{2,1}}p_{{5,0}}-p_{{2,1}}p_{{5,1}}+p_{{2,2}}p_{{3,0}}-p_{{2,2}}p_{
{3,2}}+p_{{2,2}}p_{{4,0}}+p_{{2,2}}p_{{4,3}}-p_{{2,2}}p_{{5,0}}+p_{{2,
2}}p_{{5,1}}+p_{{2,3}}p_{{3,0}}
\\
&
\quad
-p_{{2,3}}p_{{3,2}}+p_{{2,3}}p_{{4,0}}+
p_{{2,3}}p_{{4,3}}+p_{{2,3}}p_{{5,0}}-p_{{2,3}}p_{{5,1}}) d_{{2
}}{n_{{1}}}^{2}n_{{2}}+
\\
&
\quad
\left( p_{{1,0}}p_{{4,0}}+p_{{1,0}}p_{{4,3}}-{
p_{{2,0}}}^{2}+2\,p_{{2,0}}p_{{2,3}}+{p_{{2,1}}}^{2}-2\,p_{{2,1}}p_{{2
,2}}+{p_{{2,2}}}^{2}-{p_{{2,3}}}^{2} \right) {n_{{1}}}^{2}{n_{{2}}}^{2
}.
\end{align*}
\item Case 2 where $(l-1,m-1)=(1,1)\mod 2$, we obtain
\begin{align*}
A&=(p_{{4,0}}p_{{7,0}}-p_{{4,3}}p_{{7,0}}-{p_{{6,0}}}^{2}+2\,p_{{6
,0}}p_{{6,3}}+{p_{{6,1}}}^{2}-2\,p_{{6,1}}p_{{6,2}}+{p_{{6,2}}}^{2}-{p
_{{6,3}}}^{2}) {d_{{2}}}^{2}{d_{{1}}}^{2}
\\
&
\quad
+
( p_{{2,0}}p_{{
7,0}}+p_{{2,1}}p_{{7,0}}+p_{{2,2}}p_{{7,0}}+p_{{2,3}}p_{{7,0}}-p_{{3,0
}}p_{{6,0}}+p_{{3,0}}p_{{6,1}}-p_{{3,0}}p_{{6,2}}+p_{{3,0}}p_{{6,3}}+p
_{{3,2}}p_{{6,0}}
\\
&
\quad
-p_{{3,2}}p_{{6,1}}
+p_{{3,2}}p_{{6,2}}-p_{{3,2}}p_{{6
,3}}+p_{{4,0}}p_{{6,0}}-p_{{4,0}}p_{{6,1}}-p_{{4,0}}p_{{6,2}}+p_{{4,0}
}p_{{6,3}}-p_{{4,3}}p_{{6,0}}+p_{{4,3}}p_{{6,1}}
\\
&
\quad
+p_{{4,3}}p_{{6,2}}-p_
{{4,3}}p_{{6,3}}-p_{{5,0}}p_{{6,0}}-p_{{5,0}}p_{{6,1}}+p_{{5,0}}p_{{6,
2}}+p_{{5,0}}p_{{6,3}}+p_{{5,1}}p_{{6,0}}+p_{{5,1}}p_{{6,1}}-p_{{5,1}}
p_{{6,2}}
\\
&
\quad
-p_{{5,1}}p_{{6,3}}) {d_{{1}}}^{2}d_{{2}}n_{{2}}+
 ( p_{{2,0}}p_{{6,0}}-p_{{2,0}}p_{{6,1}}-p_{{2,0}}p_{{6,2}}+p_{{2
,0}}p_{{6,3}}+p_{{2,1}}p_{{6,0}}
-p_{{2,1}}p_{{6,1}}-p_{{2,1}}p_{{6,2}}
\\
&
\quad
+p_{{2,1}}p_{{6,3}}+p_{{2,2}}p_{{6,0}}-p_{{2,2}}p_{{6,1}}-p_{{2,2}}p_{
{6,2}}+p_{{2,2}}p_{{6,3}}+p_{{2,3}}p_{{6,0}}-p_{{2,3}}p_{{6,1}}-p_{{2,
3}}p_{{6,2}}+p_{{2,3}}p_{{6,3}}
\\
&
\quad
-p_{{3,0}}p_{{5,0}}+p_{{3,0}}p_{{5,1}}+
p_{{3,2}}p_{{5,0}}-p_{{3,2}}p_{{5,1}} ) {d_{{1}}}^{2}{n_{{2}}}^{
2}+ ( p_{{2,0}}p_{{7,0}}-p_{{2,1}}p_{{7,0}}-p_{{2,2}}p_{{7,0}}
+p_{{2,3}}p_{{7,0}}
\\
&
\quad
-p_{{3,0}}p_{{6,0}}-p_{{3,0}}p_{{6,1}}+p_{{3,0}}p_{{6,
2}}+p_{{3,0}}p_{{6,3}}-p_{{3,2}}p_{{6,0}}-p_{{3,2}}p_{{6,1}}+p_{{3,2}}
p_{{6,2}}+p_{{3,2}}p_{{6,3}}+p_{{4,0}}p_{{6,0}}
\\
&
\quad
+p_{{4,0}}p_{{6,1}}+p_{
{4,0}}p_{{6,2}}+p_{{4,0}}p_{{6,3}}-p_{{4,3}}p_{{6,0}}-p_{{4,3}}p_{{6,1
}}-p_{{4,3}}p_{{6,2}}-p_{{4,3}}p_{{6,3}}-p_{{5,0}}p_{{6,0}}+p_{{5,0}}p
_{{6,1}}
\\
&
\quad
-p_{{5,0}}p_{{6,2}}+p_{{5,0}}p_{{6,3}}-p_{{5,1}}p_{{6,0}}+p_{{
5,1}}p_{{6,1}}-p_{{5,1}}p_{{6,2}}+p_{{5,1}}p_{{6,3}} ) d_{{1}}{d
_{{2}}}^{2}n_{{1}}+ ( p_{{1,0}}p_{{7,0}}+4\,p_{{2,0}}p_{{6,3}}
\\
&
\quad
+4\,p_{{2,1}}p_{{6,2}}+4\,p_{{2,2}}p_{{6,1}}+4\,p_{{2,3}}p_{{6,0}}-{p_{{
3,0}}}^{2}+{p_{{3,2}}}^{2}+{p_{{4,0}}}^{2}-{p_{{4,3}}}^{2}-{p_{{5,0}}}
^{2}+{p_{{5,1}}}^{2}
) d_{{1}}d_{{2}}n_{{1}}n_{{2}}
\\
&
\quad
+  p_{{
1,0}}p_{{6,0}}-p_{{1,0}}p_{{6,1}}-p_{{1,0}}p_{{6,2}}+p_{{1,0}}p_{{6,3}
}-p_{{2,0}}p_{{3,0}}+p_{{2,0}}p_{{3,2}}+p_{{2,0}}p_{{4,0}}+p_{{2,0}}p_
{{4,3}}-p_{{2,0}}p_{{5,0}}
\\
&
\quad
+p_{{2,0}}p_{{5,1}}-p_{{2,1}}p_{{3,0}}+p_{{2
,1}}p_{{3,2}}+p_{{2,1}}p_{{4,0}}+p_{{2,1}}p_{{4,3}}+p_{{2,1}}p_{{5,0}}
-p_{{2,1}}p_{{5,1}}+p_{{2,2}}p_{{3,0}}-p_{{2,2}}p_{{3,2}}
\\
&
\quad
+p_{{2,2}}p_{
{4,0}}
+p_{{2,2}}p_{{4,3}}-p_{{2,2}}p_{{5,0}}+p_{{2,2}}p_{{5,1}}+p_{{2,
3}}p_{{3,0}}-p_{{2,3}}p_{{3,2}}+p_{{2,3}}p_{{4,0}}+p_{{2,3}}p_{{4,3}}+
p_{{2,3}}p_{{5,0}}
\\
&
\quad
-p_{{2,3}}p_{{5,1}}) d_{{1}}n_{{1}}{n_{{2}}}^
{2}
+ ( p_{{2,0}}p_{{6,0}}+p_{{2,0}}p_{{6,1}}+p_{{2,0}}p_{{6,2}}+p
_{{2,0}}p_{{6,3}}-p_{{2,1}}p_{{6,0}}-p_{{2,1}}p_{{6,1}}-p_{{2,1}}p_{{6
,2}}
\\
&
\quad
-p_{{2,1}}p_{{6,3}}-p_{{2,2}}p_{{6,0}}
-p_{{2,2}}p_{{6,1}}-p_{{2,2}
}p_{{6,2}}-p_{{2,2}}p_{{6,3}}+p_{{2,3}}p_{{6,0}}+p_{{2,3}}p_{{6,1}}+p_
{{2,3}}p_{{6,2}}+p_{{2,3}}p_{{6,3}}
\\
&
\quad
-p_{{3,0}}p_{{5,0}}-p_{{3,0}}p_{{5,
1}}-p_{{3,2}}p_{{5,0}}-p_{{3,2}}p_{{5,1}}) {d_{{2}}}^{2}{n_{{1}
}}^{2}+ ( p_{{1,0}}p_{{6,0}}+p_{{1,0}}p_{{6,1}}+p_{{1,0}}p_{{6,2}
}+p_{{1,0}}p_{{6,3}}
\\
&
\quad
-p_{{2,0}}p_{{3,0}}
-p_{{2,0}}p_{{3,2}}+p_{{2,0}}p_
{{4,0}}+p_{{2,0}}p_{{4,3}}-p_{{2,0}}p_{{5,0}}-p_{{2,0}}p_{{5,1}}+p_{{2
,1}}p_{{3,0}}+p_{{2,1}}p_{{3,2}}-p_{{2,1}}p_{{4,0}}
\\
&
\quad
-p_{{2,1}}p_{{4,3}}
-p_{{2,1}}p_{{5,0}}-p_{{2,1}}p_{{5,1}}-p_{{2,2}}p_{{3,0}}-p_{{2,2}}p_{
{3,2}}-p_{{2,2}}p_{{4,0}}-p_{{2,2}}p_{{4,3}}+p_{{2,2}}p_{{5,0}}+p_{{2,
2}}p_{{5,1}}
\\
&
\quad
+p_{{2,3}}p_{{3,0}}
+p_{{2,3}}p_{{3,2}}+p_{{2,3}}p_{{4,0}}+
p_{{2,3}}p_{{4,3}}+p_{{2,3}}p_{{5,0}}+p_{{2,3}}p_{{5,1}}) d_{{2
}}{n_{{1}}}^{2}n_{{2}}
\\
&
\quad
+ \left( p_{{1,0}}p_{{4,0}}+p_{{1,0}}p_{{4,3}}-{
p_{{2,0}}}^{2}+2\,p_{{2,0}}p_{{2,3}}+{p_{{2,1}}}^{2}-2\,p_{{2,1}}p_{{2
,2}}+{p_{{2,2}}}^{2}-{p_{{2,3}}}^{2} \right) {n_{{1}}}^{2}{n_{{2}}}^{2
}.
\end{align*}
\item Case 3  where $(l-1,m-1)=(1,0)\mod 2$, we obtain
\begin{align*}
A&=( p_{{4,0}}p_{{7,0}}+p_{{4,3}}p_{{7,0}}-{p_{{6,0}}}^{2}-2\,p_{{6
,0}}p_{{6,3}}+{p_{{6,1}}}^{2}+2\,p_{{6,1}}p_{{6,2}}+{p_{{6,2}}}^{2}-{p
_{{6,3}}}^{2}) {d_{{1}}}^{2}{d_{{2}}}^{2}
\\
&\quad
+ ( p_{{2,0}}p_{{7,0}}+p_{{2,1}}p_{{7,0}}-p_{{2,2}}p_{{7,0}}-p_{{2,3}}p_{{7,0}}-p_{{3,0
}}p_{{6,0}}+p_{{3,0}}p_{{6,1}}+p_{{3,0}}p_{{6,2}}-p_{{3,0}}p_{{6,3}}-p
_{{3,2}}p_{{6,0}}
\\
&\quad
+p_{{3,2}}p_{{6,1}}+p_{{3,2}}p_{{6,2}}-p_{{3,2}}p_{{6
,3}}+p_{{4,0}}p_{{6,0}}
-p_{{4,0}}p_{{6,1}}+p_{{4,0}}p_{{6,2}}-p_{{4,0}
}p_{{6,3}}+p_{{4,3}}p_{{6,0}}-p_{{4,3}}p_{{6,1}}
\\
&\quad
+p_{{4,3}}p_{{6,2}}-p_
{{4,3}}p_{{6,3}}-p_{{5,0}}p_{{6,0}}-p_{{5,0}}p_{{6,1}}-p_{{5,0}}p_{{6,
2}}-p_{{5,0}}p_{{6,3}}+p_{{5,1}}p_{{6,0}}+p_{{5,1}}p_{{6,1}}+p_{{5,1}}
p_{{6,2}}
\\
&\quad
+p_{{5,1}}p_{{6,3}} {d_{{1}}}^{2}d_{{2}}n_{{2}}+
 ( p_{{2,0}}p_{{6,0}}-p_{{2,0}}p_{{6,1}}+p_{{2,0}}p_{{6,2}}-p_{{2
,0}}p_{{6,3}}+p_{{2,1}}p_{{6,0}}-p_{{2,1}}p_{{6,1}}+p_{{2,1}}p_{{6,2}}
\\
&\quad
-p_{{2,1}}p_{{6,3}}-p_{{2,2}}p_{{6,0}}+p_{{2,2}}p_{{6,1}}-p_{{2,2}}p_{
{6,2}}+p_{{2,2}}p_{{6,3}}-p_{{2,3}}p_{{6,0}}+p_{{2,3}}p_{{6,1}}-p_{{2,
3}}p_{{6,2}}+p_{{2,3}}p_{{6,3}}
\\
&\quad
-p_{{3,0}}p_{{5,0}}+p_{{3,0}}p_{{5,1}}-
p_{{3,2}}p_{{5,0}}+p_{{3,2}}p_{{5,1}}) {d_{{1}}}^{2}{n_{{2}}}^{
2}+
( p_{{2,0}}p_{{7,0}}-p_{{2,1}}p_{{7,0}}+p_{{2,2}}p_{{7,0}}-p_
{{2,3}}p_{{7,0}}
\\
&\quad
-p_{{3,0}}p_{{6,0}}-p_{{3,0}}p_{{6,1}}-p_{{3,0}}p_{{6,
2}}-p_{{3,0}}p_{{6,3}}+p_{{3,2}}p_{{6,0}}+p_{{3,2}}p_{{6,1}}+p_{{3,2}}
p_{{6,2}}+p_{{3,2}}p_{{6,3}}+p_{{4,0}}p_{{6,0}}
\\
&\quad
+p_{{4,0}}p_{{6,1}}-p_{
{4,0}}p_{{6,2}}-p_{{4,0}}p_{{6,3}}
+p_{{4,3}}p_{{6,0}}+p_{{4,3}}p_{{6,1
}}-p_{{4,3}}p_{{6,2}}-p_{{4,3}}p_{{6,3}}-p_{{5,0}}p_{{6,0}}+p_{{5,0}}p
_{{6,1}}
\\
&\quad
+p_{{5,0}}p_{{6,2}}-p_{{5,0}}p_{{6,3}}-p_{{5,1}}p_{{6,0}}+p_{{
5,1}}p_{{6,1}}+p_{{5,1}}p_{{6,2}}-p_{{5,1}}p_{{6,3}} ) d_{{1}}{d
_{{2}}}^{2}n_{{1}}
+ ( p_{{1,0}}p_{{7,0}}-4\,p_{{2,0}}p_{{6,3}}
\\
&\quad
-4\,p_{{2,1}}p_{{6,2}}-4\,p_{{2,2}}p_{{6,1}}-4\,p_{{2,3}}p_{{6,0}}-{p_{{
3,0}}}^{2}+{p_{{3,2}}}^{2}+{p_{{4,0}}}^{2}-{p_{{4,3}}}^{2}-{p_{{5,0}}}
^{2}+{p_{{5,1}}}^{2}) d_{{1}}d_{{2}}n_{{1}}n_{{2}}
\\
&\quad
+ ( p_{{
1,0}}p_{{6,0}}-p_{{1,0}}p_{{6,1}}+p_{{1,0}}p_{{6,2}}-p_{{1,0}}p_{{6,3}
}-p_{{2,0}}p_{{3,0}}-p_{{2,0}}p_{{3,2}}+p_{{2,0}}p_{{4,0}}-p_{{2,0}}p_
{{4,3}}-p_{{2,0}}p_{{5,0}}
\\
&\quad
+p_{{2,0}}p_{{5,1}}-p_{{2,1}}p_{{3,0}}-p_{{2
,1}}p_{{3,2}}+p_{{2,1}}p_{{4,0}}-p_{{2,1}}p_{{4,3}}+p_{{2,1}}p_{{5,0}}
-p_{{2,1}}p_{{5,1}}-p_{{2,2}}p_{{3,0}}-p_{{2,2}}p_{{3,2}}
\\
&\quad
-p_{{2,2}}p_{
{4,0}}+p_{{2,2}}p_{{4,3}}+p_{{2,2}}p_{{5,0}}-p_{{2,2}}p_{{5,1}}-p_{{2,
3}}p_{{3,0}}-p_{{2,3}}p_{{3,2}}-p_{{2,3}}p_{{4,0}}+p_{{2,3}}p_{{4,3}}-
p_{{2,3}}p_{{5,0}}
\\
&\quad
+p_{{2,3}}p_{{5,1}} ) d_{{1}}n_{{1}}{n_{{2}}}^
{2}+ ( p_{{2,0}}p_{{6,0}}+p_{{2,0}}p_{{6,1}}-p_{{2,0}}p_{{6,2}}-p
_{{2,0}}p_{{6,3}}-p_{{2,1}}p_{{6,0}}-p_{{2,1}}p_{{6,1}}+p_{{2,1}}p_{{6
,2}}
\\
&\quad
+p_{{2,1}}p_{{6,3}}+p_{{2,2}}p_{{6,0}}+p_{{2,2}}p_{{6,1}}-p_{{2,2}
}p_{{6,2}}-p_{{2,2}}p_{{6,3}}-p_{{2,3}}p_{{6,0}}-p_{{2,3}}p_{{6,1}}+p_
{{2,3}}p_{{6,2}}+p_{{2,3}}p_{{6,3}}
\\
&\quad
-p_{{3,0}}p_{{5,0}}-p_{{3,0}}p_{{5,
1}}+p_{{3,2}}p_{{5,0}}+p_{{3,2}}p_{{5,1}} ) {d_{{2}}}^{2}{n_{{1}
}}^{2}+ ( p_{{1,0}}p_{{6,0}}+p_{{1,0}}p_{{6,1}}-p_{{1,0}}p_{{6,2}
}-p_{{1,0}}p_{{6,3}}
\\
&\quad
-p_{{2,0}}p_{{3,0}}+p_{{2,0}}p_{{3,2}}+p_{{2,0}}p_
{{4,0}}-p_{{2,0}}p_{{4,3}}-p_{{2,0}}p_{{5,0}}-p_{{2,0}}p_{{5,1}}+p_{{2
,1}}p_{{3,0}}-p_{{2,1}}p_{{3,2}}-p_{{2,1}}p_{{4,0}}
\\
&\quad
+p_{{2,1}}p_{{4,3}}
-p_{{2,1}}p_{{5,0}}-p_{{2,1}}p_{{5,1}}+p_{{2,2}}p_{{3,0}}-p_{{2,2}}p_{
{3,2}}+p_{{2,2}}p_{{4,0}}-p_{{2,2}}p_{{4,3}}-p_{{2,2}}p_{{5,0}}-p_{{2,
2}}p_{{5,1}}
\\
&\quad
-p_{{2,3}}p_{{3,0}}+p_{{2,3}}p_{{3,2}}-p_{{2,3}}p_{{4,0}}+
p_{{2,3}}p_{{4,3}}-p_{{2,3}}p_{{5,0}}-p_{{2,3}}p_{{5,1}} ) d_{{2
}}{n_{{1}}}^{2}n_{{2}}
\\
&\quad
+ ( p_{{1,0}}p_{{4,0}}-p_{{1,0}}p_{{4,3}}-{
p_{{2,0}}}^{2}-2\,p_{{2,0}}p_{{2,3}}+{p_{{2,1}}}^{2}+2\,p_{{2,1}}p_{{2
,2}}+{p_{{2,2}}}^{2}-{p_{{2,3}}}^{2} ) {n_{{1}}}^{2}{n_{{2}}}^{2
}.
\end{align*}
\item Case 4 where $(l-1,m-1)=(0,1)\mod 2$, we obtain
\begin{align*}
A&= \left( p_{{4,0}}p_{{7,0}}+p_{{4,3}}p_{{7,0}}-{p_{{6,0}}}^{2}-2\,p_{{6
,0}}p_{{6,3}}+{p_{{6,1}}}^{2}+2\,p_{{6,1}}p_{{6,2}}+{p_{{6,2}}}^{2}-{p
_{{6,3}}}^{2} \right) {d_{{1}}}^{2}{d_{{2}}}^{2}
\\
&\quad
+ ( p_{{2,0}}p_{{
7,0}}-p_{{2,1}}p_{{7,0}}+p_{{2,2}}p_{{7,0}}-p_{{2,3}}p_{{7,0}}-p_{{3,0
}}p_{{6,0}}-p_{{3,0}}p_{{6,1}}-p_{{3,0}}p_{{6,2}}-p_{{3,0}}p_{{6,3}}+p
_{{3,2}}p_{{6,0}}
\\
&\quad
+p_{{3,2}}p_{{6,1}}+p_{{3,2}}p_{{6,2}}+p_{{3,2}}p_{{6
,3}}+p_{{4,0}}p_{{6,0}}+p_{{4,0}}p_{{6,1}}-p_{{4,0}}p_{{6,2}}-p_{{4,0}
}p_{{6,3}}+p_{{4,3}}p_{{6,0}}+p_{{4,3}}p_{{6,1}}
\\
&\quad
-p_{{4,3}}p_{{6,2}}-p_
{{4,3}}p_{{6,3}}-p_{{5,0}}p_{{6,0}}+p_{{5,0}}p_{{6,1}}+p_{{5,0}}p_{{6,
2}}-p_{{5,0}}p_{{6,3}}-p_{{5,1}}p_{{6,0}}+p_{{5,1}}p_{{6,1}}+p_{{5,1}}
p_{{6,2}}
\\
&\quad
-p_{{5,1}}p_{{6,3}} ) {d_{{1}}}^{2}d_{{2}}n_{{2}}+
 ( p_{{2,0}}p_{{6,0}}+p_{{2,0}}p_{{6,1}}-p_{{2,0}}p_{{6,2}}-p_{{2
,0}}p_{{6,3}}-p_{{2,1}}p_{{6,0}}-p_{{2,1}}p_{{6,1}}+p_{{2,1}}p_{{6,2}}
\\
&\quad
+p_{{2,1}}p_{{6,3}}+p_{{2,2}}p_{{6,0}}+p_{{2,2}}p_{{6,1}}-p_{{2,2}}p_{
{6,2}}-p_{{2,2}}p_{{6,3}}-p_{{2,3}}p_{{6,0}}-p_{{2,3}}p_{{6,1}}+p_{{2,
3}}p_{{6,2}}+p_{{2,3}}p_{{6,3}}
\\
&\quad
-p_{{3,0}}p_{{5,0}}-p_{{3,0}}p_{{5,1}}+
p_{{3,2}}p_{{5,0}}+p_{{3,2}}p_{{5,1}} ) {d_{{1}}}^{2}{n_{{2}}}^{
2}+ ( p_{{2,0}}p_{{7,0}}+p_{{2,1}}p_{{7,0}}-p_{{2,2}}p_{{7,0}}-p_
{{2,3}}p_{{7,0}}
\\
&\quad
-p_{{3,0}}p_{{6,0}}+p_{{3,0}}p_{{6,1}}+p_{{3,0}}p_{{6,
2}}-p_{{3,0}}p_{{6,3}}-p_{{3,2}}p_{{6,0}}+p_{{3,2}}p_{{6,1}}+p_{{3,2}}
p_{{6,2}}-p_{{3,2}}p_{{6,3}}+p_{{4,0}}p_{{6,0}}
\\
&\quad
-p_{{4,0}}p_{{6,1}}+p_{
{4,0}}p_{{6,2}}-p_{{4,0}}p_{{6,3}}+p_{{4,3}}p_{{6,0}}-p_{{4,3}}p_{{6,1
}}+p_{{4,3}}p_{{6,2}}-p_{{4,3}}p_{{6,3}}-p_{{5,0}}p_{{6,0}}-p_{{5,0}}p
_{{6,1}}
\\
&\quad
-p_{{5,0}}p_{{6,2}}-p_{{5,0}}p_{{6,3}}+p_{{5,1}}p_{{6,0}}+p_{{
5,1}}p_{{6,1}}+p_{{5,1}}p_{{6,2}}+p_{{5,1}}p_{{6,3}} ) d_{{1}}{d
_{{2}}}^{2}n_{{1}}+ ( p_{{1,0}}p_{{7,0}}-4\,p_{{2,0}}p_{{6,3}}
\\
&\quad
-4\,p_{{2,1}}p_{{6,2}}-4\,p_{{2,2}}p_{{6,1}}-4\,p_{{2,3}}p_{{6,0}}-{p_{{
3,0}}}^{2}+{p_{{3,2}}}^{2}+{p_{{4,0}}}^{2}-{p_{{4,3}}}^{2}-{p_{{5,0}}}
^{2}+{p_{{5,1}}}^{2}) d_{{1}}d_{{2}}n_{{1}}n_{{2}}
\\
&\quad
+ ( p_{{
1,0}}p_{{6,0}}+p_{{1,0}}p_{{6,1}}-p_{{1,0}}p_{{6,2}}-p_{{1,0}}p_{{6,3}
}-p_{{2,0}}p_{{3,0}}+p_{{2,0}}p_{{3,2}}+p_{{2,0}}p_{{4,0}}-p_{{2,0}}p_
{{4,3}}-p_{{2,0}}p_{{5,0}}
\\
&\quad
-p_{{2,0}}p_{{5,1}}+p_{{2,1}}p_{{3,0}}-p_{{2
,1}}p_{{3,2}}-p_{{2,1}}p_{{4,0}}+p_{{2,1}}p_{{4,3}}-p_{{2,1}}p_{{5,0}}
-p_{{2,1}}p_{{5,1}}+p_{{2,2}}p_{{3,0}}-p_{{2,2}}p_{{3,2}}
\\
&\quad
+p_{{2,2}}p_{
{4,0}}-p_{{2,2}}p_{{4,3}}-p_{{2,2}}p_{{5,0}}-p_{{2,2}}p_{{5,1}}-p_{{2,
3}}p_{{3,0}}+p_{{2,3}}p_{{3,2}}-p_{{2,3}}p_{{4,0}}+p_{{2,3}}p_{{4,3}}-
p_{{2,3}}p_{{5,0}}
\\
&\quad
-p_{{2,3}}p_{{5,1}}) d_{{1}}n_{{1}}{n_{{2}}}^
{2}+ ( p_{{2,0}}p_{{6,0}}-p_{{2,0}}p_{{6,1}}+p_{{2,0}}p_{{6,2}}-p
_{{2,0}}p_{{6,3}}+p_{{2,1}}p_{{6,0}}-p_{{2,1}}p_{{6,1}}+p_{{2,1}}p_{{6
,2}}
\\
&\quad
-p_{{2,1}}p_{{6,3}}-p_{{2,2}}p_{{6,0}}+p_{{2,2}}p_{{6,1}}-p_{{2,2}
}p_{{6,2}}+p_{{2,2}}p_{{6,3}}-p_{{2,3}}p_{{6,0}}+p_{{2,3}}p_{{6,1}}-p_
{{2,3}}p_{{6,2}}+p_{{2,3}}p_{{6,3}}
\\
&\quad
-p_{{3,0}}p_{{5,0}}+p_{{3,0}}p_{{5,
1}}-p_{{3,2}}p_{{5,0}}+p_{{3,2}}p_{{5,1}} ) {d_{{2}}}^{2}{n_{{1}
}}^{2}+ ( p_{{1,0}}p_{{6,0}}-p_{{1,0}}p_{{6,1}}+p_{{1,0}}p_{{6,2}
}-p_{{1,0}}p_{{6,3}}
\\
&\quad
-p_{{2,0}}p_{{3,0}}-p_{{2,0}}p_{{3,2}}+p_{{2,0}}p_
{{4,0}}-p_{{2,0}}p_{{4,3}}-p_{{2,0}}p_{{5,0}}+p_{{2,0}}p_{{5,1}}-p_{{2
,1}}p_{{3,0}}-p_{{2,1}}p_{{3,2}}+p_{{2,1}}p_{{4,0}}
\\
&\quad
-p_{{2,1}}p_{{4,3}}
+p_{{2,1}}p_{{5,0}}-p_{{2,1}}p_{{5,1}}-p_{{2,2}}p_{{3,0}}-p_{{2,2}}p_{
{3,2}}-p_{{2,2}}p_{{4,0}}+p_{{2,2}}p_{{4,3}}+p_{{2,2}}p_{{5,0}}-p_{{2,
2}}p_{{5,1}}
\\
&\quad
-p_{{2,3}}p_{{3,0}}-p_{{2,3}}p_{{3,2}}-p_{{2,3}}p_{{4,0}}+
p_{{2,3}}p_{{4,3}}-p_{{2,3}}p_{{5,0}}+p_{{2,3}}p_{{5,1}}) d_{{2
}}{n_{{1}}}^{2}n_{{2}}
\\
&\quad
+( p_{{1,0}}p_{{4,0}}-p_{{1,0}}p_{{4,3}}-{
p_{{2,0}}}^{2}-2\,p_{{2,0}}p_{{2,3}}+{p_{{2,1}}}^{2}+2\,p_{{2,1}}p_{{2
,2}}+{p_{{2,2}}}^{2}-{p_{{2,3}}}^{2}) {n_{{1}}}^{2}{n_{{2}}}^{2
}.
\end{align*}
\end{itemize}

\bibliographystyle{plain}

\begin{thebibliography}{10}

\bibitem{Adler}
V~E~Adler, \emph{On a discrete analog of the Tzitzeica equation} (2011) arXiv:1103.5139.

\bibitem{ABS}
V~Adler, A~Bobenko, and Y~Suris, \emph{Classification of integrable equations
  on quad-graphs. {T}he consistency approach}, Commun. Math. Phys. \textbf{3}
  (2003), 513--543.

\bibitem{BV}
M~P~Bellon and C-M Viallet,
\emph{Algebraic entropy}, Commun. Math. Phys. \textbf{204}
  (1999), 425--437.

\bibitem{dobrushkin}
V~A Dobrushkin, \emph{Methods in algorithmic analysis}, Chapman \& Hall/CRC
  Computer and Information Science Series, Chapman \& Hall/CRC, Boca Raton, FL,
  2010.

\bibitem{Halburd09HowtoDetect}
B~Grammaticos, R~G Halburd, A~Ramani, and C-M Viallet, \emph{How to detect the
  integrability of discrete systems}, J. Phys. A: Math. Theor. \textbf{42}
  (2009), 454002, 30pp.

\bibitem{non_ABS_GR}
B~Grammaticos and A~Ramani, \emph{Singularity confinement property for the
  (non-autonomous) {A}dler-{B}obenko-{S}uris integrable lattice equations},
  Lett. Math. Phys. \textbf{92} (2010), 33--45.

\bibitem{LGS}
G~Gubbiotti, S~Christian, and D~Levi, \emph{A non autonomous generalization of
  the {$Q_V$} equation} (2015) arXiv:1512.00395.

\bibitem{Halburd2005Diophantine-Int}
R~G Halburd, \emph{Diophantine integrability}, J. Phys. A: Math. Gen.
  \textbf{38} (2005), L263--L269.


\bibitem{HasselPropp}
B~Hasselblatt and J~Propp,
\emph{Degree-growth of monomial maps},
Ergod. Th. $\&$ Dynam. Sys. \textbf{27} (2007), 1375--1397.

\bibitem{HJN}
J~Hietarinta, N~Joshi and F~W~Nijhoff,
\emph{Discrete Systems and Integrability},
(CUP, Cambridge, 2016).

\bibitem{Hietarinta1998Singularity-Con}
J~Hietarinta and C~Viallet, \emph{Singularity confinement and chaos in discrete
  systems}, Phys. Rev. Lett. \textbf{81} (1998),  325--328.

\bibitem{HietVial}
J~Hietarinta and C~Viallet, \emph{Searching for integrable lattice maps using
  factorization}, J. Phys. A \textbf{40} (2007), 12629--12643.

\bibitem{Hirota1995Lotka-Voltera}
R~Hirota and S~Tsujimoto, \emph{Conserved quantities of a class of nonlinear
  difference-difference equations}, J. Phys. Soc. Japan \textbf{64} (1995),
  3125--3127.

\bibitem{HydonViallet}
P~E~Hydon and C~M~Viallet \emph{Asymmetric integrable quad-graph equations},
Applicable Analysis \textbf{89} (2010), 493--506

\bibitem{Kenji}
K~Kajiwara and Y~Ohta, \emph{Bilinearization and casorati determinant solution
  to the non-autonomous disrete kdv equation}, Journal of the Physical Society
  of Japan \textbf{77} (2008), 054004.

  \bibitem{coprimeness}
  M~Kanki, J~Mada, T~Mase, and T~Tokihiro, \emph{Irreducibility and co-primeness as an intergrability criterion for discrete equations} J. Phys. A: Math. Theor. \textbf{47} (2014), 465204, 15pp.

  \bibitem{extended_Hietarinta}
  M~Kanki, T~Mase, and T~Tokihiro, \emph{Algebraic entropy for an extended Hietarinta-Viallet equation} J. Phys. A: Math. Theor. \textbf{48} (2015), 355202, 19pp.

\bibitem{LeviYamilov2009}
D~Levi and R~I Yamilov, \emph{The generalized symmetry method for discrete
  equations}, J. Phys. A: Math. Theor. \textbf{42} (2009), 454012, 18pp.

  \bibitem{LeviYamilov2}
  D~Levi and R~I Yamilov,\emph{On a nonlinear integrable difference equation on the square 3D-inconsistent} (2009)
   arXiv: 0902.2126.

\bibitem{Pavlos}
A~V Mikhailov and P~Xenitidis, \emph{Second order integrability conditions for
  difference equations: An integrable equation}, Letters in Mathematical
  Physics (2013), 1--20.

\bibitem{Nijhoff1995KdV}
F~Nijhoff and H~Capel, \emph{The discrete {K}orteweg-de {V}ries equation}, Acta
  Appl. Math. \textbf{39} (1995), 133--158.

\bibitem{Ramaninotsonew}
A~Ramani, B~Grammaticos, J~Satsuma, and R~Willox, \emph{On two (not so) new
  integrable partial difference equations}, J. Phys. A: Math. Theor.
  \textbf{42} (2009), 282002, 6pp.

\bibitem{RGWM16}
A~Ramani, B~Grammaticos, R~Willox and T~Mase, \emph{Calculating algebraic entropies: an express method},  (2016)
arXiv:1611.05111v2.


\bibitem{Roberts_Tran_FF}
J~A~G Roberts and D~T Tran, \emph{Signatures over finite fields of growth
  properties for lattice equations}, J. Phys. A: Math. Theor. \textbf{48} (2015), 085201,
  20pp.

\bibitem{RobertsTran14}
J~A~G~Roberts  and D~T~Tran, \emph{Towards some exact results for the (vanishing) algebraic entropy of (integrable) lattice equations}, UNSW preprint (2014),
\url{web.maths.unsw.edu.au/~jagr/}.

\bibitem{SC}
R~Sahadevan and H~W Capel, \emph{Complete integrability and singularity
  confinement of nonautonomous modified {K}orteweg-de {V}ries and sine {G}ordon
  mappings}, Phys. A \textbf{330} (2003), 373--390.

\bibitem{Saha}
R~Sahadevan and G~Nagavigneshwari, \emph{New integrable and linearizable
  nonlinear difference equations}, Journal of Nonlinear Mathematical Physics
  \textbf{20} (2013),  179--190.

\bibitem{SRH}
R~Sahadevan, O~G Rasin, and P~E Hydon, \emph{Integrability conditions for
  nonautonomous quad-graph equations}, J. Math. Anal. Appl. \textbf{331}
  (2007), 712--726.

\bibitem{YingShi}
Y~Shi, D~J Zhang, and S~L Zhao, \emph{Solutions to the non-autonomous {ABS}
  lattice equations: Casoratians and bilinearization}, Scientia Sinica
  Mathematica \textbf{44} (2014), 37.

\bibitem{TranRob}
D~T~Tran and J~A~G~Roberts, \emph{Linear degree growth in lattice equations}, UNSW preprint (2017),
arXiv:1702.08295

\bibitem{Tremblay}
S~Tremblay, B~Grammaticos, and A~Ramani, \emph{Integrable lattice equations and
  their growth properties}, Phys. Lett. A \textbf{278} (2001), 319--324.

\bibitem{Kamp_growth}
P~H van~der Kamp, \emph{Growth of degrees of integrable mapping}, J. Difference
  Equ. Appl. \textbf{18} (2012), 447--460.

\bibitem{Viallet}
C~{Viallet}, \emph{{Algebraic Entropy for lattice equations}} (2006) arXiv:math-ph/0609043

\bibitem{Q5}
{C}~{M} Viallet, \emph{Integrable lattice maps: {$Q_V$}, a rational version of
  {$Q_4$}}, Glasg. Math. J. \textbf{51} (2009), 157--163.

\bibitem{Viallet2015}
C~M~Viallet, \emph{On the algebraic structure of rational discrete dynamical systems}
J. Phys. A: Math. Theor. \textbf{48} (2015), 16FT01, 21pp.

\end{thebibliography}

\end{document}